\documentclass[a4paper]{article}
\usepackage[margin=2cm]{geometry}
\usepackage[utf8]{inputenc}
\usepackage{amsmath}
\usepackage{amsfonts}
\usepackage{amssymb}
\usepackage{amsthm}
\newtheorem{theorem}{Theorem}
\newtheorem{lemma}{Lemma}
\newtheorem{definition}{Definition}
\usepackage{thmtools,thm-restate}
\usepackage[nottoc]{tocbibind}
\usepackage{enumerate}
\usepackage{hyperref}
\hypersetup{colorlinks,linkcolor=blue,citecolor=blue,urlcolor=blue} 
\usepackage{cleveref}
\crefname{lemma}{Lemma}{Lemmas}
\crefname{theorem}{Theorem}{Theorems}
\crefname{definition}{Definition}{Definitions}
\crefname{section}{Section}{Sections}
\crefname{appendix}{Appendix}{Appendices}
\usepackage{xcolor}
\definecolor{dgreen}{rgb}{.2,.6,.2}
\definecolor{dblue}{rgb}{.2,.2,.8}
\definecolor{dred}{rgb}{.8,.2,.2}
\newcommand{\ket}[1]{|{#1}\rangle}
\newcommand{\bra}[1]{\langle{#1}|}
\newcommand{\braket}[2]{\langle{#1}|{#2}\rangle}
\newcommand{\ketbra}[2]{\ket{#1}\bra{#2}}
\newcommand{\proj}[1]{\ketbra{#1}{#1}}
\usepackage{xspace}
\newcommand{\BU}{\ensuremath{\mathsf{BU}}\xspace}
\newcommand{\EUCMA}{\ensuremath{\mathsf{EU\text{-}CMA}}\xspace}
\newcommand{\world}[1]{\textsf{\upshape{#1}}\xspace}
\newcommand{\RealWorld}{\world{Real world}}
\newcommand{\QWorld}{\world{Quantum independent world}}
\newcommand{\IWorld}[1]{\world{Intermediate world~#1}}
\newcommand{\IWorlds}{\world{Intermediate worlds 1} and~\world{2}}
\usepackage{mathtools}
\DeclarePairedDelimiter{\set}{\lbrace}{\rbrace}
\DeclarePairedDelimiter{\abs}{\lvert}{\rvert}
\DeclarePairedDelimiter{\norm}{\lVert}{\rVert}
\DeclarePairedDelimiter{\floor}{\lfloor}{\rfloor}
\DeclarePairedDelimiter{\ceil}{\lceil}{\rceil}
\DeclarePairedDelimiter{\of}{\lparen}{\rparen}
\DeclarePairedDelimiter{\sof}{\lbrack}{\rbrack}

\DeclareMathOperator{\KeyGen}{\ensuremath{\mathsf{KeyGen}}\xspace}
\DeclareMathOperator{\Sign}{\ensuremath{\mathsf{Sign}}\xspace}
\DeclareMathOperator{\Ver}{\ensuremath{\mathsf{Ver}}\xspace}
\DeclareMathOperator{\BlindForge}{\ensuremath{\mathsf{BlindForge}}\xspace}
\usepackage[sc,osf]{mathpazo} %
\newcommand{\sk}{\textup{sk}}
\newcommand{\pk}{\textup{pk}}
\newcommand{\acc}{\textup{\texttt{acc}}}
\newcommand{\rej}{\textup{\texttt{rej}}}
\newcommand{\tp}{^{\mathsf{T}}} %
\newcommand{\ct}{^{\dagger}} %
\newcommand{\1}{\mathbb{1}} %
\newcommand{\x}{\otimes} %
\newcommand{\xp}[1]{^{\otimes #1}} %
\newcommand{\rgets}{\overset{\scriptscriptstyle\smash\$}{\gets}} %
\newcommand{\N}{\mathbb{N}} %
\newcommand{\R}{\mathbb{R}} %
\newcommand{\C}{\mathbb{C}} %
\newcommand{\A}{\mathcal{A}} %

\renewcommand{\H}{\mathcal{H}} %
\newcommand{\NOT}{\mathrm{NOT}}
\newcommand{\CNOT}{\mathrm{CNOT}}

\renewcommand{\P}{\Phi}

\usepackage{authblk}

\title{Quantum-access security of the\\Winternitz one-time signature scheme}

\newcommand{\email}[1]{\thanks{\href{mailto:#1}{#1}}}
\newcommand{\aff}[1]{\textit{\normalsize#1}}

\author[1]{Christian Majenz\email{christian.majenz@cwi.nl}}
\author[2]{Chanelle Matadah Manfouo\email{cmatadah@quantumleapafrica.org}}
\author[3]{Maris Ozols\email{marozols@gmail.com}}

\affil[1]{\aff{Centrum Wiskunde \& Informatica and QuSoft, The Netherlands}}
\affil[2]{\aff{African Institute for Mathematical Science \& Quantum Leap Africa, Rwanda}}
\affil[3]{\aff{Institute for Logic, Language, and Computation, Korteweg-de Vries Institute for Mathematics, and Institute for Theoretical Physics, University of Amsterdam and QuSoft, The Netherlands}}

\date{}

\begin{document}
\maketitle

\begin{abstract}
Quantum-access security, where an attacker is granted superposition access to secret-keyed functionalities, is a fundamental security model and its study has inspired results in post-quantum security. We revisit, and fill a gap in, the quantum-access security analysis of the Lamport one-time signature scheme (OTS) in the quantum random oracle model (QROM) by Alagic et al.~(Eurocrypt 2020). We then go on to generalize the technique to the Winternitz OTS. Along the way, we develop a tool for the analysis of hash chains in the QROM based on the superposition oracle technique by Zhandry (Crypto 2019) which might be of independent interest.
\end{abstract}

\tableofcontents

\section{Overview}
\subsection{Introduction}

Recently, research and development efforts towards building a universal quantum computer have intensified. As quantum computers will break currently deployed public-key cryptosystems \cite{shor1994algorithms}, finding adequate replacement schemes (called \emph{post-quantum} secure) has been increasingly a priority, too, as reflected by the ongoing NIST standardization effort for post-quantum secure digital signature schemes and key encapsulation mechanisms \cite{alagic2020status}. 

\paragraph{Quantum-access security.} While post-quantum security is the most important attack model involving quantum computers, the stronger \emph{quantum-access} or \emph{quantum world} attack model \cite{boneh2013secure,GHS16}, where attackers are granted quantum access to secret-keyed functionalities, has received considerable attention, too. There are a number of reasons why this stronger attack model is important. On the one hand, it is of theoretical importance because it captures the strongest-known achievable security notions for standard classical cryptographic primitives. On the other hand, there are a number of conceivable scenarios where they become relevant, e.g.~for composability with obfuscation or when constructing quantum-cryptographic schemes, or to prevent implementation-level vulnerabilities in a future hybrid quantum-classical computing infrastructure. Finally, results in the quantum access model can inform post-quantum cryptographic research, as exemplified by the offline Simon's algorithm attack \cite{BHNPSS19}.

\paragraph{Blind unforgeability.} 
In this work, we study the security of signature schemes under quantum-access attacks, in the quantum random oracle model (QROM) \cite{boneh2011random}. Here, generalizing the standard notion of existential unforgeability under chosen message attacks, the attacker is granted quantum query access to the signing algorithm. In the end, the adversary should output a forgery that they did not obtain from a query. Formalizing such a security notion is complicated due to the so-called \emph{quantum no-cloning principle} according to which quantum states cannot be copied. We use the notion of blind unforgeability introduced in \cite{alagic2020quantum} (see \cite{boneh2013secure,GYZ17} for previous and complementary notions). We remark that the choice of the blind unforgeability definition is due to the fact that it implies the previous notions, which are the Boneh and Zhandry definition \cite{boneh2013secure} and the one-time unforgeabilty~\cite{GYZ17}, as established in~\cite{alagic2020quantum}. Informally, blind unforgeability credits an adversary with a successful break of, e.g., a digital signature scheme, if it outputs a valid message-signature pair given a modified signing oracle that is ``blinded'' on a random subset of all messages in the sense that it outputs a dummy symbol instead of a signature, and if the output message is among these blinded messages (see section \cref{sec:BU} for details).

\paragraph{Hash-based signature schemes.}
Hash-based signature schemes are prominent candidates for the replacement of digital signature schemes based on quantum-broken number-theoretic hardness assumptions. In particular, the stateful hash-based signature scheme XMSS \cite{BDH11} has been standardized as RFC8391 \cite{RFC8391}, and the stateless hash-based signature scheme SPHINCS+ \cite{BHKNRS19} is an alternate candidate in the ongoing NIST standardization process for post-quantum cryptographic schemes \cite{alagic2020status}.  The security of hash-based signature schemes can be based on weak computational assumptions, like e.g.~the one-wayness of the underlying hash function. Common hash based signature schemes, including the mentioned examples, are constructed using one-time\footnote{And sometimes few-time signature schemes, e.g.~in SPHINCS+.} signature (OTS) schemes in combination with a hash-based authentication graph (e.g.~a Merkle tree). The most well-known OTSs are the Lamport \cite{Lam79} and Winternitz \cite{merkle1989certified} OTS. Variations of the latter are used in both XMSS and SPHINCS+.

\paragraph{Previous work.}
In \cite{alagic2020quantum}, the Lamport OTS is studied in the context of blind-unforgeability. More precisely, a proof of one-time blind-unforgeability in the QROM is provided. That proof, however, contains an imprecision in the analysis of the adversarial success. In particular, an auxiliary measurement is used to ``collapse'' an invariant property that holds \emph{in superposition} into holding \emph{classically}, but the effect of the dependence of this auxiliary measurement on the forgery message is not analyzed.

\paragraph{Related work.}
Quantum-access security for encryption is an active research area, and generalizing chosen-ciphertext security notions to the quantum access setting has posed, and poses, similar challenges as the ones encountered in the authenticity setting \cite{boneh2013secure,GHS16,GKS20}. On the negative side, key recovery attacks in the quantum access model against a number of symmetric-key primitives that are secure in the respective standard attack models have been discovered \cite{SS17,KLLNP16}, and have lead to the discovery of quantum attacks that can be performed without quantum access to secret-keyed functionalities \cite{BHNPSS19}.

There are a number of works that prove query lower bounds using variants of the superposition oracle technique \cite{LZ19,CFHL20,GHHM20}. In particular, the works \cite{CFHL20,blocki2020security} prove query bounds for generating hash chains in the QROM, considering \emph{parallel queries}. This analysis does not help when proving the security of the Winternitz OTS, as here the difficulty of \emph{inverting} an existing hash chain has to be exploited, and not the difficulty of \emph{generating} a hash chain.%

\subsection{Summary of results}\label{sec:summary}

\paragraph{The Lamport OTS is blind-unforgeable.}
We revisit the analysis of the Lamport OTS in the QROM presented in \cite{alagic2020quantum} and give a complete proof of blind unforgeability as stated in the following theorem.

\begin{theorem}[Blind unforgeability of the Lamport OTS, informal]
The Lamport OTS is blind-unforgeable if the underlying hash function $h$ is modeled as a quantum-accessible random oracle.
More precisely, the success probability of any blind unforgeability adversary $\A$ against the Lamport OTS that makes $q>0$ quantum queries to the random oracle is bounded as
\begin{equation*}
   \Pr[\textnormal{$\A$ succeeds}] \le C_L q^2l^3\cdot 2^{-n},
\end{equation*}
where $C_L$ is a constant, $n$ is the security parameter of the Lamport OTS and $l$ is the message length.
\end{theorem}

Compared to \cite{alagic2020quantum}, our security proof features the following improvements:
\begin{itemize}
    \item We streamline the usage of the superposition oracle technique of Zhandry \cite{Zhandry19}. In particular, our analysis only uses (a variant of) the superposition oracle technique to sample the secret key. We reprogram, \emph{in superposition}, the standard random oracle at inputs contained in the secret key. This technique represents a general tool to analyze hash chains in the QROM and might be of independent interest.
    \item We give a full analysis of the success probability using an auxiliary measurement idea from \cite{alagic2020quantum}. To tackle the problem mentioned above, we introduce a novel technique of tracking an invariant property \emph{in superposition} using projectors and commutators.
\end{itemize}

\paragraph{The Winternitz OTS is blind-unforgeable.}
With the full blind unforgeability analysis of the Lamport OTS in hand, we generalize the approach to the Winternitz OTS. 

\begin{theorem}[Blind unforgeability of the Winternitz OTS, informal]
The Winternitz OTS is blind-unforgeable if the underlying hash function $h$ is modeled as a quantum-accessible random oracle.
More precisely, the success probability of any blind unforgeability adversary $\A$ against the Winternitz OTS that makes $q>0$ quantum queries to the random oracle is bounded as 
\begin{equation*}
   \Pr[\textnormal{$\A$ succeeds}] \le C_W q^2 a^3\frac{w^4}{\log^3 w}\cdot 2^{-n},
\end{equation*}
where $C_W$ is a constant, $n$ is the security parameter of the Winternitz OTS, $a$ is the message length and $w\ge 2$ is the Winternitz parameter used to trade off signature size versus signing and verification time.
\end{theorem}

While the simplified analysis of hash chains in the QROM described above was advantageous in proving the blind unforgeability security of the Lamport OTS, it is indispensable in the analysis of the Winternitz scheme. Here, long hash chains are considered and the technique of using the superposition oracle to detect which hash chain elements are known to the adversary relies on the oracle register (or rather here: the hash chain register) being in a product state.

\subsection{Technical overview}

In this technical overview, we give a high-level description of our techniques for analyzing the blind unforgeability security of the Lamport and Winternitz OTSs in the QROM.

\paragraph{The superposition oracle technique and hash chains.}
As in many contexts that concern message authenticity and integrity, the main roadblock we have to overcome in our analysis is the so-called \emph{recording barrier}: quantum oracle queries can, in general, not be recorded for later use. In particular, after a single quantum signing query, it is not possible to reason about the unused parts of the secret key. This is because, in general, all secret key strings have been used in some part of the superposition.

In \cite{alagic2020quantum}, Zhandry's superposition oracle technique is used in a novel way to recover the ability to reason about which secret key strings are (un)known to the adversary. There, the secret key of the Lamport scheme, which is a $2 \times l$ array of independent uniformly random $n$-bit strings, is essentially regarded as a random function from $\{0,1\}\times \{1,\dotsc,l\}$. This function, as well as the hash function the Lamport OTS is constructed from, is then modelled using the superposition oracle technique.

We improve this technique as follows. We use the fact that sampling two correlated random variables $X$ and $Y$ can be done by first sampling $X$, and then $Y$ according to the conditional distribution, or vice versa. In the context of \emph{hash chains} in the (Q)ROM, i.e.~sequences of strings $x_0, x_1=H(x_0), x_2=H(x_1),\dotsc$ for a random oracle $H$, this means that instead of sampling $x_0$ and $H$, and then computing the remaining hash chain elements, we can as well sample $x_0, x_1, \dotsc$ from their joint distribution, sample $H$, and \emph{reprogram $H$ to be consistent with the $x_i$}. This allows us to i) change the distribution of the $x_i$ to a  simpler one that is close in total variational distance, and ii) refrain from using the full superposition oracle technique for $H$.
In particular, we use i) to replace the hash chains that are generated by the key generation algorithms of the Lamport and Winternitz schemes by tuples of independent random strings. This incurs only a small error, as the uniform distribution and the distribution of a hash chain in the (Q)ROM with random starting value $x_0$ are equal conditioned on all $x_i$ being distinct. But collisions between different hash chain elements are unlikely.

Now that the hash chain elements are  independent strings, we can use the full power of the superposition oracle technique. In particular, the one-to-one correspondence between the adversary's ignorance of a hash chain element, and the corresponding superposition oracle register being in uniform superposition, is restored.

Throughout the paper, and in the rest of this technical overview, we perform the analysis in a world where hash chains are formed using a superposition oracle modeling independent uniformly random strings, and the random oracle is reprogrammed accordingly. We call this the \QWorld. To conclude our analysis, we make use of the approximate indistinguishability of the \world{Real} and \QWorld.

\paragraph{Blind unforgeability and classical invariants in superposition.}
With the tools for analyzing hash chains in the QROM in hand, the next challenge consists of generalizing the classical security arguments for the Blind Unforgeability (\BU) of the Lamport and Winternitz OTSs to the quantum access setting. The core of these security arguments is, at a high level, that for each unqueried message, any valid signature contains a string that is unknown to the adversary.\footnote{When basing the security on one-wayness, ``unknown'' is to be taken in a computational sense, but as this paper is about security in the (Q)ROM, it is sufficient to interpret ``unknown to'' as ``independent of the state of''.} As mentioned above, this kind of reasoning does not generalize to the quantum access setting, as here an adversary can query all messages in superposition.

In the security game for the notion of \BU, however, the adversary is not provided with an oracle for the full signing algorithm functionality. Instead, the adversary is provided with an oracle for a modified signing algorithm that is ``blinded'' on a random subset of the messages, in the sense that for these messages it outputs a dummy symbol instead of a signature. These ``blinded messages'' can now replace the unqueried messages in security arguments, as by definition the adversary is prevented from obtaining a valid signature for them from the blinded signing oracle.

For obtaining a quantum generalization, we need to reformulate this argument. The statement that for each unqueried message any valid signature contains a string unknown to the adversary, is equivalent to saying that, for each fixed message $m^*$ and all $m \neq m^*$, some information related to the secret key is necessary to compute the signature for $m^*$ that is not revealed by the signature for $m$. For Blind Unforgeability, it suffices to consider blinded $m^*$ and unblinded $m$. In the superposition oracle framework, the statement ``there exists an unblinded message such that the registers corresponding to all parts of the secret key that the signature for that message does not reveal, are in the uniform superposition state'' defines a subspace $I$. By definition, the global state after a \BU-adversary makes a single query to the blinded signing oracle, and no queries to the random oracle, is in that subspace.

The crucial step in our analysis is to show that the adversary-oracle state approximately remains in the subspace $I$, even if the adversary performs a moderate number of quantum queries to the random oracle. This means the subspace $I$ can serve as an \emph{invariant}.

\paragraph{Random oracle queries and commutators.}
To analyze the ``leakage'' from the invariant subspace $I$, we use bounds on the norm of matrix commutators: to prove that the final oracle-adversary state is approximately in the invariant subspace $I$, we can equivalently show that applying the corresponding projector $\Pi_I$ does not change the state by a lot. We know, however, that the projector does not change the state at all before any random oracle queries have been made. Therefore it suffices to bound the operator norm of the commutator between the projector $\Pi_I$ and the unitary operator that facilitates random oracle queries in the \QWorld. We derive such a norm bound (see e.g.~\cref{commutator hash unitary and P_S} for the Lamport case), and the proof follows the classical intuition about the one-wayness of the random oracle.

\section{Preliminaries}

Let us introduce some notation that will be used throughout the paper. In this document, quantum systems are associated with finite-dimensional complex Euclidean vector spaces endowed with an inner product. Registers will be denoted by capital letters.
We say that $\epsilon = \epsilon(n)$ is negligible if, for all polynomials $p(n)$, $\epsilon(n) < 1/p(n)$ for large enough $n \in \N$.
We use the notation $x \rgets D$ to say that $x$ is chosen uniformly at random from a set $D$. 
We write $S^c$ to denote the complement of set $S$ (in a superset that is clear from the context).
We write $s \parallel t$ to denote the \emph{concatenation} of strings $s$ and $t$, and $[A,B] = AB - BA$ to denote the \emph{commutator} of operators $A$ and $B$.
Throughout this paper, quantum adversaries refer to quantum poly-time algorithms and are denoted by $\A$.

\subsection{Quantum computing}

In this section, we introduce some basic notions from quantum computing. We refer the reader to \cite{nielsen2002quantum} for more details.

\subparagraph*{Quantum states.}
To a given quantum system with $d$ distinguished states we associate a $d$-dimensional complex Euclidean vector space $\H = \C^d$ endowed with an inner product $\bra\cdot\cdot\rangle$. We refer to the standard basis of $\C^d$ as the computational basis. The \emph{state} of such system is described by a unit vector, i.e., a vector $\ket \psi \in \H$ such that $\bra\psi\psi\rangle = 1$. For example, a \emph{qubit} state is described by a vector $\ket{\psi} = \alpha \ket{0} + \beta \ket{1} \in \C^2$ such that $|\alpha|^2 + |\beta|^2 = 1$, where $\ket{0}$ and $\ket{1}$ are the \emph{computational basis} vectors.
The corresponding \emph{dual vector} is given by $\bra{\psi} = \bar{\alpha} \bra{0} + \bar{\beta} \bra{1}$, which can also be expressed through entry-wise complex conjugation and transpose: $\bra{\psi} = \ket{\psi}\ct = \overline{\ket{\psi}}{}\tp$.

Given two quantum systems $A$ and $B$ with state spaces $\H_A$ and $\H_B$, the composite system $AB$ has state space $\H_A \x \H_B$ described by the tensor product. In particular, if $\ket\psi_A\in\H_A$ and $\ket\psi_B\in\H_B$ are states of the two individual systems, then their joint state is given by $\ket\psi_A \otimes \ket\psi_B$ or simply $\ket\psi_A \ket\psi_B$. We will often refer to the subsystems $A$ and $B$ as \emph{registers}.
For example, an $n$-qubit system consists of $n$ qubit registers and its computational basis is given by $\ket{x} = \ket{x_1} \dotsb \ket{x_n}$ where $x = x_1 \dots x_n \in \set{0,1}^n$. A general $n$-qubit state is then a linear combination of the computational basis states:
\begin{equation*}
    \sum_{x \in \set{0,1}^n} \alpha_x \ket{x}
    \quad \text{with} \quad
    \sum_{x \in \set{0,1}^n} \abs{\alpha_x}^2 = 1.
\end{equation*}
In particular, when all $\alpha_x$ are equal to $2^{-n/2}$, we call this the \emph{uniform superposition}.
Throughout this paper, we will denote this state and the corresponding projector by
\begin{align}
    \ket{\Phi} &= \frac{1}{\sqrt{2^n}}
        \sum_{x \in \set{0,1}^n} \ket{x}, &
    \P &= \proj{\Phi},
    \label{eq:Phi}
\end{align}
where the latter corresponds to $2^n \times 2^n$ matrix with all entries equal to $1/2^n$.

Quantum computation proceeds by applying unitary transformations to the state. The information is then read out by applying a measurement.

\paragraph{Unitary transformations.} The evolution of a $d$-dimensional quantum system is described by a \emph{unitary transformation}, i.e., a complex $d \times d$ matrix $U$ such that $UU^\dagger = I$, where $U^\dagger = \bar{U}\tp$ denotes the conjugate transpose of $U$. If a unitary $U$ is applied only on the $A$ register of a joint system $AB$ that is in state $\ket{\psi}_{AB}$, we write $(U_A\otimes\1_B) \ket{\psi}_{AB}$
where $\1$ denotes the \emph{identity transformation}. We will often abbreviate this as $U_A \ket{\psi}_{AB}$.

\paragraph{Measurement.} We can extract information from a quantum state $\ket{\psi}$ by performing a measurement. For our purpose it will be enough to consider only projective measurements. A \emph{projective measurement} is described by a set $\{P_1,\dotsc,P_k\}$ of orthogonal projectors ($P_i\ct = P_i$ and $P_i^2 = P_i$) such that $\sum_{i=1}^k P_i = \1$. When performing a measurement on a quantum state $\ket \psi$, the probability of getting outcome $i$ is $p(i) = \bra\psi P_i\ket \psi$. Upon getting outcome $i$, the state $\ket\psi$ collapses to $P_i\ket\psi/\sqrt{p(i)}$. Given a composite system $AB$, a measurement on the subsystem $A$ has operators of the form $(P_i)_A \x \1_B$, and the outcome probabilities and post-measurement states are determined analogously \cite{watrous2018theory}. 

\paragraph{Quantum-accessible oracles.}
On a quantum computer, a function can be evaluated on a superposition of inputs. The standard way of modelling superposition black-box access to a function $f:\{0,1\}^n\to\{0,1\}^m$ is by providing an oracle for the unitary operation $O_f$ that acts on $n+m$ qubits and is defined by 
\begin{equation}
    O_f \ket x \ket y \mapsto \ket x \ket{y \oplus f(x)},
    \label{eq:Of}
\end{equation}
for all $x\in\{0,1\}^n$ and $y\in\{0,1\}^m$.
Without loss of generality, an algorithm $\A$ that makes $q$ queries to such an oracle has the following form:
\[
    U_q O_f \cdots U_1 O_f U_0 \ket{\Psi_0}
    = V^{O_f}_{\A}\ket{\Psi_0}
    = \ket{\Psi},
\]
possibly followed by a measurement.
Here, $\ket{\Psi_0}$ is an initial state $U_i$ are arbitrary unitary operations that do not depend on $f$.

In this work, we will deal with algorithms that have two oracles, $O_1$ and $O_2$, but may only query $O_2$ at most once ($O_1$ will be a random oracle and $O_2$ a signing oracle for a one-time signature scheme). In this case, we can regard an algorithm $\A^{O_1, O_2} = (\A_0^{O_1},\A_1^{O_1})$ as a two-stage process: $\A_0^{O_1}$ prepares the input for $O_2$ and an internal register, while $\A_1^{O_1}$ receives the internal state and the output of $O_2$, and produces the final output of $\A$. In other words, the execution of $\A$ results in the state
\[
    \ket{\Psi} =
    V^{O_1}_{\A_1}
    O_2
    V^{O_1}_{\A_0}
    \ket{\Psi_0}.
\]

The most well-known situation in cryptography that features a quantum oracle is the so-called \emph{quantum random oracle model} (QROM) \cite{boneh2011random}. In the QROM, just as in the classical random oracle model (ROM) \cite{Bellare1993a}, a hash function is modeled as a uniformly random function $h$ that all agents have oracle access to, meaning that quantum oracle access to the unitary $O_h$ defined in \cref{eq:Of} is provided. This model is used to prove cryptographic security against quantum adversaries when basing security on concrete properties like, e.g., collision resistance, is hard or inefficient.

\subsection{Tools from linear algebra}

In this section, we state a couple of simple lemmas used in security proofs in \cref{sec:lamport,sec:winternitz}. For the first lemma, we use the formulation from \cite{boneh2013secure} (Lemma~2.1), and the proof is also provided in the same reference.

\begin{lemma}[Special case of the pinching lemma \cite{hayashi2002optimal}]\label{Pinching} 
Let $\A$ be a quantum algorithm and $x$ any output value of $\A$. Let $\A_0$ be another quantum algorithm obtained from $\A$ by pausing $\A$ in an arbitrary stage of execution, performing a projective measurement that obtains one of $k$ outcomes, and then resuming $\A$. Then,
\begin{equation}
    \Pr[\A_0(1^n) = x] \geq \frac{\Pr[\A(1^n) = x]}{k}.
\end{equation}
\end{lemma}

\begin{lemma}\label{lem:commutators}
Let $A$ and $\{B_i\}_{i=1}^n$ be operators, acting on the same space, such that $\|A\|_\infty, \|B_i\|_\infty\leq 1$. Then
\begin{align*}
    \norm[\Big]{\sof[\Big]{
    A,\prod_{i=1}^nB_i
    }}_\infty
    &\leq
    \sum_{i=1}^n
    \norm{[A,B_i]}_\infty.
\end{align*}
\end{lemma}

\begin{proof}
Note that
\begin{equation}
    [A,BC] = [A,B]C + B[A,C]
    \label{eq:ABC}
\end{equation}
for any operators $A,B,C$ acting on the same space.
Hence,
\begin{align*}
    \norm[\Big]{\sof[\Big]{
        A, \prod_{i=1}^n B_i
    }}_\infty
 &= \norm[\Big]{
        [A,B_1\big] \prod_{i=2}^n B_i
        + B_1 \sof[\Big]{A,\prod_{i=2}^n B_i}
    }_\infty\\
    &\leq
    \norm[\Big]{
        [A,B_1] \prod_{i=2}^n B_i
    }_\infty +
    \norm[\Big]{
        B_1 \sof[\Big]{A,\prod_{i=2}^n B_i}
    }_\infty\\
    &\leq
    \norm{[A,B_1]}_\infty
    \norm[\Big]{\prod_{i=2}^n B_i}_\infty +
    \norm{B_1}_\infty
    \norm[\Big]{\sof[\Big]{
        A, \prod_{i=2}^n B_i
    }}_\infty\\
    &\leq
    \norm{[A,B_1]}_\infty +
    \norm[\Big]{\sof[\Big]{
        A, \prod_{i=2}^n B_i
    }}_\infty,
\end{align*}
where the first two inequalities follow from the triangle inequality and the sub-multiplicative property of the norm, respectively, and the last inequality holds because 
$\norm{B_1}_\infty \leq 1$ and
\[\norm[\Big]{\prod_{i=2}^n B_i}_\infty
\leq
\norm{B_2}_\infty
\norm{B_3}_\infty \dotsb
\norm{B_n}_\infty
\leq 1.\]
The desired inequality follows by applying the same argument inductively.
\end{proof}

\begin{lemma}\label{lem:PQ}
Let $X$ and $Y$ be two $n$-qubit quantum systems and let
\[
    P^=_{XY}
    = \sum_{x \in \{0,1\}^n}
    \proj{x}_X \x \proj{x}_Y
\]
be the projector onto the subspace spanned by those computational basis vectors where the two registers are equal.
Let $\P = \proj{\Phi}$ denotes the projector onto the uniform superposition, see \cref{eq:Phi}.
Then
\begin{equation}\label{relation1}
    \norm{P^=_{XY}\P_Y}_\infty
    = 2^{-n/2}.
\end{equation}
\end{lemma}

\begin{proof}
Recall that
$\ket{\Phi} = 2^{-n/2} \sum_{x \in \set{0,1}^n} \ket{x}$,
so $\ket{x}\braket{x}{\Phi}\bra{\Phi} = 2^{-n/2} \ketbra{x}{\Phi}$ for any $x \in \set{0,1}^n$.
Hence
\begin{align}
    \norm{P^=_{XY}\P_Y}_\infty
 &= 2^{-n/2} \norm*{
        \sum_{x\in\{0,1\}^n}
        \proj{x}_X \x
        \ketbra{x}{\Phi}_Y
    }_\infty \\
 &= 2^{-n/2}
    \max_{x\in\{0,1\}^n}
    \norm[\big]{
        \ketbra{x}{\Phi}_Y 
    }_\infty \\
 &= 2^{-n/2},\label{bound oprator norm}
\end{align}
where the second equality holds since the matrix is block diagonal and the last line follows from
the fact that $\ket{x}$ and $\ket{\Phi}$ are unit vectors.
\end{proof}

By applying the triangle inequality, \cref{lem:PQ} implies the following bound on the commutator:
\begin{equation}\label{eq:commPeqQ}
    \norm{\sof{P^=_{XY},\P_Y}}_\infty
    \le 2\cdot 2^{-n/2}.
\end{equation}

\subsection{Hash-based one-time signature schemes}

Hash-based signature schemes \cite{Lam79,merkle1989certified} are among the digital signature schemes whose security relies on the weakest assumptions.
In this paper, we study hash-based one-time signatures (OTSs) which are digital signature schemes that use a pair of keys to sign and verify a single message. Their classical security can be based on the existence of a family of hash functions with certain properties such as one-wayness, collision resistance, pre-image resistance and/or second pre-image resistance. In this section, we present the  Lamport OTS and the Winternitz OTS.

\subsubsection{The Lamport OTS}\label{sec:LOTS}

The Lamport OTS (also known as Lamport--Diffie OTS), first introduced in \cite{Lam79}, is used as basis for many other hash-based signature schemes.
This scheme uses a hash function $h:\{0,1\}^n \to \{0,1\}^n$ for key generation and verification, where $n \in \N$ denotes the security parameter.
\begin{enumerate}
    \item \emph{Parameters}. Security parameter $n \in \N$ and message length $l \in \N$.
    \item \emph{Key generation algorithm} ($\KeyGen$). On input of the security parameter $n$ in unary, $\KeyGen$ outputs a secret signing key $\sk$ and a public verification key $\pk$: $(\sk,\pk) \leftarrow \KeyGen(1^n)$ as follows,
    \begin{align*}
        \sk &= (s^j_i)_{i = 1,\dotsc,l}^{j = 0,1}
        \quad \text{with} \quad
        s^j_i \rgets \{0,1\}^n, \\
        \pk &= (p^j_i)_{i = 1,\dotsc,l}^{j = 0,1}
        \quad \text{where} \quad
        p^j_i = h(s^j_i) \in \{0,1\}^n.
    \end{align*}
    \item \emph{Signature algorithm} ($\Sign_{\sk}$). On input of a message $m = m_1 \dots m_l \in \{0,1\}^l$ of length $l$, $\Sign_{\sk}$ outputs the following signature:
    \begin{align*}
        \Sign_{\sk}(m) = \sigma = \sigma_1 \dots \sigma_l
        \quad \text{where} \quad
        \sigma_i = s^{m_i}_i \in \{0,1\}^n.
    \end{align*}
    \item \emph{Verification procedure} ($\Ver_{\pk}$). The verification procedure checks the correctness of the signature using the public key $\pk$. Upon receiving a message $m$ and a signature $\sigma = \sigma_1 \dots \sigma_l$, $\Ver_{\pk}$ outputs the following:
    \begin{equation*}
        \Ver_{\pk}(m,\sigma) =
        \begin{cases}
            \acc & \text{if $h(\sigma_i) = p_i^{m_i}$ for all $i \in \set{1,\dotsc,l}$}, \\
            \rej & \text{otherwise}.
        \end{cases}
    \end{equation*}
\end{enumerate}

\subsubsection{The Winternitz OTS}\label{sec:WOTS}

The Winternitz OTS was first introduced by Merkle \cite{merkle1989certified} and many variants of the Winternitz OTS have been proposed since then. In this work, we study a variant that uses a hash function $h: \{0,1\}^n\to \{0,1\}^n$. It involves several parameters and consists of three probabilistic polynomial-time algorithms defined as follows:
\begin{enumerate}
    \item \emph{Parameters}. The scheme is parameterized by a security parameter $n$, binary message length $a$, and the Winternitz parameter $w \ge 2$ that determines the time-memory trade-off (these parameters are integers and are publicly known). Based on parameters $a$ and $w$ we define
    \begin{equation}\label{eq:Winternitz-params}
        l_1 = \ceil*{\frac{a}{\log(w)}}, \qquad 
        l_2 = \floor*{\frac{\log(l_1(w-1))}{\log(w)}} + 1, \qquad 
        l = l_1 + l_2.
    \end{equation}

    \item \emph{Key generation algorithm} ($\KeyGen$). On input of security parameter $n$, the key generation algorithm first chooses uniformly at random $l$ values that form the signing key, that is,
    $\sk = (s_1,\dotsc,s_l) \rgets (\{0,1\}^n)^l$ 
    that form the secret signing key.
   Then, computes the public verification key $\pk$ as follows:
    \begin{equation*}
        \pk
        = (p_1,\dotsc,p_l)
        = \of[\big]{
            h^{w-1}(s_1),\dotsc,
            h^{w-1}(s_l)
        }.
    \end{equation*}

    \item \emph{Signature algorithm} $(\Sign_\sk)$. Given an input message $x\in\{0,1\}^a$ and a secret key $\sk$, the signature algorithm first computes a base-$w$ representation of $x$ with $l_1$ digits: $m = (b_1,\dotsc,b_{l_1})$ where $b_i\in\{0,\dotsc,w-1\}$.
    Next, it computes the checksum
    \[C(m) = \sum_{i=1}^{l_1}(w-1-b_i)\]
    and represents it as $l_2$ digits $C(m) = (b_{l_1+1},\dotsc,b_l)$ in base $w$. Note that the length of the base-$w$ representation of $C(m)$ is at most $l_2$ since $C(m) \leq l_1 (w-1)$. The reader may refer to \cite{even1996line} for more details on the checksum. We set
    \begin{equation}
        b(m) = (b_1,\dotsc,b_l) = m \parallel C(m),
        \label{eq:b(m)}
    \end{equation}
    the concatenation of the base-$w$ representations of $m$ and $C(m)$. The signature is then computed as
    \begin{equation*}
        \sigma = (\sigma_1,\dotsc,\sigma_l)
        = \of[\big]{
            h^{b_1}(s_1),\dotsc,
            h^{b_l}(s_l)
        }.
    \end{equation*}
    Notice that the checksum can be considered as an intermediate verification step. Given the value of $b$ corresponding to a message $m$, it guarantees that the $b'$ corresponding to any other message $m'\neq m$ contains at least one $b'_i < b_i$, $1\leq i\leq l$.

    \item \textit{Verification algorithm} $(\Ver_\pk)$. Given a message $m$ of binary length $a$, a signature $\sigma$ and the public verification key $\pk$, the verification algorithm first computes the $(b_1,\dotsc,b_l)$ as described above and then checks whether the value of $h^{w-1-b_i}(\sigma_i)$ agrees with the public key $p_i$:
    \begin{equation*}
        \Ver_{\pk}(m,\sigma) =
        \begin{cases}
            \acc & \text{if $p_i = h^{w-1-b_i}(\sigma_i)$ for all $i \in \set{1,\dotsc,l}$}, \\
            \rej & \text{otherwise}.
        \end{cases}
    \end{equation*}
\end{enumerate}

\subsection{Blind unforgeability}\label{sec:BU}

\emph{Blind unforgeability} (\BU) \cite{alagic2020quantum} is a quantum-access replacement for the standard security notion of \EUCMA, introduced in \cite{goldwasser1988digital}, for message authentication codes and digital signature schemes. It uses the concept of \emph{blinding}.
Before we recall the definition of this notion, we need to introduce some additional background, including the concept of \emph{blinding} a function and the \emph{blind forgery} experiment.

\begin{definition}[Blinding a function]\label{def:blinding set}
Let $f : X \to Y$ be a function and $B \subset X$ a subset of $X$. The \emph{blinded} function $Bf$ with respect to the \emph{blinding set} $B$ is defined as
\begin{equation}
    Bf(x) =
    \begin{cases}
        \bot & \text{if $x \in B$},\\
        f(x) & \text{otherwise},
    \end{cases}
\end{equation}
where $\bot$ is a special blinding symbol.
One concrete way to instantiate this is by means of an extra bit. In that case, given a function $f : \{0,1\}^n \to \{0,1\}^m$, we define $Bf : \{0,1\}^n \to \{0,1\}^{m+1}$ by
\begin{equation}
   Bf(x) =
    \begin{cases}
        0^n \parallel 1 & \text{if $x \in B$},\\
        f(x) \parallel 0 & \text{otherwise}.
    \end{cases}
\end{equation}
\end{definition}

The second definition is more convenient because it enables us to easily measure and control from this bit without modifying the output of the function $f$.

\paragraph{Blinding a signing algorithm.}\label{p:blinding}
Let $\Sign_\sk$ be a signing algorithm for a signature scheme with message space $M$. Now, sample a blinding set $B\subseteq M$ by adding every input with probability $\epsilon$, independently. Then, the blinded signing algorithm is given by
\begin{equation}
    B \Sign_{\sk}(m) =
    \begin{cases}
        \bot & \text{if $m \in B$},\\
        \Sign_{\sk}(m) & \text{otherwise}.
    \end{cases}
\end{equation}
Note that in the rest of the document, we refer to $B^c$ as the subspace of un-blinded messages.
\paragraph{The blind forgery (BF) experiment.}
Let $S = (\KeyGen,\Sign,\Ver)$ be a digital signature scheme with a security parameter $n$ and message space $M$. Let $\A$ be an adversary and let $\epsilon : \N \to \R_+$ be a negligible function.
We define the \emph{blind forgery} experiment  $\BlindForge_{S,\A}(n,\epsilon)$ as follows:
\begin{itemize}
    \item \textit{Key generation}: $(\sk,\pk) \gets \KeyGen(1^{n})$;
    
    \item \textit{Generation of blinding set}: select the blinding set $B \subseteq M$ by choosing each $m \in M$ independently at random with probability $\epsilon (n)$ provided by the adversary $\A$.;
    
    \item \textit{Forgery}: $(m,\sigma) \gets \A^{B\Sign_{\sk}}(1^{n})$;
     
    \item \textit{Outcome}: win if
    $\Ver_\pk(m,\sigma) = \acc$ and $m \in B$,
    and lose otherwise.
\end{itemize}

\begin{definition}[Blind unforgeability (\BU)]\label{BU}
A digital signature scheme $S$ is \emph{q-\BU secure} if for any adversary $\A$ making at most $q$ queries to $B\Sign_\sk$ and for all $\epsilon$, the success probability of winning the BF experiment is negligible in the security parameter $n$, i.e.
\begin{align}
    \Pr\sof[\big]{\textnormal{$\A$ wins $\BlindForge_{S,\A}(n,\epsilon)$}} \le \mu(n)
\end{align}
for some negligible function $\mu$.
\end{definition}
We remark that the quantum algorithm $\A$ and the blinding fraction $\epsilon$ are declared in a uniform fashion prior to the experiment.

This paper is concerned with one-time signature schemes for which the pair of keys is used only once. That is, a $\BlindForge$ notion of security for one-time signature schemes in which the adversary is allowed to query the $B\Sign_\sk$ algorithm only once, i.e.~$q=1$.

\section{Hash chains in the QROM}\label{Hash chains}
\subsection{Quantum hash chain sampling}\label{hash oracles}

In this section, we introduce hash chains, several closely related worlds, and show that we can work in the one that is the easiest to handle. More precisely, we describe this technical tool that
we will use to prove \BU security for the Lamport and Winternitz OTSs.
For both OTSs, the $\KeyGen$ routine computes so-called \emph{hash chains}, i.e.~sequences of strings obtained by iteratively applying a hash function.

In the (Q)ROM, to generate a hash chain based on a hash function $h$, we first sample an initial string $s_0$ uniformly at random and then compute $s_i = h(s_{i-1})$ for $i=1,\dotsc,w-1$ to obtain a hash chain of length $w$. For key generation in the Lamport or Winternitz OTS, the secret key $\sk$ is a tuple of $l$ initial strings $s_1, \dotsc, s_l \rgets \{0,1\}^n$ sampled uniformly at random. Then a tuple of hash chains $\gamma = (\gamma_i^j)_{i = 1,\dotsc,l}^{j = 0,\dotsc,w-1}$ is obtained by querying the hash function $h$ on each string of the secret key $w-1$ times, i.e.
\[
\gamma_i^0 = s_i, \quad
\gamma_i^j = h^j(\gamma_i^0), \quad
p_i = \gamma_i^{w-1} = h^{w-1}(\gamma_i^0), \quad
j = 0,\dotsc,w-1, \quad
i = 1,\dotsc,l,
\]
where $w$ is the length of the hash chain ($w=2$ for Lamport) and $l$ is the number of hash chains. The final entry of each chain is used as a public key.

In the $\BlindForge$ game, the secret key is only used by the blinded signing oracle. When analyzing this experiment, we can thus modify the key generation, signing and random oracle algorithms in an arbitrary way, as long as the modified triple is indistinguishable from the real one to an adversary.

In the proofs in \cref{sec:lamport,sec:winternitz} we make use of the following modified triple, which we will refer to as defining the \QWorld. We construct the secret key and the intermediate hash chain elements initially in uniform superposition. That is, we prepare each hash chain register $(\Gamma_i^j)_{i=1,\dotsc,l}^{j=0,\dotsc,w-2}$ in the uniform superposition state $\ket{\Phi}$, with the intention of measuring them to sample the strings $\gamma_i^j$ in mind. 
Then, we sample the final hash chain at random. The random oracle is then ``reprogrammed in superposition'' to be approximately consistent with the hash chains.

We proceed to show that the way of implementing the hash chain and the random oracle in the \RealWorld and in the \QWorld are indistinguishable.
For that purpose, we first formally define both worlds and some intermediate worlds between them. Each world is specified by two oracles, $H$ and $\Sign$, replacing the random oracle $h$ and the signing oracle in the \RealWorld (the $\KeyGen$ algorithm is implicitly replaced by the setup described in each world below, that generates the initial state and the public key). The oracles of the \QWorld are described below as well.

\paragraph{Real world.}
In the \RealWorld, the first element $\gamma_i^0$ of each hash chain $\gamma_i$ is generated at random and the hash function is queried to generate the rest of the hash chain, i.e.,
\[s_i = \gamma_i^0 \rgets \{0,1\}^n, 
\gamma_i^j = h^j(\gamma_i^0)~;~
p_i = \gamma_i^{w-1} = h^{w-1}(\gamma_i^0);~
j = 0,\dotsc,w-1;~i = 1,\dotsc,l.
\]
Here, the random oracle is implemented at random, i.e. $H=h$, the $\Sign$ oracle uses the secret key $\sk$ defined above.

\paragraph{Intermediate world 1.}
Here, the first hash chain element  is generated at random, the following hash chain elements are successively sampled uniformly except for the collision tuples. That is
\[s_i = \gamma_i^0 \rgets \{0,1\}^n ~;~\gamma_i^1 \text{ is uniform except for the case where } \gamma_i^1 = \gamma_{i'}^1 \text{ if }
\gamma_i^0 = \gamma_{i'}^0,\]
\[\gamma_i^2 \text{ is uniform except for the cases where }
\gamma_i^2 = \gamma_{i'}^1 \text{ if }
\gamma_i^1 = \gamma_{i'}^0 ~;~ \gamma_{i'}^2 = \gamma_i^2 \text{ if } \gamma_i^1 = \gamma_{i'}^1, \cdots
\]
This world is very similar to the \RealWorld, the only difference is that here we first sample the secret and public key (hash chain), then we reprogram the random oracle according to the secret and public key that we sampled, i.e.,
$$
H(x)=\begin{cases}
   \gamma_i^{j+1}& \text{if }x=\gamma_i^j \text{ with } j\le w-2, \\
   h(x)& \text{else}.
\end{cases}
$$
The $\Sign$ oracle is the same as in the \RealWorld.

\paragraph{Intermediate world 2.}
In this world, the hash chain elements $\gamma_i^j$ are first sampled uniformly at random with possible collision tuples. It means that the $\gamma_i^j$ are uniformly independent strings. Afterwards, the random oracle is reprogrammed such that it is consistent with the secret and public keys. When the random oracle is queried, it compares the input with the hash chain. If the input is not equal to any of the hash chain elements, the random oracle answers with a random function $\hat{h}$. Otherwise, for each hash chain element the input is equal to, it XORs the next hash chain element into the output register. If there are two hash chain elements that are the same, the random oracle XORs both next hash chain elements into the output register. More formally,
\[ H(x) =
        \begin{cases}
            \displaystyle\bigoplus _{\substack{i,j:j\le w-2 \\\text{and }\gamma_i^j=x}}
            \gamma_i^{j+1} & \text{if there exists $(i,j)$ such that $\gamma_i^j = x$ with $j = 0,\dotsc,w-2$}, \\
            h(x) & \text{otherwise}.
        \end{cases}
\]
In this case, the $\Sign$ oracle uses the full list of hash chains $(\gamma_i^j)_{i=1,\dotsc,l}^{j=0,\dotsc,w-1}$ to answer the query with all the hash chain elements consistent with the input. 

\paragraph{Quantum independent world.}
In this world, the hash chain registers $(\Gamma_i^j)_{i=1,\dotsc,l}^{j=0,\dotsc,w-2}$ are initially prepared in the uniform superposition $\ket \Phi$, and the last hash chain elements $(\gamma_i^{w-1})_{i=1,\dotsc,l}$ are sampled uniformly at random. The random oracle is constructed in such a way that it is compatible with the hash chain. When queried with register $X$ and $Y$, the random oracle compares the $X$ and $\Gamma$ registers, then answers the query in the $Y$ register. Abstractly speaking, $H$ is implemented as in the \IWorld{2}, except that the comparison and XOR operations involving $\gamma_i^j$ are replaced by controlled unitary operations with $\Gamma$ as the control register.
Here, $X$ is input register of length $n$, $Y$ is the output register which is of length $n$, and $\Gamma$ represents the hash chain register and is of length $lw$.

To be more specific, let us describe in detail the behavior of the random oracle. The following definitions are for the Winternitz OTS, with the Lamport OTS being a special case. For each $i \in \set{1,\dotsc,l}$ and $j \in \set{0,\dotsc,w-2}$, let $U^j_i$ be the following unitary that compares the input register $X$ with the hash chain register $\Gamma^j_i$ and places the contents of the subsequent register $\Gamma^{j+1}_i$ in $Y$ if they are equal:
\begin{equation}
    (U^j_i)_{XY\Gamma^j_i\Gamma^{j+1}_i}
    = P^=_{X\Gamma^j_i} \x
      (\CNOT\xp{n})_{\Gamma^{j+1}_i:Y}
    + P^{\neq}_{X\Gamma^j_i}
    \x \1_{\Gamma^{j+1}_i Y},
    \label{eq:Uij definition}
\end{equation}
where the controlled-$\NOT$ gates use $\Gamma^{j+1}_i$ as control and $Y$ as target, and the projectors $P_{=}$ and $P_{\neq}$ check whether the input register $X$ is equal to the corresponding hash chain register $\Gamma^j_i$:
\begin{align}
    P^=_{X\Gamma^j_i}
    &= \sum_{x\in\{0,1\}^n} \proj{x}_X \x \proj{x}_{\Gamma^j_i}, &
    P^{\neq}_{X\Gamma^j_i}
    &= \1 - P^=_{X\Gamma^j_i}.
    \label{Px=}
\end{align}
Combining these equations, we can equivalently write
\begin{equation}
    (U^j_i)_{XY\Gamma^j_i\Gamma^{j+1}_i}
    = P^=_{X\Gamma^j_i} \x \of[\Big]{
      (\CNOT\xp{n})_{\Gamma^{j+1}_i:Y}
      - \1
     }
    + \1,
    \label{Uij}
\end{equation}
In case $j+1=w-1$, the above definition of $U_i^j$ still applies in the sense that we can take $\Gamma_i^{w-1}$ to be the register that stores the $i$-th block of the public key.

The overall unitary that is applied upon a hash query is \footnote{Note that the ordering of the product is unimportant because the operators $U_i^j$ commute.}
\begin{align}
    (U_h)_{XY\Gamma} =
    \of*{
        \prod_{i=1}^l
        \prod_{j=0}^{w-2}
        (U_i^j)_{XY\Gamma_i^j\Gamma_i^{j+1}}
    }
    U^{\neq}_{XY\Gamma}
    \label{eq:Uh}
\end{align}
where the unitary $U^{\neq}_{XY\Gamma}$ corresponds to the case where the input $x$ is not equal to any part of the hash chain register $\Gamma$:
\begin{align}
    U^{\neq}_{XY\Gamma}
    &= P^{\neq}_{X\Gamma} U'_{XY}
    + \of*{
        \1_{X\Gamma}
        - P^{\neq}_{X\Gamma}
    } \x \1_{Y} \\
    &= P^{\neq}_{X\Gamma}
    \cdot \of*{
        U'_{XY} - \1
    }
    + \1
    \label{eq:Uneq}
\end{align}
where, with a slight abuse of notation,
\begin{equation}
    P^{\neq}_{X\Gamma} =
    \prod_{i=1}^l
    \prod_{j=0}^{w-2}
    P^{\neq}_{X\Gamma_i^j}
    \label{eq:Pneq}
\end{equation}
denotes the projector onto the subspace of $X\Gamma$ where $X$ is not equal to any of the $\Gamma_i^j$ registers, and $U'_{XY}$ is the standard random oracle unitary that answers the query by XOR-ing the hash value $h(x)$ in the $Y$ register regardless of the entire hash chain register:
\begin{equation}
    U'_{XY}\ket{x}_X\ket{y}_Y =
    \ket{x}_X\ket{y\oplus h(x)}_Y.
    \label{eq:U'}
\end{equation}
The conditions on the control registers in \cref{eq:Uh} are such that, for any input $x$, only one of the unitaries in the product is applied.

\subsubsection{Additional details for the Lamport and Winternitz OTS}\label{subsubsec:QI-Lamport} 

As mentioned, the above definitions specialize to the Lamport OTS. Here, the pair of indices $(i,j)$, specifying a message bit's position and its value, replace the index $i$ in the Winternitz setting, and the hash chains have only length two, with the first ($j=0$ above), and second ($j=1$ above) elements given by the secret key strings, and public key strings, respectively.

Lastly, we describe the behaviour of the blinded signing oracle.  We define the action of the oracle for inputs where the register $M$ is in a computational basis state $\ket m$, which is sufficient by the linearity of the quantum oracle for the blinded signing function. 

When queried with register $M$ and $\Sigma$, the signing oracle controls on the message $m$ not being in the blinding set $B$ and answers the query by XOR-ing the signature into the $\Sigma$ register. 
For ease of notation, let $\Gamma_i^{w-1}$ be registers prepared in state $\ket{p_i}$ for $i=1,\dotsc,l$.  Then, for a fixed message $m$, the signing oracle for the Winternitz OTS operates as follows:
\begin{align*}
    B\Sign _{\sk}\ket m _M =
        \begin{cases}
           \ket m_M\otimes\1  & \text{ if } m\in B, \\
            \ket m_M\otimes
            \Big(\displaystyle\bigotimes_{i=1}^l
            \CNOT^{\otimes n}_{\Gamma_i^{b_i}:\Sigma_i}
            \Big)& \text{otherwise}.
        \end{cases}
\end{align*}
For the Lamport OTS, in the \RealWorld, when queried with an input $m$ of length $l$, the signing oracle answers the query with an $l$ $n$-secret key strings corresponding to each bit of the message input. In contrast, in the \QWorld, the queried message is a quantum state. Thus the signing oracle answers the query by XOR-ing the corresponding secret key sub-registers in the output register $Y$. Specifically,
for a fixed message $m$, the signing oracle acts as follows:
\begin{align*}
    B\Sign _{\sk}\ket m _M =
        \begin{cases}
           \ket m_M\otimes\1  & \text{ if } m\in B, \\
            \ket m_M\otimes
            \Big(\displaystyle\bigotimes_{i=1}^l
            \CNOT^{\otimes n}_{S_i^{m_i}:\Sigma}
            \Big)& \text{otherwise}.
        \end{cases}
\end{align*}

\subsection{Indistinguishability}

Finally, we prove a number of lemmas which together allow us to conclude the indistinguishability of the \RealWorld and the \QWorld.
\begin{lemma}\label{lem:ind-r-v-qi}
    Let $p$ and $q$ be the output distributions over $n$-bit strings of an algorithm $\A$ interacting with the \RealWorld and the \QWorld, respectively. Then
    \begin{align}
        \big\| p - q \big\|_1
        \leq \frac{3(wl)^2}{2^n}.
    \end{align}
\end{lemma}
 
We prove \cref{lem:ind-r-v-qi} via a sequence of lemmas.
 
\begin{lemma}\label{lem:real world and intermediate world 1}
The \RealWorld and the \IWorld{1} are indistinguishable.
\end{lemma}
 
\begin{lemma}\label{lem:intermediate world 1 and intermediate world 2}
The distribution $p$ and $q$ of hash chains in the \IWorlds are close:
\begin{align}
    \big\| p - q \big\|_1
    \leq \frac{3(wl)^2}{2^n}.
\end{align}
\end{lemma}
 
\begin{proof}
Let $p$ and $q$ be hash chain probability distributions corresponding to the \IWorlds, respectively:
\begin{align*}
    p &: \underline{\gamma_i^0} \rgets \{0,1\}^n, \quad
    i = 1,\dotsc,l, \qquad
    \underline{\gamma_i^j} = H^j
    (\underline{\gamma_i^0}), \quad
   \underline\gamma 
   = (\gamma_i^j)_{i=1,\dotsc,l}^{j=0,\dotsc,w-1}, \text{ and }\\
    q &: \underline{\gamma_i^j} \rgets \{0,1\}^n, \quad
    i = 1,\dotsc,l, \quad
    j = 0,\dotsc,w-1
\end{align*}
We want to show that $p$ and $q$ are close and that the probability that collisions occur in both distributions is small.
Let $C\subset\left(\{0,1\}^n\right)^{lw}$ denote the subset of tuples containing a collision, i.e.~$\gamma\in C$ iff there exist $i,j,i',j'$ such that $\gamma_i^j=\gamma_{i'}^{j'}$.
One can easily check that $p$ and $q$ are equal, conditioned on the subset $C^c$ of collision-free tuples. Then the total variation distance between $p$ and $q$ is
\[\big\| p - q \big\|_1 
\leq p(C) + q(C) + \max\{p(C), q(C)\}.\]

Given that we can easily compute the probability of collision-free tuples in both distributions, we can first compute $p(C^c)$ and $q(C^c)$, and deduce the probabilities $p(C)$ and $q(C)$.

In the distribution $q$, the $\gamma_i^j$ are independent random $n$-bit strings. Thus,
\begin{align}
    q(C^c)
    = \binom{2^n}{lw}(lw)!2^{-nlw} 
    = \frac{(2^n)!2^{-nlw}}{(2^n - lw)!}
    &= 2^n(2^n-1)\cdots(2^n-lw+1)2^{-nlw}.
    \label{probability of no collision in q}
\end{align}
Now, let us find a lower bound of the probability of collision-free tuples in the distribution $q$ so that we can have an insight on the range of possible collisions occurring in $q$. We know that
\[2^n(2^n-1)\cdots(2^n-lw+1) \geq (2^n-wl)^{wl} = 2^{nwl}\left(1-\frac{wl}{2^n}\right)^{wl}.\]
But, setting $f(x) = \left(1-x\right)^{wl}$ with $x = \frac{wl}{2^n}$, one can easily see that $f$ is convex and differentiable on the interval $[0,1]$. In particular, $f$ is differentiable at $0$ and $f(x) \geq xf'(0) + f(0) = - \frac{(wl)^2}{2^n} +1$. Thus,
\[(2^n-wl)^{wl}
= 2^{nwl}\left(1-\frac{wl}{2^n}\right)^{wl}
= 2^{nwl}f(wl/2^n)
\geq 2^{nwl}\left(1 - \frac{(wl)^2}{2^n}\right).\]
Hence
\begin{align*}
    q(C^c) &\geq 1 - \frac{(wl)^2}{2^n}, &
    q(C) &\leq \frac{(wl)^2}{2^n}.
\end{align*}
    
Next, we compute $p(C)$. Here we first compute $p(C^c)$ as well. 
To derive the probability of collision-free tuples, we can sample the elements of the distribution one by one, starting with $\gamma_1^0$ and choosing the $\gamma_i^0$ before moving onto $\gamma_1^1$, etc. The crucial observation is that, when conditioning on collision-free tuples, each $\gamma_i^j$ is choosen uniformly from the set of strings that have not yet occurred.
We hence get the same probability for the set $C^c$ as for the uniform distribution $q$:
\begin{align*}
    p(C^c) &= q(C^c) \geq  1 - \frac{(wl)^2}{2^n}, &
    p(C) &= q(C) \leq \frac{(wl)^2}{2^n}.
\end{align*}
Therefore, the probability of collision occurring in both distributions can be at most $(wl)^2/2^n$ which is negligible because $w$ is constant, $l$ is polynomial in $n$, and $n $ quite large. 

Plugging the probabilities $ p(C)$and $q(C)$ of collisions occurring in both distributions  into our above expression of total variation distance give 
\begin{align}
    \big\| p - q \big\|_1 
    &\leq p(C) + q(C) + \max\{p(C), q(C)\}
    \leq \frac{3(wl)^2}{2^n}
    \label{final total variation distance bound}
\end{align}
as desired.
\end{proof}

\begin{lemma}\label{lem:intermediate world 2 and quantum independent world}
The way of implementing the random oracle in the \IWorld{2} and in the \QWorld are indistinguishable.
\end{lemma}

\begin{proof}
In the \IWorld{2}, the random oracle contains the hash chain, so that when queried, it compares the input with the hash chain and answers the query in the output register. Specifically, the hash chain is only used as a comparison tool, so it is not modified. Similarly, in quantum setting the random oracle is implemented in such a way that it contains the hash chain register. When it receives a query, it controls the hash chain register to see whether there is similitude between the hash chain and the queried input, and answers the query by acting on the output register. Those controlled operations commute with computational basis measurements. From this fact it is easy to see that \IWorld{2} and \QWorld are exactly indistinguishable, see \cite{zhandry2015secure} for details.
\end{proof}

\begin{proof}[Proof of \cref{lem:ind-r-v-qi}]
\Cref{lem:ind-r-v-qi} follows directly from \cref{lem:real world and intermediate world 1,lem:intermediate world 1 and intermediate world 2,lem:intermediate world 2 and quantum independent world}.
\end{proof}

Next, we prove the security of the Lamport and Winternitz OTS from \cref{sec:LOTS,sec:WOTS}, respectively, in the case where adversaries are granted both quantum access to the signing oracle and random oracle.

\section{One-time \BU security of the Lamport OTS}\label{sec:lamport}

In the $\BlindForge$ experiment, the adversary is granted both quantum access to a blinding signing oracle for the digital signature scheme and a the random oracle. For one-time signature schemes, the adversary is allowed only to query the signing oracle at most once. So, to produce a forged message-signature pair, the adversary can make a desired number of quantum queries to the random oracle, then query the signing oracle once, and then query again the random oracle as many times as desired. Our goal is to prove that the success probability of any efficient adversary in producing a valid fresh forged message-signature pair is small. Equivalently, we want to show that the probability that an adversary outputs a correct forged signature on a valid forged message is negligible.

In Lamport OTS, the signature algorithm uses only half of the secret key to produce the signature, and the unused part constitutes the invariant of the secret key. Classically, the property that enables security is that the adversary does not have any information about the invariant of the secret key. Quantumly, our intuition is that since in the $\BlindForge$ experiment the forged message must be outside the queried region, for any queried message, there exists at least one bit in which the forged and queried messages are different. Thus, the secret key corresponding to that specific bit should still be in its initial state, hence in the invariant of the secret key. Therefore, we want to show that regardless of the number of queries to the random oracle and to the blinded signing oracle, no adversary can win the $\BlindForge$ experiment except with negligible probability. Towards that end, we separately analyze three cases: \textit{hash queries before $\Sign$ query}, $\Sign$ query and \textit{hash queries after $\Sign$ query}.

We start by describing the overall strategy that we will use to achieve our goal. For \textit{hash queries before $\Sign$ query}, we know that before any query the entire secret key is in uniform superposition state, thus we define a projector of the secret key register being in uniform superposition, and we establish that this projector approximately commutes with the random oracle unitary. This means that after a moderate number of queries, the secret key registers will still be in uniform superposition. The interpretation of this fact is that the adversary learns almost no information about the secret key. %

In the $\Sign$ query case, the first step is to track the unused part of the secret key. This part can be easily determined in the classical setting since the adversary queries only one message in each query. In contrast, in the quantum setting, since we are looking at quantum queries we have to track the invariant in superposition over the different queried messages. This is difficult because the invariant is different within each term of the superposition, so we cannot simply describe the invariant for the whole state. To address this problem, we perform a partial measurement that tracks the unused part of the secret key register. Then, we show that for any forgery pair, the outcome where \textit{none of the secret key registers relevant to the forged signature is in the invariant} can never occur. Next, we use this result to show that if there are no hash queries, no adversary can produce a valid forgery pair except with small probability. If there are hash queries before the $\Sign$ query, then we define the invariant projector that tracks the invariant of the secret keys after queries and show that this projector is orthogonal to the projector corresponding to the outcome where \textit{none of the secret key registers relevant to the forged signature belong to the invariant}. Then, we show that if there is only $\Sign$ query, this new projector does not change the adversary state immediately after the signature. 
We also establish that if the adversary state after producing forgery is in the range of this new projector, then the adversary has negligible probability to win the $\BlindForge$ game. Besides, we prove that the new projector approximately commutes with the random oracle unitary. %

Finally, for the case of \textit{hash queries after $\Sign$ query}, we use the latter argument of the commutator to prove that after hash queries the final adversary state remains roughly in the image of the invariant projector of the secret key. This implies that \textit{hash queries after $\Sign$ query} do not help the adversary to get relevant information about the secret key. Those results show that even with hash queries before and after the $\Sign$ query, any query-limited adversary has only small probability to win the $\BlindForge$ experiment.

In the remainder of this section, we prove the following theorem.
\begin{theorem}\label{thm:lamport}
The Lamport OTS in \cref{sec:LOTS} is 1-\BU secure if the hash function $h$ is modeled as a quantum-accessible random oracle.
More precisely, let $\A$ be an adversary that plays the $\BlindForge$ game for the Lamport OTS, making a total of $q$ queries to the random oracle. Then $\A$ succeeds with a probability bounded as
\begin{align}
   \Pr[\text{$\A$ wins $\BlindForge$}] &\le 
  l^2\cdot 2^{-n}\left( 3137q^2(l+1)+12\right)\nonumber\\
   &\le 6286 q^2 l^3\cdot 2^{-n} \label{eq:lamport}
\end{align}
where $n$ is the security parameter of the Lamport OTS, $l$ is the message length, and the simplified bound in the last line holds for $q>0$.
\end{theorem}

The proof of \cref{thm:lamport} is presented in steps in the following subsections. In particular, we prove \cref{thm:lamport} in the \QWorld first, and conclude the statement in the \RealWorld via an application of \cref{lem:ind-r-v-qi}. In the remainder of the article, we use a subscript $QI$ to indicate that a probability statement holds in the \QWorld.

We begin by presenting some concepts and tools which will be used in the proof. Subsequently, we prove the steps outlined above as separate lemmas. Finally, we combine them to prove \cref{thm:lamport}.

\subsection{$Q$ measurement for Lamport OTS}\label{sec:Q measurement}
Recall the superposition hash chain formalism from \cref{Hash chains}, in particular the special case of the Lamport OTS key generation discussed in \cref{subsubsec:QI-Lamport}.
Our proof will make use of a projective measurement to track an invariant on the secret key register for the verification of the forged message in the case of no hash queries. Let $(m^*,\sigma^*)$ be a forged message-signature pair with $\sigma^* = s_1^{m^*_1}\cdots s_l^{m^*_l}$, where $(s_i^j)_{i=1,\dotsc,l}^{j=0,1}$ is the secret key and $l$ is the message length.

For any message $m^* \in \set{0,1}^l$ let us define an $(l+1)$-outcome projective measurement $\{Q_1^{m^*},\dots,Q_{l+1}^{m^*}\}$ acting on the secret key registers $(S_i^j)_{i=1,\dotsc,l}^{j=0,1}$. It finds the smallest index $i^*\in \{1,\dotsc,l\}$ for which the register \textit{$S^{m^*_{i^*}}_{i^*}$ is in uniform superposition}, or determines that \textit{none of the relevant secret key registers are in uniform superposition} (this corresponds to the outcome $l+1$). The projectors $Q_i^{m^*}$ are defined in terms of projectors
\begin{align}
    \Phi &= \proj{\Phi}, &
    \Phi^\perp &= \1 - \proj{\Phi}
    \label{eq:Phi01}
\end{align}
that correspond to the uniform superposition $\ket{\Phi}$ and its orthogonal complement. We place them onto different registers depending on the message $m^*$:
\begin{align}
    Q_{i^*}^{m^*}
    &=
    \Phi^\perp_{S_1^{m^*_1}} \x \dotsb \x
    \Phi^\perp_{S_{i^*-1}^{m^*_{i^*-1}}} \x
    \Phi_{S_{i^*}^{m^*_{i^*}}}, &
    Q_{l+1}^{m^*}
    &= \bigotimes_{i=1}^l
    \Phi^\perp_{S_i^{m^*_i}}
    \label{eq:Q}
\end{align}
where $i^* \in \set{1,\dotsc,l}$.
These operators act as $\1$ on all other registers $S_i^j$ that are not specified.

\subsection{Invariant projector}\label{PS projector}

In the hash queries part of our proof, we will need a separate projector $P_S$ to track an invariant of the secret key register. In this section we define this projector and state its several properties as lemmas.

Let $\alpha = (\alpha_i^j)_{i=1,\dotsc,l}^{j=0,1}$ be a $2l$-bit string whose each bit $\alpha_i^j \in \set{0,1}$ indicates that the projector $\Phi(\alpha_i^j)$ is applied on the corresponding secret key register $S_i^j$ where
\begin{align}
    \Phi(0) &= \Phi, &
    \Phi(1) &= \Phi^\perp.
\end{align}
For each string $\alpha$, we define the associated projector $\Phi(\alpha)$ on the whole secret key register $S$ as
\begin{equation}
    \Phi(\alpha)_S =
    \bigotimes_{i=1}^l
    \bigotimes_{j=0}^1
    \Phi(\alpha_i^j)_{S^j_i}.
    \label{eq:Phi(a)}
\end{equation}
Note that this is a complete set of projectors, i.e., $\sum_{\alpha\in\{0,1\}^{2l}} \Phi(\alpha)_S = \1_S$.

Since we are interested in the unused part of the secret key register $S$, we need to filter those $\alpha$'s for which $S^j_i$ is in state $\ket\Phi$.
Recall from our discussion of blind unforgeability in \cref{sec:BU} that $B$ denotes the set of blinded messages.
Since the blinded signing oracle has signed (at most) a single, un-blinded message, the state after the signing oracle call can be written as a superposition of states where, for some un-blinded message $m \in B^c$, the secret key register of the complementary value $\bar{m}_i$ is still in the uniform superposition $\ket{\Phi}$, for all $i$. We collect all strings $\alpha$ that are consistent with no blinded messages having been signed in
\begin{align}
    \widehat{B^c} = \bigcup_{m \in B^c}
    \set[\Big]{
        \alpha \in \set{0,1}^{2l}
        \mathrel{\Big|}
        \text{$\alpha_i^{\bar{m}_i} = 0$ for all $i=1,\dotsc,l$}
    }.
    \label{def of B^c hat}
\end{align}
These strings indicate which secret key registers were not used during hash queries and $\Sign$ query.
Finally, we define
\begin{equation}
    P_S = \sum_{\alpha\in\widehat{B^c}} \Phi(\alpha)_S
    \label{eq:invariant projector}
\end{equation}
as the projector on the subspace %
compatible with $\widehat{B^c}$.
Note that $P_S$ is indeed a projector since it is a sum of mutually orthogonal projectors.

Now, we state some ingredients that we will need to prove our main results both for Lamport OTS and Winternitz OTS. Proofs of these lemmas are provided in \cref{apx:Lamport lemmas}.

\begin{restatable}{lemma}{LemCom}\label{Commutator argument1}
Let $U_h$ be the random oracle unitary for any given function $h$ (see \cref{Hash chains}) and let $\Phi = \proj{\Phi}$ denote the projector onto the uniform superposition. Then, for any $i \in \set{1,\dotsc,l}$ and $j \in \set{0,1}$,
\begin{equation}
    \norm[\Big]{\sof[\Big]{
        U_h,
        \Phi_{S^j_i}
    }}_\infty
    \leq \frac{6}{2^{n/2}}
    = \epsilon_L(n)
    \label{eq:UhPhi}
\end{equation}
is negligible in $n$.
\end{restatable}

For any message in the blinding set $B$, there exists at least one secret key necessary for its corresponding signature in the invariant of the secret key register. In other words, for any valid forged message, at least one of the secret keys needed for its corresponding forged signature is in the uniform superposition state $\ket\Phi$. This implies the following lemma.

\begin{restatable}{lemma}{LemOrth}\label{Orthogonality argument}
For all $m^* \in B$, the projectors $Q_{l+1}^{m^*}$ and $P_S$ defined in \cref{eq:invariant projector,eq:Q} are orthogonal:
\begin{equation}
    Q_{l+1}^{m^*} P_S = 0.
\end{equation}
\end{restatable}

\begin{restatable}{lemma}{LemNoHash}\label{final state invariant if no hash queries}
Let $B\Sign_{\sk}$ be the blinded signing oracle for the Lamport OTS and let $\ket{\psi_0}$ be the adversary's state before the $\Sign$ query. If there are no hash queries, then after making at most one $\Sign$ query the adversary's state $\ket{\psi_1} = B\Sign_{\sk}\ket{\psi_0}$ is completely in the image of the invariant projector $P_S$ defined in \cref{eq:invariant projector}. That is,
\begin{equation}
    P_SB\Sign_{\sk}\ket{\psi_0} = B\Sign_{\sk}\ket{\psi_0}.
\end{equation}
\end{restatable}

\begin{restatable}{lemma}{LemComP}\label{commutator of ps and uh}
The invariant projector $P_S$ defined in \cref{eq:invariant projector} and the random oracle unitary $U_h$ in the \QWorld, see \cref{eq:Uh}, approximately commute, i.e.,
\begin{equation}\label{InvPG:Bound}
     \bigl\| [U_h,P_S]
     \bigl\|_\infty
     \leq \delta_L (n),
\end{equation}
where \[\delta_L (n) = \frac{32l}{2^{n/2}}\]
is negligible in $n$.
\end{restatable}

Next, we use the above lemmas to prove our main results. In the following sections, we analyze the situation where the adversary makes $q_0$ hash queries before the $\Sign$ query and $q_1$ hash queries after, maximizing the resulting bound under the condition that $q_0+q_1=q$.

\subsection{Hash queries before \texorpdfstring{$\Sign$}{Sign} query}

In this section, we study the impact of hash queries before $\Sign$ query on the secret key register $S$. Our main goal is to show that, for a moderate number of queries to the random oracle, no adversary can learn a significant amount of information about the secret key. Therefore, she cannot produce a valid forgery except with small probability. 

\begin{table}[ht]
    \centering
    \begin{tabular}{c|l}
        Register & Meaning \\ \hline
        $X$ & adversary's input \\
        $Y$ & adversary's output \\
        $M$ & $\Sign$ query input \\
        $\Sigma$ & $\Sign$ query output \\
        $E$ & adversary's internal workspace \\
        $S$ & secret key
    \end{tabular}
    \caption{Registers used in the analysis.} 
    \label{tab:registers}
\end{table}

Let $\ket{\psi}_{XYM\Sigma E}$ be adversary's initial state before any queries (see \cref{tab:registers} for a summary of registers and their roles). Before any query is performed, the whole secret key register $S$ is in the uniform superposition state $\ket\Phi\xp{2l}$. Assume the adversary $\A_0$ queries the \textrm{random oracle} %
$q_0$ times before querying the signing oracle.
If $V_{XYE}^i$ denotes the unitary she performs after the $i$-th query, the final adversary state after $q_0$ hash queries is
\begin{equation}
    \ket{\psi_0}_{XYM\Sigma ES}
  = V^{q_0}_{XYE}
    (U_h)_{XYS}
    V^{q_0-1}_{XYE}\cdots
    V^2_{XYE}
    (U_h)_{XYS}
    V^1_{XYE}
    (U_h)_{XYS}
    \ket{\psi}_{XYM\Sigma E}
    \ket{\Phi}\xp{2l}_S
    \label{eq:before hash query}
\end{equation}
where $U_h$ is the random oracle unitary used to answer hash queries. The following lemma shows that the secret key registers of this state are still close to being in uniform superposition.

\begin{lemma}\label{lem:hash-before-sign}
In the \QWorld, without querying the $B\Sign$ oracle, hash queries leave the state of the secret key registers approximately unchanged:
\begin{align*}
    \norm*{
    \Phi\xp{2l}_S
    \ket{\psi_0}_{XYM\Sigma ES} - 
    \ket{\psi_0}_{XYM\Sigma ES}
    }_2
    \leq 2lq_0\epsilon_L(n).
\end{align*}
\end{lemma}

\begin{proof}
We want to show that after $q_0$ hash queries, the state of the secret key register $S$ is still approximately in the uniform superposition state $\ket \Phi^{\otimes 2l}$. Let us abbreviate the overall unitary in \cref{eq:before hash query} by $W_{XYES}$. Since the only operations in $W_{XYES}$ that act on the $S$ register are the hash queries $U_h$, and they are in fact controlled by the $S$ register, we have $W_{XYES} \Phi\xp{2l}_S = W_{XYES}$. Using this, we get
\begin{align}
  & \norm*{
    \Phi\xp{2l}_S
    \ket{\psi_0}_{XYM\Sigma ES} - 
    \ket{\psi_0}_{XYM\Sigma ES}
    }_2\nonumber\\
 &= \norm*{
    \Phi\xp{2l}_S
    W_{XYES}
    \ket{\psi}_{XYM\Sigma E}
    \ket{\Phi}\xp{2l}_S
  - W_{XYES}
    \Phi\xp{2l}_S
    \ket{\psi}_{XYM\Sigma E}
    \ket{\Phi}\xp{2l}_S
    }_2 \label{slightlyInv:lamport}\\
 &= \norm*{
    \bigl[
        \Phi\xp{2l}_S,
        W_{XYES}
    \bigr]
    \ket{\psi}_{XYM\Sigma E}
    \ket{\Phi}\xp{2l}_S
    }_2\label{defComOp:lamport}\\
  & \leq
    \norm*{
    \bigl[
        \Phi\xp{2l}_S,
        W_{XYES}
    \bigr]
    }_\infty
    \underbrace{
    \norm*{
        \ket{\psi}_{XYM\Sigma E}
        \ket{\Phi}\xp{2l}_S
    }_2}_{=1}
    \label{sub-multiplicativePpty:lamport}\\
 &= \norm*{
    \bigl[
        \Phi\xp{2l}_S,
        V^{q_0}_{XYE}(U_h)_{XYS}
        V^{q_0-1}_{XYE}\cdots
        V^2_{XYE}(U_h)_{XYS}
        V^1_{XYE}(U_h)_{XYS}
    \bigr]
    }_\infty\\
  & \leq
    q_0
    \norm*{
    \bigl[
        \Phi\xp{2l}_S,
        (U_h)_{XYS}
    \bigr]
    }_\infty
  + \sum_{i=1}^{q_0}
    \norm*{
    \bigl[
        \Phi\xp{2l}_S,
        V^i_{XYE}
    \bigr]
    }_\infty,
    \label{subadditivity:lamport}
\end{align}
where \cref{sub-multiplicativePpty:lamport} follows from the definition of the operator norm and the last inequality follows from \cref{lem:commutators}.

The first term in \cref{subadditivity:lamport} can be bounded as follows:
\begin{align*}
    \norm*{
    \bigl[
        \Phi\xp{2l}_S,
        (U_h)_{XYS}
    \bigr]
    }_\infty 
  & \leq
    \sum_{\substack{i\in\set{1,\dotsc,l}\\j\in\set{0,1}}}
    \norm[\Big]{
    \bigl[
        \Phi_{S_i^j},
        (U_h)_{XYS}
    \bigr]
    }_\infty 
    \leq 2 l \epsilon_L(n),
\end{align*}
which follows by first applying \cref{lem:commutators} and then \cref{Commutator argument1}.
Since $\Phi\xp{2l}_S$ and $V^i_{XYE}$ act on different registers, they commute and the second term in \cref{subadditivity:lamport} vanishes. Hence
\begin{equation}
    \norm*{
        \Phi\xp{2l}_S
        \ket{\psi_0}_{XYM\Sigma ES} - 
        \ket{\psi_0}_{XYM\Sigma ES}
    }_2
    \leq 2 l q_0 \epsilon_L(n)
    \label{bound hash queries before sign query}
\end{equation}
as desired.
\end{proof}

Recall from \cref{eq:UhPhi} that $\epsilon_L(n) = 4/2^{n/2}$ is negligible in $n$. Since $l$ is constant, the magnitude of $2lq_0\epsilon_L(n)$ is determined only by the number of queries $q_0$. The bound in \cref{bound hash queries before sign query} is negligible for any adversary making $2^{c n}$ queries to the random oracle when $c < 1/2$. Therefore 
$\Phi\xp{2l}_S \ket{\psi_0}_{XYM\Sigma ES}$ and $\ket{\psi_0}_{XYM\Sigma ES}$ are close, which means that hash queries before $\Sign$ query do not significantly change the secret key register. Equivalently, it means that the adversary learns almost no information about the secret key.

\subsection{Query to the signing oracle}\label{sec:sign-query}

Now that we have control over the advantage an adversary can gain from making hash queries before the sign query, we need to analyze the possible advantage from hash queries after the sign query and bound the overall success probability using \cref{lem:hash-before-sign}.

A crucial property of the Lamport OTS when analyzing classical security is that for all messages $m$ that have not been queried, there exists an index $j$ such that $s_j^{m_j}$ is hidden from the adversary by the one-wayness of the used hash function. In blind-unforgeability (for classical adversaries), this property holds for all \emph{blinded messages}. In the setting of quantum queries, we have to track this property in superposition while the adversary is making hash queries after the sign query. As this is complicated by the ``for all''-quantifier, we begin by analyzing the case where the adversary makes no hash queries after the sign query to ease the reader into our proof technique.

The discussion in this section does not concern the random oracle, so we absorb the random oracle query registers $XY$ into $E$ for the purpose of this section. In the $1$-$\BlindForge$ game, an adversary $\A$ is allowed to query the $\Sign$-oracle at most once to produce a valid forged message-signature pair $(m^*,\sigma^*)$. To analyze the interaction between $\A$ and the signing oracle, we will break it into the following steps:
\begin{equation}\label{Interaction:Bsign-Adver}
    \ket{\psi_0}_{M\Sigma BES}
    \xmapsto{B\Sign_{\sk}}
    \ket{\psi_1}_{M\Sigma BES}
    \xmapsto{U_{M\Sigma E}}
    \ket{\psi_2}_{M\Sigma BES}
    \xmapsto{\bra{m^*}_M}
    \ket{\psi_3(m^*)}_{\Sigma BES}
    \xmapsto{\bra{\sigma^*}_\Sigma}
    \ket{\psi_4(m^*,\sigma^*)}_{BES}.
\end{equation}
They correspond to applying the $\Sign$-oracle and an arbitrary unitary $U_{M\Sigma E}$, followed by measuring the message and signature registers $M$ and $\Sigma$. Let us now analyze these steps in more detail and write down the corresponding quantum states.

First, $\A$ prepares her input state as an arbitrary superposition of messages:
\begin{align}
    \ket{\psi_0}_{M\Sigma BES}
  = \left(\sum_{m\in\{0,1\}^l}
    \sum_{\sigma\in(\{0,1\}^n)^l}
    \sum_{b\in\{0,1\}}
    \kappa_{m\sigma b}
    \ket{m}_M
    \ket{\sigma}_{\Sigma}
    \ket{b}_B \ket{\alpha_{m\sigma b}}_E\right)
    \otimes
    \left(\ket{\Phi}^{\otimes 2l}\right)_S\label{eq:psi0-L}
\end{align}
where the %
$B$ will indicate whether the message is blinded or not ($\ket{1}_B$ for blinded and $\ket{0}_B$ for un-blinded).

Next, the adversary supplies this to the $\Sign$ oracle which produces the following signed state:
\begin{align}
   \ket{\psi_1}_{M\Sigma BES} =  
   B \Sign_{\sk}
   \ket{\psi_0}_{M\Sigma BES} =
   \ket{\psi_1^1}_{M\Sigma BES} +
   \ket{\psi_1^0}_{M\Sigma BES}
\end{align}
where the superscripts $1$ and $0$ refer to blinded ($B$) and un-blinded ($B^c$) messages, respectively:
\begin{align*}
    \ket{\psi_1^1}_{M\Sigma BES} &= 
        \sum_{m\in B}
        \sum_{\sigma\in(\{0,1\}^n)^l}
        \kappa_{m\sigma 1}
        \ket{m}_M
        \ket{\sigma}_{\Sigma}
        \ket{1}_B\ket{\alpha_{m\sigma1}}_E
        \ket{\Phi}^{\otimes 2l}_S, \\
    \ket{\psi_1^0}_{M\Sigma BES} &= 
        \sum_{m\in B^c}
        \sum_{\sigma\in(\{0,1\}^n)^l}
        \frac{1}{2^{nl/2}}
        \sum_{s\in(\{0,1\}^n)^l}
        \kappa_{m\sigma 0}
        \ket{m}_M
        \ket{\sigma \oplus s}_\Sigma
        \ket{0}_B
        \ket{\alpha_{m\sigma0}}_{E}
        \ket{\Omega(s,m)}_S,
\end{align*}
where $m = m_1 \dots m_l$, $\sigma = \sigma_1 \dots \sigma_l$, and 
\begin{equation}
    \ket{\Omega(s,m)}_S =
    \ket{s^{m_1}_1}_{S^{m_1}_1}\cdots
    \ket{s^{m_l}_l}_{S^{m_l}_l}
    \ket{\Phi}_{S^{\bar{m}_1}_1}\cdots
    \ket{\Phi}_{S^{\bar{m}_l}_l}.
    \label{eq:Omega}
\end{equation}

Once the adversary $\A$ gets the signed state $\ket{\psi_1}_{M\Sigma BES}$, she performs some operations with the intention of producing a forgery message $m^*$. Intuitively, those operations can be considered as applying some arbitrary unitary $U_{M\Sigma E}$ to $\ket{\psi_1}_{M\Sigma BES}$. Let us denote the resulting state by
\begin{align*}
    \ket{\psi_2}_{M\Sigma BES} = U_{M\Sigma E} \ket{\psi_1}_{M\Sigma BES}.
\end{align*}

Then $\A$ measures the message register $M$, which yields outcome $m^* \in \set{0,1}^l$.
After the measurement, the state $ \ket{\psi_2}_{M\Sigma BES}$ collapses to the (unnormalized) state
\begin{align*}
     \ket{\psi_3(m^*)}_{\Sigma BES} 
   &=\bra{m^*}_M\ket{\psi_2^1}_{M\Sigma BES} +
     \bra{m^*}_M\ket{\psi_2^0}_{M\Sigma BES}\\
   &=\ket{\psi_3^1(m^*)}_{\Sigma BES} +
     \ket{\psi_3^0(m^*)}_{\Sigma BES}
\end{align*}
where
\begin{align*}
    \ket{\psi_3^1(m^*)}_{\Sigma BES}
  &=\sum_{m\in B}
    \sum_{\sigma\in(\{0,1\}^{n})^l}\kappa_{m\sigma 1}
    \bra{m^*}_M U_{M\Sigma E}
    \ket{m}_M
    \ket{\sigma}_{\Sigma}
    \ket{1}_B\ket{\alpha_{m\sigma 1}}_E
    \ket{\Phi}^{\otimes 2l}_S, \\
    \ket{\psi_3^0(m^*)}_{\Sigma BES}
  &=\sum_{m\in B^c}
    \sum_{\sigma\in(\{0,1\}^n)^l}
    \frac{1}{2^{nl/2}}
    \sum_{s\in(\{0,1\}^n)^l}\kappa_{m\sigma 0}
    \bra{m^*}_M U_{M\Sigma E}
    \ket{m}_M
    \ket{\sigma \oplus s}_\Sigma
    \ket{0}_B
    \ket{\alpha_{m\sigma 0}}_E
    \ket{\Omega(s,m)}_S.
\end{align*}

Having obtained $m^*$, the purpose of the adversary $\A$ is to produce a forged signature $\sigma^*$ that corresponds to $m^*$. To that end, she measures the signature register $\Sigma$ of $\ket{\psi_3(m^*)}_{\Sigma BES}$, getting outcome $\sigma^* \in \of{\set{0,1}^n}^l$.
The (unnormalized) post-measurement state is
\begin{align}
    \ket{\psi_4(m^*,\sigma^*)}_{BES} 
 &= \bra{\sigma^*}_\Sigma\ket{\psi_3^1(m^*)}_{\Sigma BES} +
    \bra{\sigma^*}_\Sigma\ket{\psi_3^0(m^*)}_{\Sigma BES} \label{eq:final state} \\
 &= \ket{\psi_4^1(m^*,\sigma^*)}_{BES} +
    \ket{\psi_4^0(m^*,\sigma^*)}_{BES} \nonumber
\end{align}
where
\begin{align*}
    \ket{\psi_4^1(m^*,\sigma^*)}_{BES}
  &=\sum_{m\in B}
    \sum_{\sigma\in(\{0,1\}^n)^l}\kappa_{m\sigma 1}
    \bra{m^*}_M
    \bra{\sigma^*}_\Sigma
    U_{M\Sigma E}
    \ket{m}_M
    \ket{\sigma}_{\Sigma}
    \ket{1}_B\ket{\alpha_{m\sigma 1}}_E
    \ket{\Phi}^{\otimes 2l}_S, \\
    \ket{\psi_4^0(m^*,\sigma^*)}_{BES}
  &=\sum_{m\in B^c}
    \sum_{\sigma\in(\{0,1\}^n)^l}
    \frac{1}{2^{nl/2}}
    \sum_{s\in(\{0,1\}^n)^l}
    \kappa_{m\sigma 0}
    \bra{m^*}_M
    \bra{\sigma^*}_\Sigma
    U_{M\Sigma E}
    \ket{m}_M
    \ket{\sigma \oplus s}_\Sigma
    \ket{0}_B
    \ket{\alpha_{m\sigma 0}}_E
    \ket{\Omega(s,m)}_S.
\end{align*}
For the sake of simplicity, let us rewrite $\ket{\psi_4^0}_{BES}$ as follows:
\begin{align}
    \ket{\psi_4^0(m^*,\sigma^*)}_{BES}
  &=\sum_{m\in B^c}
    \frac{1}{2^{nl/2}}
    \sum_{s\in(\{0,1\}^n)^l}
    \ket{\eta(m,s)}_{BE}
    \ket{\Omega(s,m)}_S
    \label{eq:psi04}
\end{align}
where only $\ket{\eta(m,s)}_{BE}$ depends on $m^*$ and $\sigma^*$:
\begin{align*}
    \ket{\eta(m,s)}_{BE} 
 &= \sum_{\sigma\in(\{0,1\}^n)^l}
    \kappa_{m\sigma 0}
    \bra{m^*}_M 
    \bra{\sigma^*}_{\Sigma}
    U_{M\Sigma E}
    \ket{m}_M
    \ket{\sigma \oplus s}_\Sigma
    \ket{0}_B
    \ket{\alpha_{m\sigma 0}}_E.
\end{align*}
Finally, the adversary $\A$ outputs the forged message-signature pair $(m^*,\sigma^*)$. The probability of producing this pair is $\norm{\ket{\psi_4^0(m^*,\sigma^*)}_{BES}}^2$.

The next step is to analyse the probability that $\A$'s forgery candidate $(m^*,\sigma^*)$ is correct.
For that purpose, we consider two cases. The first case, namely when $m^*\notin B$, is trivial since then $\A$ has lost the $\BlindForge$ experiment because $m^*$ must be blinded by definition. The rest of this section is devoted to analyzing the second case.

If $m^*\in B$, the forged message $m^*$ has not been signed 
since the blinded signing oracle signs only un-blinded messages. Hence, for any message $m\notin B$, there exists at least one index $i\in\{1,\dotsc,l\}$ such that $m_i\neq m^*_i$. 
This implies that for some index $i^* \in \{1,\dotsc,l\}$ the register $S^{m^*_{i^*}}_{i^*}$ has not been used for the signature of the adversary's queried message and is therefore still in the uniform superposition state $\ket\Phi$. Note that this holds only in superposition over $m$. Indeed, $i^*$ depends on $m$ and is in general different for each term of the superposition.

We know that the secret key register $S$ consists of $2l$ $n$-qubit registers out of which only $l$ are used for the signature procedure while the other $l$ are still in the uniform superposition $\ket{\Phi}$. Despite the secret key being in superposition, we want to track the invariant part of the secret key and show that some of the secret key sub-registers relevant for the forged signature satisfy this invariant and are thus unknown to the adversary, so it is unlikely that the adversary would have used the correct secret key sub-register to produce the forged signature.

For that purpose, we analyze a modified $\BlindForge$ experiment, where an additional measurement is performed on the secret key register after the adversary has output their forgery, but before the secret key register is measured to actually sample the secret key as required in the \QWorld. This additional measurement was defined in \cref{eq:Q} and we will refer to it as the \emph{$Q$-measurement}. Since it has few outcomes, its effect on the adversary's winning probability is limited and can be bounded by the pinching lemma.

If the $Q$-measurement yields outcome $i^* \in \{1,\dotsc,l\}$, then  the secret key sub-register $S_{i^*}^{m_{i^*}}$ is in uniform superposition, and the adversary is bound to fail as $\sigma^*$ is independent of the secret key string $s_{i^*}^{m_{i^*}}$ (the result of measuring $S_{i^*}^{m_{i^*}}$).

It remains to analyze the outcome $l+1$ that corresponds to the projector $Q_{l+1}^m = (\Phi^\perp)\xp{l}$, see \cref{eq:Q}, where $\Phi^\perp = \1 - \proj{\Phi}$ is the projector onto the orthogonal complement of $\ket{\Phi}$.
The final adversary state after the measurement, see \cref{eq:final state}, contains both blinded and un-blinded terms. If we apply $\Phi^\perp$ to any secret key register of the blinded term $\ket{\psi_4^1(m^*,\sigma^*)}_{BES}$, we get $0$ since all secret key sub-registers are in state $\ket{\Phi}$.

For the rest of our analysis, we fix the message $m^*$ and focus on the un-blinded term $\ket{\psi_4^0(m^*,\sigma^*)}_{BES}$. Given that for each $m \notin B$ there is at least one index $i \in \{1,\dotsc,l\}$ such that $m_i \neq m^*_i$, we define
\begin{align*}
    i(m) = \min \{j \in \set{1,\dotsc,l} \mid m_j\neq m^*_j\} 
\end{align*}
as the smallest index for which $m\neq m^*$. Intuitively, it is the first sub-register of $S$ that still remains in uniform superposition. In the following, let $S(m) := S_1^{m_1} \dotsb S_l^{m_l}$ and recall from \cref{eq:psi04,eq:Omega} that the un-blinded term is given by
\begin{align}
    \ket{\psi_4^0(m^*,\sigma^*)}_{BES}
  &=\sum_{m\in B^c}
    \frac{1}{2^{nl/2}}
    \sum_{\gamma\in(\{0,1\}^n)^l}
    \ket{\eta(m,\gamma)}_{BE}
    \ket{s^m}_{S(m)}
    \ket{\Phi}\xp{l}_{S(\bar{m})}.
\end{align}
We want to split the first sum into $l$ parts, one for each value of $i(m)$, so that we can easily evaluate $(\Phi^\perp)\xp{l}_{S(\bar m)} \ket{\psi_4^0(m^*,\sigma^*)}_{BES}$. For that purpose, we define $B^c_j = \{m \in B^c \mid i(m) = j\}$ and note that $\bigcup _{j=1}^l B^c_j = B^c$.

We can now rewrite $\ket{\psi_4^0(m^*,\sigma^*)}_{BES}$ as
\begin{align}
    \ket{\psi_4^0(m^*,\sigma^*)}_{BES}
    &=\sum_{j=1}^l \sum_{m \in B^c_j}
    \frac{1}{2^{nl/2}}
    \sum_{s\in(\{0,1\}^n)^l}
    \ket{\eta(m,s)}_{BE}
    \ket{s^{m}}_{S(m)}
    \ket{\Phi}^{\otimes l}_{S(\bar{m})}\nonumber\\
    &=
    \sum_{j=1}^l\ket{\hat
    \eta(m^*,\sigma^*,j)}_{BES_{\{(j,m^*_j)\}^c}}
    \ket{\Phi}_{S^{m^*_j}_j},\label{B^cSplit}
\end{align}
where we absorbed all registers except for $S^{m^*_j}_j$ into the first system. The remaining register $S^{m^*_j}_j$ is still in the uniform superposition $\ket{\Phi}$ since $j = i(m)$ is the smallest index such that $m_j \neq m_j^*$, meaning that $\bar{m}_j = m_j^*$ and thus $S^{m^*_j}_j = S^{\bar{m}_j}_j$.
Applying $Q_{l+1}$ onto the $l$ sub-registers $S(m^*)$ of the register $S$ in \cref{B^cSplit} gives
\begin{align}
    \of[\big]{Q_{l+1}^{m^*}}_{S^{m^*_1}_1\cdots S^{m^*_l}_l}
    \ket{\psi_4^0(m^*,\sigma^*)}_{BES}
  = (\Phi^\perp)\xp{l}_{S^{m^*_1}_1\cdots S^{m^*_l}_l}
    \left(\sum_{j=1}^l
    \ket{\hat\eta(m^*,\sigma^*,j)}_{BES_{\{(j,m^*_j)\}^c}} \ket{\Phi}_{S^{m^*_j}_j}\right)
  = 0,\label{eq:no-l+1}
\end{align}
which vanishes because, for each $j$, the register $S^{m^*_j}_j$ is in state $\ket{\Phi}$ and $\Phi^\perp \ket{\Phi} = 0$. Hence, the situation where none of the secret key sub-registers relevant for the verification of the forged signature $\sigma^*$ is in state $\ket{\Phi}$ can ever occur.

Now, we execute the last part of the $\BlindForge$ experiment which consists of checking the correctness of the forged signature $\sigma^*$. For this purpose, we perform a computational basis measurement on the entire secret key register $S$ to sample the strings $s_i^j$. 

We recall that in the $\BlindForge$ experiment, there is no partial measurement. Therefore, we first evaluate the success probability of $\A$ in case of the modified $\BlindForge$ experiment (MBF) in which we performed a partial measurement. Afterwards, we use the Pinching lemma (\cref{Pinching}) to deduct the success probability of the adversary in the real $\BlindForge$ experiment from the modified experiment.

Knowing that after applying the partial measurement, at least one of the secret key sub-registers relevant to $\sigma^*$ is still in the state $\ket{\Phi}$, the probability that the adversary $\A$ used the right $s_i^{m_i^*}$ to produce $\sigma^*$ is $1/2^n$. In addition, given that the state $ \ket{\psi_4(m^*,\sigma^*)}_{BES}$ is unnormalized, the success probability of $\A$ in producing a fixed valid forged message-signature pair $(m^*,\sigma^*)$ after applying the partial measurement is
\begin{align}
    &\Pr_{QI,MBF}
    \sof[\big]{\textnormal{Success $\wedge$
    $\A$ outputs $(m^*,\sigma^*)$}}\nonumber\\
    &=\sum_{j=1}^l \Pr
    \sof[\big]{\textnormal{Success $\wedge$
    $\A$ outputs $(m^*,\sigma^*)$
    $\wedge$
    $Q$-measurement returns $j$}}\nonumber\\
  &=\frac{1}{2^n}\bigl\| 
    \ket{\psi_4(m^*,\sigma^*)}_{EBS}\bigr\|^2.\label{eq:mod-succ-no-hqasq}
\end{align}
Here, the sum over the outcomes of the $Q$-measurement is restricted to $1,\dotsc,l$ as the outcome $l+1$ never occurs by \cref{eq:no-l+1}. Since we made an ($l+1$)-outcome partial measurement on the secret key register previously, by the Pinching (see \cref{Pinching}), this partial measurement can only increase the success probability of $\A$  in the real $\BlindForge$ experiment by at most $l+1$. Thus, the probability that the adversary  outputs a valid forged message-signature pair $(m^*,\sigma^*)$ with respect to the real $\BlindForge$ experiment is upper bounded by
\begin{align}
    \Pr_{QI,\BlindForge}
    \sof[\big]{\textnormal{Success $\wedge$
    $\A$ outputs $(m^*,\sigma^*)$}}
    \leq \frac{l+1}{2^n}\bigl\|
    \ket{\psi_4(m^*,\sigma^*)}_{BES}\bigr\|^2.
\end{align}
Therefore, the success probability of $\A$ in producing a valid forged message-signature pair $(m^*,\sigma^*)$ is
\begin{align}
    \Pr_{QI}
    \sof[\big]{\textnormal{$\A$ wins $\BlindForge$}}
    &=\sum_{m^*,\sigma^*}    
    \Pr_{\BlindForge}
    \sof[\big]{\textnormal{Success $\wedge$
    $\A$ outputs $(m^*,\sigma^*)$}} 
    \leq\frac{l+1}{2^n}.
\end{align}
We conclude that the same holds in the \RealWorld, up to a difference as permitted by \cref{lem:ind-r-v-qi} with $w=2$,
\begin{align}
    \Pr
    \sof[\big]{\textnormal{$\A$ wins $\BlindForge$}}
    \leq\frac{l+1}{2^n}+12 l^2\cdot 2^{-n}.
\end{align}
Hence, the success probability of the adversary $\A$ in winning the $\BlindForge$ experiment game is at most $(l+1)/2^n$ which is negligible since $l$ is polynomial in $n$, and $n$ is large enough. We conclude that $\Sign$ query does not help the adversary to get significant information about the secret key.

\subsection{Hash queries after \texorpdfstring{$\Sign$}{Sign} query}\label{sec:HQ-SQ-ltos}
In this section, we analyse the adversary's \textit{hash queries after $\Sign$ query} to bound the success probability that an adversary with a given number of queries can achieve in the $\BlindForge$ game and thus prove \cref{thm:lamport}.
In this case it is not obvious how to track the invariant of the secret key, i.e. the fact that there is at least one unused part of the secret key that is relevant for the forged signature. Therefore we use a special projector $P_S$ defined in \cref{eq:invariant projector} that projects onto the subspace of the secret key register that is consistent with a single blinded sign query and no hash queries. If the final adversary state after producing the forgery candidate is in the image of $P_S$, then the outcome $l+1$ corresponding to the situation when \textit{none of the secret key sub-registers useful for the forged signature is in state $\ket\Phi$} can never occur, according to \cref{Orthogonality argument}. We thus want to show that adversary's final state is approximately in the range of $P_S$.

If there are no hash queries before the $\Sign$ query, then from \cref{final state invariant if no hash queries} the adversary state after the $\Sign$ query remains completely in the range of $P_S$, which means that the outcome $l+1$ cannot occur. That is,
\[P_S\ket{\psi_1} = P_SB\Sign_{\sk}\ket{\psi_0} = B\Sign_{\sk}\ket{\psi_0} = \ket{\psi_1}\]
where $\ket{\psi_0}$ and $\ket{\psi_1}$ are respectively the adversary state immediately before and after the $\Sign$ query.

Now, assuming there are hash queries before the $\Sign$ query, since the projector $P_S$ and the random oracle unitary $U_h$ approximately commute by \cref{commutator of ps and uh}, it follows that hash queries before $\Sign$ query give no significant information to the adversary about the invariant of the secret key register.

Suppose there are hash queries after the $\Sign$ query and examine in detail what happen in this case. From the previous case, we know that the adversary's state directly after the $\Sign$ query is $\ket{\psi_1}_{M\Sigma XYES}$. Just like for hash queries before the $\Sign$ query, suppose that the adversary makes $q_1$ hash queries after querying the signing oracle. Let $(W^i_{XYE})_{i=1,\dotsc,q_1}$ be the unitaries applied between hash queries. Then, let 
\[\ket{\psi_1'}_{M\Sigma XYES} = 
(U_h)_{XYS}
W^{q_1}_{XYE}
(U_h)_{XYS}
W^{q_1-1}_{XYE}\cdots
W^2_{XYE}
(U_h)_{XYS}
W^1_{XYE}
\ket{\psi_1}_{M\Sigma XYES}
\]
be the adversary's state after $q_1$ hash queries and before performing some unitary operations $U_{M\Sigma E}$ on the post hash queried state or any measurement leading to the forgery candidate. For simplicity, we set 
\[T = (U_h)_{XYS}
W^{q_1}_{XYE}
(U_h)_{XYS}
W^{q_1-1}_{XYE}\cdots
W^2_{XYE}
(U_h)_{XYS}
W^1_{XYE}.\]
We now prove the following lemma.
\begin{lemma}\label{lem:hash-after-sign}
In the \QWorld, the state right before the adversary's measurement determining the forgery is applied is approximately in the range of $P_S$, i.e.
\begin{align}
  \bigl\|
  P_S\ket{\psi_1'}_{M\Sigma XYES} - 
  \ket{\psi_1'}_{M\Sigma XYES}
  \bigr\|_2  
  \leq q_1\delta_L(n) + 4lq_1\epsilon_L(n) 
  =q_1(\delta_L(n) + 4l\epsilon_L(n)) 
\end{align}
\end{lemma}
\begin{proof}
To see how much those hash queries affect the entire secret register, we compute the difference norm between the states $P_S\ket{\psi_1'}_{M\Sigma XYES}$ and $\ket{\psi_1'}_{M\Sigma XYES}$ to see how closed they are. We have:
\begin{align}
    &\norm[\big]{
        P_S\ket{\psi_1'}_{M\Sigma XYES} - 
        \ket{\psi_1'}_{M\Sigma XYES}
    }_2\nonumber\\
  &=\norm[\Big]{
        P_ST
        \ket{\psi_1}_{M\Sigma XYES} + 
        TP_S\ket{\psi_1}_{M\Sigma XYES}
      - TP_S\ket{\psi_1}_{M\Sigma XYES}
      - T\ket{\psi_1}_{M\Sigma XYES}
    }_2 \label{slightly invarianta}\\
  &\leq\norm[\big]{
        [P_S,T]
    }_\infty 
    \underbrace{
    \norm[\big]{
        \ket{\psi_1}_{M\Sigma XYES}
    }_2
    }_{=1} +
    \underbrace{
        \norm{T}_\infty
    }_{=1}
    \norm[\big]{
        P_S\ket{\psi_1}_{M\Sigma XYES} - 
        \ket{\psi_1}_{M\Sigma XYES}
    }_2
    \label{definition of commutator operatora}\\
  &=
    \norm[\big]{
        [P_S,T]
    }_\infty
    +
    \norm[\big]{
        P_S\ket{\psi_1}_{M\Sigma XYES} -
        \ket{\psi_1}_{M\Sigma XYES}
    }_2\label{sub-multiplicative property of the norma}\\
  &\leq
    q_1 \norm[\big]{
        \sof*{P_S,(U_h)_{XYS}}
    }_\infty
    +
    \sum_{i=1}^{q_1}
    \norm[\big]{
        \underbrace{
            \sof[\big]{P_S,W^i_{XYE}}
        }_{=0}
    }_\infty
    +
    \norm[\big]{
        P_S\ket{\psi_1}_{M\Sigma XYES} -
        \ket{\psi_1}_{M\Sigma XYES}
    }_2
    \label{commutator algebrametric property}\\
    &\leq
    q_1 \norm[\big]{
        \sof*{(U_h)_{XYS},P_S}
    }_\infty
    +
    \norm[\big]{
        P_S\ket{\psi_1}_{M\Sigma XYES} -
        \ket{\psi_1}_{M\Sigma XYES}
    }_2
    \label{commutator hash unitary and P_S}\\
    &\leq
    q_1 \delta_L(n) +
    \norm[\big]{
        P_S\ket{\psi_1}_{M\Sigma XYES} -
        \ket{\psi_1}_{M\Sigma XYES}
    }_2
 \end{align}
whereby, \cref{definition of commutator operatora,sub-multiplicative property of the norma} come respectively from the definition of commutator and the definition of the operator norm. On top of that, $\bigl\|\ket{\psi_1}_{M\Sigma XYES}\bigr\|_2 = 1$ because $\ket{\psi_1}_{M\Sigma XYES}$ is normalized. \Cref{commutator algebrametric property} follows from \cref{lem:commutators} and from the fact that $P_S$ and $W^i_{XYE}$ commute. Finally, the first term of the right-hand side of the last equation follows from \cref{commutator of ps and uh}.

To evaluate the second term of the right-hand side of the latter equation, we will make use of the following bound on the operator norm %
from \cref{lem:hash-before-sign}:
 \[\bigg\|
 \ket{\psi_0}_{M\Sigma XYEB}
 - \Phi\xp{2l}_S
 \ket{\psi_0}_{M\Sigma XYEB}
 \biggr\|_\infty \leq 2lq_1\epsilon_L(n).\]
We have:
 \begin{align*}
    \bigg\|
    P_S\ket{\psi_1}_{M
    \Sigma XYES}
    - \ket{\psi_1}_{M\Sigma XYEBS}
    \biggr\|_2
    &=
    \bigg\|
    P_SB
    \Sign_{\sk}
    \ket{\psi_0}_{M\Sigma XYEBS} -
   P_SB
   \Sign_{\sk} 
   \Phi\xp{2l}_S
   \ket{\psi_0}_{M\Sigma XYEBS}\\
   &\quad
   + P_SB
   \Sign_{\sk}
   \Phi\xp{2l}_S
   \ket{\psi_0}_{M\Sigma XYEBS} -
   B
   \Sign_{\sk}
   \Phi\xp{2l}_S
   \ket{\psi_0}_{M\Sigma XYEBS}\\
   &\quad
   + B
   \Sign_{\sk}
   \Phi\xp{2l}_S
   \ket{\psi_0}_{M\Sigma XYEBS} ~-
    B
   \Sign_{\sk}
   \ket{\psi_0}_{M
   \Sigma XYESS}
   \biggr\|_2.
\end{align*}
We will use the triangle inequality to split this into three terms and then bound each of them separately.

We can bound the first term as follows:
\begin{align*}
    &\bigg\|
    P_SB
  \Sign_{\sk}
  \ket{\psi_0}_{M
  \Sigma XYEBS} ~-
  P_SB
  \Sign_{\sk}
  \Phi\xp{2l}_S
  \ket{\psi_0}_{M\Sigma XYEBS}\biggr\|_2\\
  &\leq
    \underbrace{
    \bigg\|
    P_SB
  \Sign_{\sk}
  \biggr\|_\infty }_{= 1}
  \underbrace{
   \bigg\|
  \ket{\psi_0}_{M
  \Sigma XYEBS} ~-
  \Phi\xp{2l}_S
  \ket{\psi_0}_{M\Sigma XYEBS}
  \biggr\|_{2}}_{\leq 2lq_1\epsilon_L(n)}\\
  &\leq
 2lq_1\epsilon_L(n)
\end{align*}
where the first inequality follows by the definition of the operator norm and the final upper bound results from \cref{lem:hash-before-sign}.

Next, we bound the second term. It is exactly the same as the expression in \cref{P_S leave invariant psi1}, thus,
\begin{align*}
    \bigg\|
  P_SB
  \Sign_{\sk}
  \Phi\xp{2l}_S
  \ket{\psi_0}_{M\Sigma XYEBS}~ -
  B
  \Sign_{\sk}
  \Phi\xp{2l}_S
 \ket{\psi_0}_{M\Sigma XYEBS}
  \biggr\|_2 = 0.
\end{align*}
Finally, looking at the third term, we observe that it is very similar to the first term. Thus, they have the same bound. 
Therefore, 
\begin{align*}
    \bigg\|
    P_S\ket{\psi_1}_{M
   \Sigma XYES}~ - 
    \ket{\psi_1}_{M\Sigma XYEBS}
    \biggr\|_2
    \leq
     2lq_1\epsilon_L(n) + 0 + 2lq_1\epsilon_L(n)
    =  4lq_1\epsilon_L(n) 
\end{align*}
and
\begin{align}
  \bigl\|
  P_S\ket{\psi_1'}_{M\Sigma XYES} - 
  \ket{\psi_1'}_{M\Sigma XYES}
  \bigr\|_2  
  \leq q_1\delta_L(n) + 4lq_1\epsilon_L(n) 
  =q_1(\delta_L(n) + 4l\epsilon_L(n)).
  \label{state after hash queries invariant by P_S}
\end{align}
\end{proof}
As long as $q_1=o(2^{n/2})$, the bound in \cref{lem:hash-after-sign} is small.

Recall that, just like in \cref{sec:sign-query}, we want to analyze the modified $\BlindForge$ experiment where the $Q$-measurement is applied after the adversary has output a forgery, but before the secret key register is measured to sample the secret key and verify the forgery. It thus remains to show that due to the fact that $\ket{\psi'_1}$ is approximately in the range of $P_S$, the outcome $l+1$ only occurs with small probability.

To that end, we define a new measurement given by projectors $\tilde Q_i$ that performs the $Q$-measurement controlled on the content of the $M$-register, i.e.
\begin{equation*}
    \tilde Q_i=\sum_{m}\proj{m}_M\otimes Q_i^{m}.
\end{equation*}
Now, observe that applying the $Q$-measurement after the adversary has output a forgery is equivalent to applying the $\tilde Q$-measurement right before the adversary's measurement that produces the forgery. To prove that, if $m^*\in B$, the outcome $l+1$ occurs only with small probability in the modified $\BlindForge$ experiment, it thus suffices to prove the following lemma.
\begin{lemma}\label{lem:lplusoneisrare}
In the \QWorld, for blinded messages, the outcome $l+1$ only occurs with small probability,
\begin{equation*}
    \left\|\tilde Q_{l+1}\Pi^{B}_M\ket{\psi_1'}_{M\Sigma XYES}\right\|_2\le q_1(\delta_L(n) + 4l\epsilon_L(n)),
\end{equation*}
where 
\begin{equation*}
    \Pi^B=\sum_{m\in B}\proj m.
\end{equation*}
\end{lemma}

\begin{proof}
By \cref{Orthogonality argument}, we have
\begin{align}
    Q_{l+1}^m\Pi^{B}_MP_S
    &=\sum_{m}\left(\proj{m}_M\Pi^{B}_M\right)\otimes Q_i^{m}P_S\\
    &=\sum_{m\in B}\proj{m}_M\otimes Q_i^{m}P_S=0.
\end{align}
Therefore we can bound
\begin{align}
    \left\|\tilde Q_{l+1}\Pi^{B}_M\ket{\psi_1'}_{M\Sigma XYES}\right\|_2
    &=\left\|\tilde Q_{l+1}\Pi^{B}_M\left(\ket{\psi_1'}_{M\Sigma XYES}-P_S\ket{\psi_1'}_{M\Sigma XYES}\right)\right\|_2\\
    &\le\left\|\ket{\psi_1'}_{M\Sigma XYES}-P_S\ket{\psi_1'}_{M\Sigma XYES}\right\|_2\\
    &\le q_1(\delta_L(n) + 4l\epsilon_L(n)),
\end{align}
where we have used the fact that $\|\tilde Q_{l+1}\Pi^{B}_M\|_\infty\le 1$ in the first and \cref{lem:hash-after-sign} in the second inequality.
\end{proof}

We are now ready to prove \cref{thm:lamport}.

\begin{proof}[Proof of \cref{thm:lamport}.]
We begin by bounding the success probability of the adversary in the modified $\BlindForge$ experiment, in the \QWorld. Analogously to \cref{eq:mod-succ-no-hqasq}, we bound, abbreviating the modified $\BlindForge$ experiment as $MBF$,
\begin{align}
   \Pr_{QI,MBF} [\A \textrm{ succeeds}]
   &= \sum_{i=1}^{l+1} \Pr_{QI,MBF}[\A \text{ succeeds} \wedge \text{outcome } i]\nonumber \\
   &= \sum_{i=1}^{l} \Pr_{QI,MBF}[\A \text{ succeeds} \wedge\text{outcome } i]
   + \Pr_{QI,MBF}[\A \text{ succeeds} \wedge \text{outcome } l+1]\nonumber \\
   &\leq \sum_{i=1}^{l} \Pr_{QI,MBF}[\text{outcome } i]\times 2^{-n} 
   + \Pr_{QI,MBF}[\text{outcome } l+1]\\
   &\leq 2^{-n} 
   + q^2(\delta_L(n) + 4l\epsilon_L(n))^2,
\end{align}
where ``outcome $i$'' is the event that the $Q$-measurement yields outcome $i$, the first inequality uses the fact that $\sigma^*$ and $s_i^{m_i^*}$ are independent conditioned on outcome $i$, and the last inequality uses the square of the inequality from \cref{lem:lplusoneisrare}.

Exactly as in the simplified case in \cref{sec:sign-query}, we can bound the success probability in the actual $\BlindForge$ experiment using the pinching lemma, \cref{Pinching},
    \begin{align*}
    \Pr_{QI,\BlindForge} [\A \textrm{ succeeds}]
    \leq (l+1)\left(2^{-n} + q^2 \of[\big]{\delta_L(n) + 4l\epsilon_L(n)}^2\right). %
\end{align*}
Finally, plugging in the functions $\epsilon_L(n)$ and $\delta_L(n)$ from \cref{Commutator argument1,commutator of ps and uh}, and applying \cref{lem:ind-r-v-qi} for $w=2$, we obtain
\begin{align*}
    \Pr_{\BlindForge} [\A \textrm{ succeeds}]
    &\leq (l+1)\left(2^{-n} +  q^2\left(\frac{32l}{2^{n/2}}+4l \frac{6}{2^{n/2}}\right)^2\right)+12 l^2 2^{-n}\\
    &\leq l^2\cdot 2^{-n}\left( 3137q^2(l+1)+12\right).
\end{align*}
\end{proof}

\section{One-time \BU security of the Winternitz OTS}\label{sec:winternitz}
The Lamport OTS that we analyzed in the last section is, in some sense, a special case of the Winternitz OTS. Indeed, the Winternitz scheme for $w=2$ is fairly similar to the Lamport OTS, except that the public key is used to sign the bits that are equal to $1$, which is compensated for by the checksum encoding. As a result, the analysis of the Winternitz OTS in the QROM is, in a similar sense, a generalization of the one of the Lamport OTS.

Before getting started, we give and overview of our strategy. In this section, we use the same register labels as in the \cref{tab:registers} of  Lamport OTS section except that the secret key register $S$ is now replaced by the hash chain register $\Gamma$. We remark that the general overview of the proof for Lamport OTS given at the beginning of \cref{sec:lamport} is similar for the Winternitz OTS except that the security argument is different. More precisely, in the Winternitz OTS, the signature algorithm uses the hash chain registers above the queried position to produce the signature. Classically, the property that enables security is that the adversary does not have any information about the part of the hash chain below the queried position, and this represents the invariant of the hash chain. Quantumly, our intuition is that since in the $\BlindForge$ experiment the forged message must be outside the queried region, and since by construction of the checksum, for any queried message there exists at least one position at which the block corresponding to the forged message is smaller than the one of the queried message. Thus, the hash chain corresponding to that specific block should still be in its initial state, and hence in the invariant of the hash chain. Therefore, we want to show that for a moderate number of queries to the random oracle%
, no adversary can win the $\BlindForge$ experiment with a significant probability. Towards that goal, we follow the same steps as in the Lamport OTS. Specifically, we prove the following theorem.
\begin{theorem}\label{thm:winternitz}
The Winternitz OTS in \cref{sec:WOTS} is 1-\BU secure if the function chain $\mathcal{C}$ is modeled as a quantum-accessible random oracle.
More precisely, let $\A$ be an adversary that plays the $\BlindForge$ game for the Winternitz OTS, making a total of $q$ queries to the random oracle. Then $\A$ succeeds with a probability bounded as
\begin{align}
   \Pr[\text{$\A$ wins $\BlindForge$}] 
    &\leq  2^{-n}\left[\left(1+q^2l^2(w-1)^2(20w-4)^2\right)(l+1) +3w^2l^2\right]\\
    &\leq 800w^4q^2l^3\cdot 2^{-n}.
\end{align}
Here, $l$ is the length of the encoded message in $w$-ary, see \cref{eq:Winternitz-params}, $w\ge 2$ is the Winternitz parameter, and the simplified bound in the last line holds for $q>0$.
\end{theorem}
 
 The main difference between the analyses of the Lamport and Winternitz OTS is as follows. For the Lamport OTS, the public key is obtained from the private key by applying a hash function once. For the Winternitz OTS, on the other hand, the secret an public key consist of the start and end points of length $w$ hash chains, respectively. Thus, while following the same proof strategy, the $Q$ projectors as well as the invariant projector $P$ needs to be defined differently. Thus, we start our analysis by describing the $Q$ projectors and the invariant projector for the Winternitz OTS.
\subsection{$Q$ projectors for Winternitz OTS} \label{sec:QProjWinternitz}

The Winternitz signature of any message is composed of $l$ hash chain elements. In complete analogy to \cref{eq:Q} in \cref{sec:Q measurement}, we define a measurement whose projectors correspond respectively to the events that \emph{the $i$-th hash chain element relevant for the forged signature is in state $\ket\Phi$} and \emph{none of them is in state $\ket\Phi$}:
\begin{align}
    Q_{i^*}^{b^*}
    &=
    \Phi^\perp_{\Gamma_1^{b^*_1}} \x \dotsb \x
    \Phi^\perp_{\Gamma_{i^*-1}^{b^*_{i^*-1}}} \x
    \Phi_{\Gamma_{i^*}^{b^*_{i^*}}}, &
    Q_{l+1}^{b^*}
    &= \bigotimes_{i=1}^l
    \Phi^\perp_{\Gamma_i^{b^*_i}}
    \label{eq:Qq}
\end{align}
where $i^* \in \set{1,\dotsc,l}$, $b_i^* = b_i(m^*)$ and $l$ is the number of blocks of the message and the checksum, see \cref{eq:Winternitz-params}.
These operators act as $\1$ on all other registers $\Gamma_i^j$ that are not specified.

 \subsection{Invariant projector for Winternitz OTS}\label{sec:invariant projector for Winternitz}

In this section, we define the invariant projector, denoted by $P_\Gamma$, that will be used to track the invariant of the hash chain register. We also state several of its properties.

Recall from our discussion of blind unforgeability in \cref{sec:BU} that $B$ denotes the set of blinded messages and $B^c$ its complement,i.e., the set of un-blinded messages. We also recall from the description of the Winternitz OTS in \cref{sec:WOTS} that a block $b$ of a message is the concatenation of the blocks obtained from the encoding of the message and its corresponding checksum in w-ary.

Define $\alpha = (\alpha_i^j)_{i=1,\dotsc,l}^{j=0,\dotsc,w-2}$ as a $l(w-1)$-bit string whose bits $\alpha_i^j \in \set{0,1}$ indicate that the projector $\Phi(\alpha_i^j)$ is applied on the corresponding hash chain register $\Gamma_i^j$ where
\begin{align}
    \Phi(0) &= \Phi, &
    \Phi(1) &= \Phi^\perp.
\end{align}
For each string $\alpha$, we define the associated projector $\Phi(\alpha)$ on the whole hash chain (except for the last) register $\Gamma$ as
\begin{equation}
    \Phi(\alpha)_\Gamma =
    \bigotimes_{i=1}^l
    \bigotimes_{j=0}^{w-2}
    \Phi(\alpha_i^j)_{\Gamma^j_i}.
    \label{eq:Phi(a)-Wots}
\end{equation}
Note that this is a complete set of projectors, i.e., $\sum_{\alpha\in\{0,1\}^{l(w-1)}} \Phi(\alpha)_\Gamma = \1_\Gamma$.

Since we are interested in the unused part of the hash chain register, we need to filter those $\alpha$'s for which $\Gamma^j_i$ is in state $\ket\Phi$.
By construction of the checksum, if a block $b$ of a message $m$ is computed,
then in the block $b'$ of any other message $m'$, there exists at least one position $i$ at which $b'_i<b_i$, $1\le i\le l$. Therefore, since the blinded signing oracle signs at most a single un-blinded message, $m\in B^c$, the state after the signing oracle call can be written as a superposition of states where, for some un-blinded message $m' \in B^c$, $b'_i<b_i$ for all $i$. The latter implies that the hash chain registers corresponding to those $b'_i$ are still in  the uniform superposition $\ket{\Phi}$, for all $i$.
Thus, we collect all strings $\alpha$ that are consistent with no blinded messages having been signed in
\begin{align}
    \widehat{B^c} = \bigcup_{m \in B^c}
    \set[\Big]{
        \alpha \in \set{0,1}^{l(w-1)}
        \mathrel{\Big|}
        \text{$\alpha_i^j = 0$ for all $i=1,\dotsc,l$ and $j < b_i(m)$}
    }
    \label{def of Be complement}
\end{align}
as the set of strings $\alpha$ that indicate which hash chain registers were not used during hash queries and $\Sign$ query, that is those that fulfill the condition $\alpha_i^{\bar{m}_i} = 0$ for all $i$.
Specifically, $\widehat{B^c}$ contains all the strings that are consistent with no blinded messages having been signed.
Finally, we define
\begin{equation}
    P_\Gamma 
  = \sum_{\alpha\in\widehat{B^c}}
    \Phi(\alpha)_\Gamma
    \label{eq:PGamma}
\end{equation}
as the projector acting on the invariant hash chain register, specifically on the subspace consistent with $\widehat{B^c}$. Note that $P_\Gamma$ is indeed a projector since it is a sum of mutually orthogonal projectors.
 
Using these new definitions of the $Q$ projectors and the invariant projector $P_\Gamma$, a set of lemmas similar to \cref{Commutator argument1,Orthogonality argument,final state invariant if no hash queries,commutator of ps and uh} forms the basis of the \BU security proof for the Winternitz OTS. In fact, \cref{Commutator argument1} is a special case of \cref{lem:Commutator} where the register $\Gamma$ is replaced by $S$ and we set $w = 2$ (see \cref{apx:L8andL15} for proof).
\Cref{Orthogonality argument} holds for the new projectors $Q_{l+1}$ and $P_\Gamma$ by construction. Finally, \cref{commutator of ps and uh,final state invariant if no hash queries} need to be changed slightly for the Winternitz OTS and are stated below. \Cref{A'stateInv:winternitz,lem:orth-Ql+1-Pga-Wots,ComPgamUh:winternitz} are proved in \cref{apx:Winternitz lemmas}.

\begin{restatable}{lemma}{LemComWin}\label{lem:Commutator}
Let $U_h$ be the random oracle unitary for any given function $h$ (see \cref{Hash chains}) and let $\Phi = \proj{\Phi}$ denote the projector onto the uniform superposition $\ket{\Phi}$. Furthermore, let $\Gamma_{i}^{\leq j} = \Gamma_i^0 \dots \Gamma_i^{j}$ and
\begin{equation}
    \Phi_{\Gamma_{i}^{\leq j}}
    = \of*{\Phi\xp{j}}_{\Gamma_{i}^{\leq j}}.
    \label{eq:R}
\end{equation}
Then, for any $i' \in \set{1,\dotsc,l}$ and $j' \in \set{0,\dotsc,w-2}$,
\begin{equation}
    \norm[\Big]{\sof[\Big]{
        (U_h)_{XY\Gamma},
        \Phi_{\Gamma_{i'}^{\leq j'}}
    }}_\infty
    \leq \frac{6(w-1)}{2^{n/2}}
    = \epsilon_W(n)
    \label{eq:UhPhi-Wots}
\end{equation}
is negligible in $n$.
\end{restatable}

\begin{restatable}{lemma}{LemNoHashWin}\label{A'stateInv:winternitz}
Let $B\Sign_{\sk}$ be the blinded signing oracle for the Winternitz OTS, and let $\ket{\psi_0}$ be the adversary's state before the $\Sign$ query. If there are no hash queries, then after making a single $\Sign$ query the adversary's state $\ket{\psi_1} = B\Sign_{\sk}\ket{\psi_0}$ is completely in the range of the invariant projector $P_\Gamma$ defined in \cref{eq:PGamma}. That is,
\begin{equation}
    P_\Gamma B\Sign_{\sk}\ket{\psi_0} = B\Sign_{\sk}\ket{\psi_0}.
\end{equation}
\end{restatable}

For every message in the blinding set $B$, there exists at least one hash chain element necessary for its corresponding signature such that the corresponding hash chain register is in the uniform superposition state $\ket\Phi$. This implies the following lemma.

\begin{restatable}{lemma}{LemOrthWots}\label{lem:orth-Ql+1-Pga-Wots} 
Let $m^* \in B$ and $b^* = b(m^*)$, see \cref{eq:b(m)}. Then the projectors $Q_{l+1}^{b^*}$ and $P_\Gamma$ defined in \cref{eq:PGamma,eq:Qq} are orthogonal:
\begin{equation}
    Q_{l+1}^{b^*} P_\Gamma = 0.
\end{equation}
\end{restatable}

\begin{restatable}{lemma}{LemComPWin}\label{ComPgamUh:winternitz}
Let $P_\Gamma$ and $U_h$ be respectively the invariant projector for the Winternitz OTS and the random oracle unitary defined with respect to the \QWorld. If there are hash queries after the $\Sign$ query, then 
\begin{equation}
     \bigl\| [U_h,P_\Gamma]
     \bigl\|_\infty
     \leq \delta_W (n)
\end{equation}
where \[\delta_W(n) = \frac{8l(w+1)(w-1)}{2^{n/2}} .\]
\end{restatable}

Just like in the Lamport OTS, we use the above lemmas to prove our main results. In the following sections, we analyze the situation where the adversary makes $q_0$ hash queries before the $\Sign$ query and $q_1$ hash queries after, maximizing the resulting bound under the condition that $q_0+q_1=q$.

The proof of \cref{thm:winternitz} is presented in steps in the following subsections. We begin by presenting some concepts and tools which will be used in the proof. Subsequently, we prove the lemmas stated above. Finally, we combine them to prove \cref{thm:winternitz}.

\subsection{Hash queries before \texorpdfstring{$\Sign$}{Sign} query}
In this section, we study the impact of hash queries before $\Sign$ query on the hash chain register $\Gamma$. Our goal is to show that, for a moderate number of queries to the random oracle, no adversary can learn a significant amount of information about the hash chain. Therefore, she cannot produce a valid forgery except with small probability.

Let $\ket{\psi}_{XYM\Sigma E}$ be adversary's initial state before any queries. Before any query is performed, the whole hash chain register $\Gamma$, except the last, is in the uniform superposition state, i.e,
\begin{align}
    \ket{\nu}_\Gamma =
    \bigotimes_{i=1}^l
    \bigotimes_{j=0}^{w-2}
    \ket{\Phi}_{\Gamma_i^j}.
    \label{initial:hash-chain}
\end{align}

Assume the adversary $\A_0$ queries the \textrm{random oracle} %
$q_0$ times before querying the signing oracle.
If $V_{XYE}^i$ denotes the unitary she performs after the $i$-th query, the final adversary state after $q_0$ hash queries is
\begin{equation}
    \ket{\psi_0}_{XYM\Sigma E\Gamma}
  = V^{q_0}_{XYE}
    (U_h)_{XY\Gamma}
    V^{q_0-1}_{XYE}\cdots
    V^2_{XYE}
    (U_h)_{XY\Gamma}
    V^1_{XYE}
    (U_h)_{XY\Gamma}
    \ket{\psi}_{XYM\Sigma E}
    \ket \nu _\Gamma
    \label{eq:before-hash-query}
\end{equation}
where $U_h$ is the random oracle unitary used to answer hash queries. The following lemma shows that the hash chain registers of this state are still close to being in uniform superposition.

\begin{lemma}\label{lem:hash-before-sign-Wots}
In the \QWorld, without querying the $B\Sign$ oracle, hash queries leave the state of the secret key registers approximately unchanged:
\begin{align}
    \norm*{
    \Phi\xp{l(w-1)}_\Gamma
    \ket{\psi_0}_{XYM\Sigma E\Gamma} - 
    \ket{\psi_0}_{XYM\Sigma E\Gamma}
    }_2
    \leq lq_0\epsilon_W(n)
    \label{eq:UhPhi-wots}.
\end{align}
\end{lemma}

\begin{proof}
The proof of \cref{lem:hash-before-sign-Wots} is very similar to the proof of \cref{lem:hash-before-sign} in the Lamport OTS security analysis, except that it uses \cref{lem:Commutator} for each of the $l$ hash chains, with $j'=w-2$, where the proof of \cref{lem:hash-before-sign} applies \cref{Commutator argument1} for each of the $2l$ secret key registers.
\end{proof}
As in for the Lamport OTS, the above Lemma means that the adversary learns almost no information about the hash chain, unless $q_0=\Omega(2^{n/2})$.

\subsection{Query to the signing oracle}\label{sec:sign-queryWots}

Now that we have control over the advantage an adversary can gain from making hash queries before the sign query, we need to analyze the possible advantage from hash queries after the sign query and bound the overall success probability using \cref{lem:hash-before-sign-Wots}.
The discussion in this section does not concern the random oracle, so we absorb the random oracle query registers $XY$ into $E$ for the purpose of this section. 

A key property of the Winternitz OTS when analyzing classical security is that for all messages $m$ that have not been queried, there exists an index $j$ such that $s_j^{m_j}$ is hidden from the adversary by the preimage resistance of the used hash function. In blind-unforgeability (for classical adversaries), this property holds for all \emph{blinded messages}. In the setting of quantum queries, we have to track this property in superposition while the adversary is making hash queries after the sign query. As this is complicated by the ``for all''-quantifier, we begin by analyzing the case where the adversary makes no hash queries after the sign query to ease the reader into our proof technique.

In the $1$-$\BlindForge$ game, an adversary $\A$ is allowed to query the $\Sign$-oracle at most once to produce a valid forged message-signature pair $(m^*,\sigma^*)$. To analyze the interaction between $\A$ and the signing oracle, we will follow the steps stated in \cref{Interaction:Bsign-Adver}.
Those steps correspond to applying the $\Sign$-oracle and an arbitrary unitary $U_{M\Sigma E}$, followed by measuring the message and signature registers $M$ and $\Sigma$. Let us now analyze these steps in more detail and write down the corresponding quantum states.

First, the adversary $\A$ prepares the state
\begin{align}
  \ket{\psi _0}_{M\Sigma EB\Gamma} 
   &=\sum_{m\in\{0,1\}^a}
    \sum_{\sigma\in(\{0,1\}^n)^l}
    \kappa_{m\sigma b} 
    \ket{m}_M
    \ket{\sigma}_{\Sigma}
    \ket{\alpha_{m\sigma b}}_E
    \ket{b}_B
    \ket \nu _\Gamma\label{eq:psi0-W}
\end{align}
where
$\ket{m}_M = \ket{m_1 \cdots m_a}_M$,
$\ket{\sigma}_{\Sigma} = \ket{\sigma_1}_{\Sigma_1}\cdots\ket{\sigma_l}_{\Sigma_l}$,
$\ket\nu_\Gamma$ is defined in \cref{initial:hash-chain}, and $b$ is the amplitude of the blinding register $B$.
Here $\Gamma$ is a composite register of the form $\Gamma = \{\Gamma_i^j : i \in \{1,\dotsc,l\}, j \in \{0,\dotsc,w-1\}\}$, $B$ indicates whether the message is blinded or not ($\ket{1}_B$ for blinded and $\ket{0}_B$ for un-blinded), and the remaining registers are defined as in \cref{tab:registers}.

Next, $\A$ queries this state to the $\Sign$ oracle which answers the query with the following signed state:
\begin{align}
  \ket{\psi _1}_{M\Sigma EB\Gamma}
  &=B\Sign_{\sk}
    \ket{\psi _0}_{M\Sigma EB\Gamma} \nonumber\\
  &=\ket{\psi _1^1}_{M\Sigma EB\Gamma} +
    \ket{\psi _1^0}_{M\Sigma EB\Gamma}
    \label{eq:psi0 and psi1}
\end{align}
where the superscripts $1$ and $0$ correspond to blinded ($B$) and un-blinded ($B^c$) messages. %
The expression of the first term  is given by
 \begin{align}
  \ket{\psi _1^1}_{M\Sigma EB\Gamma}
  &=\sum_{m\in B}
    \sum_{\sigma\in(\{0,1\}^n)^l}
    \kappa_{m\sigma 1} 
    \ket{m}_M
    \ket{\sigma}_{\Sigma}
    \ket{\alpha_{m\sigma 1}}_E
    \ket{1}_B
  \ket \nu _\Gamma
\end{align}
where the latter follows because for blinded messages ($m\in B$) there is no signature. 

From now on, we will describe how the second term $\ket{\psi _1^0}_{M\Sigma EB\Gamma}$ in \cref{eq:psi0 and psi1} is obtained. While generally the signature is computed in superposition, we will describe it on a fixed message for the sake of simplicity. The general operation corresponds to extending this description by linearity.

Given a fixed message $\ket{m}_M = \ket{m_1 \cdots m_a}_M \in B^c$ represented in computational basis, the signing oracle first encodes the message in $l$ blocks, each in base-$w$ representation:
\begin{align}
    U_b \ket m _M \ket 0 _W  
    \ket{\sigma}_{\Sigma}
    \ket \nu _\Gamma
    \mapsto 
    \ket m _M 
    \ket{ 0 \oplus b(m)}_W
    \ket{\sigma}_{\Sigma}
    \ket \nu _\Gamma
\end{align}
where $l$ is defined in \cref{eq:Winternitz-params}, $W$ is an ancilla register used to store $b(m)$, and $b(m)$ is defined in \cref{eq:b(m)} as
\begin{equation*}
    b(m) = (b_1,\dotsc,b_l) = m \parallel C(m)
\end{equation*}
where $C(m)$ is the checksum corresponding to $m$.
Note that the process by which the unitary $U_b$ computes $\ket{b(m)}$ is similar to the classical way described in \cref{sec:WOTS}.

Using the blocks of $m$, the signing oracle computes the signature as follows:
\begin{align*}
     U_{\Sign_{\sk}
    {W\Gamma:\Sigma}}
    \ket{m}_M
    \ket{b_1\cdots b_l}_W
    \ket{\sigma \oplus \gamma}_\Sigma
    \ket{\Omega(m,\gamma)}_\Gamma
    &=
    \CNOT^{\otimes n}_{\Gamma^{b_1}_1:\Sigma_1}\cdots
    \CNOT^{\otimes n}_{\Gamma^{b_l}_l:\Sigma_l}
    \ket{m}_M
    \ket{b_1\cdots b_l}_{W}
    \ket{\sigma}_{\Sigma}
    \ket \nu _\Gamma\\
    &=
    \ket{m}_M
    \ket{b_1\cdots b_l}_W
    \ket{\sigma \oplus \gamma}_\Sigma
    \ket{\Omega(m,\gamma)}_\Gamma.
\end{align*}
    
Once the signature is obtained, the ancilla register $W$ is not useful for further analysis so we can remove it from the signed state by applying $U_b^{\dagger}$. Thus, the final signed state of a fixed message $\ket m_M$ is given by
\begin{align}
     \ket{m}_M
    \ket{\sigma \oplus \gamma}_\Sigma
    \ket{\Omega(m,\gamma)}_\Gamma
\end{align}
with
\begin{align}
    \ket{\sigma \oplus \gamma}_\Sigma
    &:=
    \ket{\sigma_1\oplus\gamma_1}_{\Sigma_1}\dots
    \ket{\sigma_l\oplus \gamma_l}_{\Sigma_l}\label{signature},\\
    \ket{\Omega(m,\gamma)}_\Gamma
    &:= \left(
    \bigotimes_{i=1,\dotsc,l; b_i\neq w-1}
    \ket{\gamma_i}_{\Gamma_i^{b_i}}
    \bigotimes_{j=0,\dotsc,w-2; j\neq b_i}
    \ket{\Phi}_{\Gamma_1^j} \dots
    \bigotimes_{j=0,\dotsc,w-2; j\neq b_i}
    \ket{\Phi}_{\Gamma_l^j}
    \right)
    \ket{p_1}_{\Gamma_1^{w-1}}\dots
    \ket{p_l}_{\Gamma_l^{w-1}}.
    \label{eq:Omega-Wots}
\end{align}

By linearity, the signature of the un-blinded term which is composed of superposition of messages is given by
\begin{align}
   \ket{\psi _1^0}_{M\Sigma EB\Gamma}
   &=
   \sum_{m\in B^c}\xi_m
    \sum_{\sigma\in(\{0,1\}^n)^l}
    \sum_{\gamma\in(\{0,1\}^n)^l} 
    \kappa_{m\sigma 0} 
    \ket{m}_M
    \ket{\sigma \oplus \gamma}_\Sigma
    \ket{\alpha_{m\sigma 0}}_E
    \ket{0}_B
    \ket{\Omega(m,\gamma)}_\Gamma
\end{align}
where $\xi_m$ is a normalization factor of all hash chain elements used to produce the signature.

Once the adversary receives the signed state $\ket{\psi _1}_{M\Sigma EB\Gamma}$, she carries out some operations with the intention of producing a forged message $m^*$. Intuitively, those operations can be viewed as applying some arbitrary unitary $U_{M\Sigma E}$ to $\ket{\psi _1}_{M\Sigma EB\Gamma}$: 
\[\ket{\psi _2}_{M\Sigma EB\Gamma} = 
U_{M\Sigma E}
\ket{\psi_1}_{M\Sigma EB\Gamma}.
\]
Afterwards, the adversary $\A$ measures the $M$ register of the latter state, which gives outcome $m^*\in \{0,1\}^a$.
After the measurement, the state $\ket{\psi _2}_{M\Sigma EB\Gamma}$ collapses to the (unnormalized) state 
\begin{align*}
\ket{\psi_3(m^*)}_{\Sigma EB\Gamma}
  &= \bra{m^*}_{M} \ket{\psi _2}_{M\Sigma EB\Gamma} \\
  &=\ket{\psi_3^1(m^*)}_{\Sigma EB\Gamma} + 
    \ket{\psi_3^0(m^*)}_{\Sigma EB\Gamma}  
\end{align*}
whereby,
\begin{align*}
    \ket{\psi_3^1(m^*)}_{\Sigma EB\Gamma}
  = \bra{m^*}_{M} U_{M\Sigma E}
    \ket{\psi_1^1}_{M\Sigma EB\Gamma}
\end{align*}
and
\begin{align*}
  \ket{\psi_3^0(m^*)}_{\Sigma EB\Gamma}
 &= \bra{m^*}_{M}
    U_{M\Sigma E}
    \ket{\psi_1^0}_{M\Sigma EB\Gamma} \\
 &= 
    \sum_{m\in B^c}\xi_m
    \sum_{\sigma\in(\{0,1\}^n)^l}
    \sum_{\gamma\in(\{0,1\}^n)^l} 
    \kappa_{m\sigma 0} 
    \bra{m^*}_{M}U_{M\Sigma E}
    \ket{m}_M
    \ket{\sigma \oplus \gamma}_\Sigma
    \ket{\alpha_{m\sigma 0}}_E
    \ket{0}_B
    \ket{\Omega(m,\gamma)}_\Gamma.
\end{align*}

The goal of the adversary $\A$ is to produce a forged message $\sigma^*$ that matches $m^*$. Towards this end, she measures the $\Sigma$ register of $\ket{\psi_3(m^*)}_{\Sigma EB\Gamma}$, obtaining outcome $\sigma^*\in (\{0,1\}^n)^l$.
Then, the (unnormalized) post-measurement state is
\begin{align*}
\ket{\psi_4(m^*,\sigma^*)}_{EB\Gamma}
  & =\bra{\sigma^*}_{\Sigma} 
     \ket{\psi_3(m^*)}_{\Sigma EB\Gamma} \\
  &= \bra{\sigma^*}_{\Sigma} 
     \ket{\psi_3^1(m^*)}_{\Sigma EB\Gamma}
  +  \bra{\sigma^*}_{\Sigma}  
     \ket{\psi_3^0(m^*)}_{\Sigma EB\Gamma}\label{eq:final state-Wots}\\
  &=\ket{\psi_4^1(m^*)}_{EB\Gamma} + 
    \ket{\psi_4^0(m^*)}_{EB\Gamma}  
\end{align*}
where
\begin{align*}
    \ket{\psi_4^1(m^*,\sigma^*)}_{EB\Gamma}
  = \bra{\sigma^*}_{\Sigma} 
    \bra{m^*}_{M}U_{M\Sigma E}
    \ket{\psi_1^1}_{M\Sigma EB\Gamma}
\end{align*}
and
\begin{align}
    \ket{\psi_4^0(m^*,\sigma^*)}_{EB\Gamma}
  &=\sum_{m\in B^c}
    \xi_m
    \sum_{\sigma\in(\{0,1\}^n)^l}
    \sum_{\gamma \in (\{0,1\}^n)^l} 
    \alpha_{m\sigma 0} 
    \bra{\sigma^*}_{\Sigma}
    \bra{m^*}_{M}U_{M\Sigma E}
    \ket{m}_M
    \ket{\sigma \oplus \gamma}_\Sigma
    \ket{\alpha_{m\sigma 0}}_E
    \ket{0}_B
    \ket{\Omega(m,\gamma)}_\Gamma.
\end{align}

For simplicity, we can rewrite $\ket{\psi_4^0(m^*,\sigma^*)}_{EB\Gamma}$ as
\begin{align}
    \ket{\psi_4^0(m^*,\sigma^*)}_{EB\Gamma}
  &= \sum_{m\in B^c}
    \sum_{\gamma\in(\{0,1\}^n)^l}
    \ket{\eta(m,\gamma)}_{EB}
    \ket{\Omega(m,\gamma)}_\Gamma\label{eq:psi04-wots}
\end{align}
where only $\ket{\eta(m,s)}_{BE}$ depends on $m^*$ and $\sigma^*$:
\begin{align*}
   \ket{\eta(m,\gamma)}_{EB}
   &= \xi_m
    \sum_{\sigma\in(\{0,1\}^n)^l}
    \kappa_{m\sigma 0} 
    \bra{\sigma^*}_{\Sigma}
    \bra{m^*}_{M}U_{M\Sigma E}
    \ket{m}_M
    \ket{\sigma \oplus \gamma}_\Sigma
    \ket{\alpha_{m\sigma 0}}_E
    \ket{0}_B.
\end{align*}
Hence, the adversary $\A$  produces a forged message-signature pair $(m^*,\sigma^*)$.
The probability of producing this pair is $\norm{\ket{\psi_4^0(m^*,\sigma^*)}_{BE\Gamma}}^2$.

The next step is to analyze the probability that $\A$'s forgery candidate is correct.
Recall from \cref{eq:b(m)} that for some un-blinded message $m\in B^c$ or blinded message $m^*\in B$, we denote by $b= b(m)$ and $b^* = b(m^*)$ the base-$w$ representation of $m$ and $m^*$ in term $l$ blocks, each in $w$-ary. 

Here, we consider two cases. The first case, namely when $m^*\notin B$, is trivial because $\A$ has lost the $\BlindForge$ experiment as $m^*$ must be blinded by definition. The rest of this section is devoted to analyzing the second case.

If $m^*\in B$, then the forged message $m^*$ has not been signed since the blinded $\Sign$ oracle signs only un-blinded messages. Hence, by construction of the checksum, for any message $m\notin B$ there exists at least one index $i^*\in\{1,\dotsc,l\}$ such that the corresponding block $b$ is larger than the block $b^*$ of $m^*$. This implies that this specific hash chain element $\Gamma_{i^*}^{b_{i^*}}$ has not been used for the signature of the adversary's queried message and is therefore still in its initial state. Note that this holds only in superposition over $m$. Indeed, $i^*$ depends on $m$ and is in general different for each term of the superposition.

In the Winternitz OTS, we know that the hash sub-chain below the queried position is not used for the signature procedure and is therefore in uniform superposition $\ket\Phi$. But given that during the signature process the hash chain is in superposition,
our main goal here is to track the invariant of the hash chain and show that some of the hash chain elements relevant for verifying the forged signature satisfy this invariant and are thus unknown to the adversary, so it is unlikely that the adversary would have used the correct hash chain to produce the forged signature.

To that end, we analyze a modified $\BlindForge$ experiment, where an additional measurement is performed on the hash chain register after the adversary has output their forgery, but before the hash chain register is measured to actually sample the hash chain as required in the \QWorld. This additional $Q$-measurement was defined in \cref{eq:Q}, with Winternitz $Q$ projectors defined in \cref{sec:QProjWinternitz}. Since it has few outcomes, its effect on the adversary's winning probability is limited and can be bounded by the pinching lemma (\cref{Pinching}).

If the $Q$-measurement yields outcome $i^* \in \{1,\dotsc,l\}$, then  the hash chain element $\Gamma_{i^*}^{b_{i^*}}$ is in uniform superposition and the adversary is bound to fail as $\sigma^*$ is independent of the hash chain $n$-bit string $\gamma_{i^*}^{b_{i^*}}$ (the result of measuring $\Gamma_{i^*}^{b_{i^*}}$).

It remains to analyze the outcome $l+1$ that corresponds to the projector $Q_{l+1}^m = (\Phi^\perp)\xp{l}$ where $\Phi^\perp = \1 - \proj{\Phi}$.
The final adversary state after the measurement, see \cref{eq:final state-Wots}, contains both blinded and un-blinded terms. If we apply $\Phi^\perp$ to any hash chain register of the blinded term $\ket{\psi_4^1(m^*,\sigma^*)}_{BE\Gamma}$, we get $0$ since all hash chain registers (except for the last) are in uniform superposition $\ket\Phi$ and the checksum guarantees that there exists at least one position ${i^*}$ such that $b_{i^*} < w-1$.

For the rest of this section, we fix a message $m^*$ and focus on the un-blinded term $\ket{\psi_4^0(m^*,\sigma^*)}_{BE\Gamma}$. Because $m^*\in B$, it has not been signed. So for any un-blinded message $m\in B^c$ that has been queried to the $\Sign$ oracle, the checksum guarantees that there exists at least one index $i^*$ for which $b_i(m^*)< b_i(m)$. Therefore, the hash chain element $\Gamma^{b_{i^*}}_{i^*}$ corresponding to that position is in state $\ket\Phi$.
To find that position, we define
\[
    i(b) = \min \{k= 1,\dotsc,l\mid b^*_k < b_k\} 
\]
as the smallest index $k$ for which $b_k(m^*) < b_k(m)$. Intuitively, it is the first hash chain element of $\Gamma$ that still remains in uniform superposition. In the following, let $\Gamma(b(m)) := \Gamma_1^{b_1(m)} \dotsb \Gamma_l^{b_l(m)}$ and recall from \cref{eq:psi04-wots,eq:Omega-Wots} that the un-blinded term is given by
\begin{align}
    \ket{\psi_4^0(m^*,\sigma^*)}_{EB\Gamma}
  &= \sum_{m\in B^c}
    \sum_{\gamma\in(\{0,1\}^n)^l}
    \ket{\eta(m,\gamma)}_{EB}
    \ket{\Omega(m,\gamma)}_\Gamma
    \ket{\Phi}_{\Gamma(b(\bar m))}
    \nonumber\\
   &=
   \sum_{k=1}^l\ket{\hat
    \eta(m^*,\sigma^*,k)}_{EB\Gamma_{\{(k,b^*_k)\}^c}} \ket{\Phi}_{\Gamma^{b^*_k}_k},
    \label{B^cSplit-Wots}
\end{align}
where all the registers except for $\Gamma^{b^*_k}_k$ are included into the first system. The state $\ket{\hat\eta(m^*,\sigma^*,k)}$ is defined as the part of the superposition over $m$ in  $\ket{\psi_4^0(m^*,\sigma^*)}$ with constant  $k = i(b)$ (excluding the register $\Gamma_k^{b^*_k}$). Thus the register $\Gamma^{b^*_k}_k$ is still in the uniform superposition $\ket{\Phi}$.
Applying $Q_{l+1}^m$ onto the $l$ hash chain elements $\Gamma(b^*)$ of the register $\Gamma$ in \cref{B^cSplit-Wots} gives
\begin{align*}
    \of[\big]{Q_{l+1}^m}_{\Gamma(b^*)}
    \ket{\psi_4^0(m^*,\sigma^*)}_{EB\Gamma}
  = \of[\big]{(\Phi^\perp)\xp{l}}_{\Gamma(b^*)}
    \left(\sum_{k=1}^l
    \ket{\hat\eta(m^*,\sigma^*,k)}_{BE\Gamma_{\{(k,b^*_k)\}^c}} \ket{\Phi}_{\Gamma^{b^*_k}_k}\right)
  = 0,
\end{align*}
which vanishes because, for each $k$, the register $\Gamma^{b^*_k}_k$ is in state $\ket{\Phi}$ and $\Phi^\perp \ket{\Phi} = 0$. Hence, the situation where none of the hash chain elements relevant for the verification of the forged signature $\sigma^*$ is in state $\ket{\Phi}$ can never occur.

Finally, we bound the success probability of the adversary in winning the $\BlindForge$ game in the \RealWorld. 
Given that this analysis  is roughly the same as the one done in \cref{sec:sign-query} of the Lamport OTS. Following the same steps, we get the following adversary's success probability in producing a valid forged message-signature pair $(m^*,\sigma^*)$:
\begin{align*}
    \Pr_{QI}
    \sof*{\textnormal{$\A$ wins $\BlindForge$}}
  &=\sum_{m^*,\sigma^*}
    \Pr_{\BlindForge}
    \sof*{\textnormal{Success $\wedge$
    $\A$ outputs $(m^*,\sigma^*)$}}
    \leq\frac{l+1}{2^n}.
\end{align*}
But, given that there is possible collision tuples in \RealWorld, adding the upper bound in \cref{lem:intermediate world 1 and intermediate world 2} to the latter equation gives
\begin{align*}
    \Pr
    \sof*{\textnormal{$\A$ wins $\BlindForge$}}
  &=\sum_{m^*,\sigma^*}
    \Pr_{\BlindForge}
    \sof*{\textnormal{Success $\wedge$
    $\A$ outputs $(m^*,\sigma^*)$}} \\
   &\leq\frac{l+1}{2^n} + \frac{3(wl)^2}{2^n}
   =\frac{1+l+3(wl)^2}{2^n}.
\end{align*}

\subsection{Hash queries after \texorpdfstring{$\Sign$}{Sign} query}\label{sec:HashQueryAfter-Wots}

In this section, we analyse the adversary's \textit{hash queries after $\Sign$ query} to bound the success probability that an adversary with a given number of queries can achieve in the $\BlindForge$ game and thus prove \cref{thm:winternitz}.
In this case it is not obvious how to track the invariant of the hash chain, i.e. the fact that there is at least one unused part of the hash chain element that is relevant for the forged signature. Therefore we use a special projector $P_\Gamma$ defined in \cref{sec:invariant projector for Winternitz} that projects onto the subspace of the hash chain register that is consistent with a single blinded sign query and no hash queries. If the final adversary state after producing the forgery candidate is in the image of $P_\Gamma$, then the outcome $l+1$ corresponding to the situation where \textit{none of the hash chain elements useful for the forged signature is in state $\ket\Phi$} can never occur, according to \cref{lem:orth-Ql+1-Pga-Wots}. We thus want to show that  the adversary's final state is approximately in the range of $P_\Gamma$.

If there are no hash queries before the $\Sign$ query, then from \cref{A'stateInv:winternitz} the adversary state after the $\Sign$ query remains completely in the range of $P_\Gamma$, which means that the outcome $l+1$ cannot occur. That is,
\[P_\Gamma\ket{\psi_1} = P_\Gamma B\Sign_{\sk}\ket{\psi_0} = B\Sign_{\sk}\ket{\psi_0} = \ket{\psi_1}
\]
where $\ket{\psi_0}$ and $\ket{\psi_1}$ are respectively the adversary state immediately before and after the $\Sign$ query.

Now, assuming there are hash queries before the $\Sign$ query, since the projector $P_\Gamma$ and the random oracle unitary $U_h$ approximately commute by \cref{ComPgamUh:winternitz}, it follows that hash queries before $\Sign$ query give no significant information to the adversary about the invariant of the hash chain register.

Suppose there are hash queries after the $\Sign$ query and examine in detail what happen in this case. From the previous case, we know that the adversary's state directly after the $\Sign$ query is $\ket{\psi_1}_{M\Sigma XYE\Gamma}$. Just like for hash queries before the $\Sign$ query, suppose that the adversary makes $q_1$ hash queries after querying the signing oracle. 
Let $(W^i_{XY\Gamma E})_{i=1,\dotsc,q_1}$ be unitaries applied to the adversary state between hash queries. Then, the adversary's state after $q_1$ hash queries and before performing some unitary operations $U_{M\Sigma E}$ on the post hash queried state or any measurement leading to the forgery candidate is
\begin{align*}
    \ket{\psi_1'}_{XYM\Sigma EB\Gamma}
  &= (U_h)_{XY\Gamma}
     W^{q_1}_{XYM\Sigma E}
    (U_h)_{XY\Gamma}
    W^{q_1-1}_{XYM\Sigma E}\cdots
    W^2_{XYM\Sigma E}
    (U_h)_{XY\Gamma}
    W^1_{XYM\Sigma E}
    \ket{\psi_1}_{XYM\Sigma EB\Gamma}
\end{align*}
where $\ket{\psi_1}_{XYM\Sigma EB\Gamma}$ is the adversary's state immediately after \textrm{$\Sign$ query}.
For sake of simplicity, we set
\[T = (U_h)_{XY\Gamma}
     W^{q_1}_{XYM\Sigma E}
    (U_h)_{XY\Gamma}
    W^{q_1-1}_{XYM\Sigma E}\cdots
    W^2_{XYM\Sigma E}
    (U_h)_{XY\Gamma}
    W^1_{XYM\Sigma E}.\]
Next, we prove the following lemma.
\begin{lemma}\label{lem:hash-after-sign-Wots}
In the \QWorld, the state right before the adversary's measurement determining the forgery is applied is approximately in the range of $P_\Gamma$, i.e.
\begin{equation}\label{eq:bound-HQ-after-SQ-Wots}
  \bigl\|
  P_\Gamma\ket{\psi_1'}_{M\Sigma XYE\Gamma} - 
  \ket{\psi_1'}_{M\Sigma XYE\Gamma}
  \bigr\|_2
  \leq q_1\delta_W(n) + 4lq_1\epsilon_W(n) 
  = q_1\of[\Big]{\delta_W(n)+2l(w-1)\epsilon_W(n)}.
\end{equation}
\end{lemma}

\begin{proof}
The proof is the same as that of \cref{lem:hash-after-sign}, except that \cref{ComPgamUh:winternitz} is used in place of \cref{commutator of ps and uh}.
\end{proof}

Notice the bound in \cref{eq:bound-HQ-after-SQ-Wots} is small if an adversary makes at most $q_1=o(2^{n/2})$ queries.

Recall that, just like in \cref{sec:sign-queryWots}, we want to analyze the modified $\BlindForge$ experiment, where the $Q$-measurement is applied after the adversary has output a forgery, but before the secret key register is measured to sample the hash chain and verify the forgery. It thus remains to show that due to the fact that $\ket{\psi'_1}$ is approximately in the range of $P_\Gamma$, the outcome $l+1$ only occurs with small probability.
 Given that this part is similar to the analysis of the last part of the $\BlindForge$ experiment made in \cref{sec:HQ-SQ-ltos} after the proof of \cref{lem:hash-after-sign}, we use the same  analysis for the Winternitz OTS.

To that end, we define a new measurement given by projectors $\tilde Q_i$ that performs the $Q$-measurement controlled on the content of the $M$-register, i.e.
\begin{equation*}
    \tilde Q_i=\sum_{m}\proj{m}_M\otimes Q_i^{(m)}.
\end{equation*}
Now, observe that applying the $Q$-measurement after the adversary has output a forgery is equivalent to applying the $\tilde Q$-measurement right before the adversary's measurement that produces the forgery is applied. To prove that, if $m^*\in B$, the outcome $l+1$ occurs only with small probability in the modified $\BlindForge$ experiment, it thus suffices to prove the following lemma.
\begin{lemma}\label{lem:lplusoneisrare-wots}
In the \QWorld, for blinded messages, the outcome $l+1$ only occurs with small probability,
\begin{equation*}
    \left\|\tilde Q_{l+1}\Pi^{B}_M\ket{\psi_1'}_{M\Sigma XYE\Gamma}\right\|_2\le 
    q_1(\delta_W(n)+2l(w-1)\epsilon_W(n))
\end{equation*}
where 
\begin{equation*}
    \Pi^B=\sum_{m\in B}\proj m.
\end{equation*}
\end{lemma}

\begin{proof}
The proof of this lemma is exactly the same as the one of \cref{lem:lplusoneisrare} of the Lamport OTS, ecxept that it uses \cref{lem:hash-after-sign-Wots} instead of \cref{lem:hash-after-sign}.

\end{proof}

Now, all is set to prove \cref{thm:winternitz}.

\begin{proof}[Proof of \cref{thm:winternitz}.]
Just like in \cref{sec:sign-queryWots} we want to bound the final adversary's success probability in winning the $\BlindForge$ game in the general case. Towards that end, we remark that the computations are the same as those done in \cref{sec:HQ-SQ-ltos}. Thus we proceed as in the Lamport OTS by deriving the bound respectively with respect to the modified $\BlindForge$ experiment and the the real experiment. Then, putting all our arguments together, we get the following:
  \begin{align*}
    \Pr_{QI,\BlindForge} [\A \textrm{ succeeds}]
    \leq (l+1)\left(2^{-n} +  q^2(\delta_W(n) + 4l\epsilon_W(n))^2\right). %
\end{align*}
Finally, plugging in the functions $\epsilon_W(n)$ and $\delta_W(n)$ from \cref{Commutator argument1,ComPgamUh:winternitz}, and applying \cref{lem:ind-r-v-qi}, we obtain
\begin{align*}
    \Pr_{\BlindForge} [\A \textrm{ succeeds}]
    &\leq (l+1)\left(2^{-n} +  q^2\left(\frac{8l(w+1)(w-1)}{2^{n/2}}+2l(w-1) \frac{6(w-1)}{2^{n/2}}\right)^2\right)+\frac{3(wl)^2}{2^n}\\
    &\leq  2^{-n}\left[\left(1+q^2l^2(w-1)^2(20w-4)^2\right)(l+1) +3w^2l^2\right]\\
    &\leq 800w^4q^2l^3\cdot 2^{-n}.
\end{align*}
\end{proof}

\section{Tightness}

The notion of blind-unforgeability does not have as close of a relation to the intuitive security property it strives to model as \EUCMA.\footnote{ Indeed, it is a nice exercise to show that an adversary against (say, $q$-time) \EUCMA with success probability $\epsilon$ can be used to construct a \BU-adversary with success probability $\Theta(\epsilon/q)$, and this reduction is tight for efficient adversaries if one-way functions exist.}
The concrete security bounds, however, arguably nevertheless provide an indication of concrete security levels. It is hence an interesting question whether the bounds proven in \cref{sec:lamport,sec:winternitz} are tight. In the following, we present an attack against the \BU security of the Lamport scheme in the QROM, and analyze its success probability, to show that the bound in \cref{thm:lamport} is tight up to a factor $l$ in the number of queries. The attack generalizes to the Winternitz scheme in a straight-forward manner.

We begin by describing the straightforward classical attack based on search. This attack proceeds as follows: To attack the \BU security of the Lamport scheme, choose a blinding probability of $\frac 1 2$. Now make $q$ distinct queries to the random oracle to search for a preimage of one of the $2l$ public key strings. This succeeds with probability
\begin{equation}
    p_{\text{search}}(q)=1-\left(1-\frac{2l}{2^n}\right)^q\ge \frac{2ql}{2^n}.
\end{equation}
Suppose this search succeeded, finding a preimage $y^*$ of $p_{i^*}^{j^*}$. Then chose $m\in\{0,1\}^l$ such that $m_{i^*}=\bar{j^*}$ and query the oracle to obtain a signature  for $m$. This succeeds with probability $1/2$. Now output $m'$ obtained from $m$ by flipping the $i^*$th bit, and $\sigma'$ obtained from $\sigma$ by replacing $\sigma_{i^*}$ with $y^*$. $m'$ is blinded with probability $1/2$, and $y^*$ is equal to the correct secret key string $s_{i^*}^{j^*}$ with constant probability. In summary, the entire attack succeeds with constant probability if 
\begin{equation}
   q=\Omega\left(\frac{2^n}{l}\right).
\end{equation}

It is now easy to see that the search step can be replaced by a Grover search in the QROM. Using the analysis of Grover's algorithm for multiple targets from \cite{brassard1997quantum}, together with a basic analysis of the number of targets (which follows a binomial distribution), it is easy to see that one can achieve a constant success probability if 
\begin{equation}
   q=\Omega\left(\sqrt{\frac{2^n}{l}}\right).
\end{equation}
To compare this result with \cref{thm:lamport}, note that the inequality in \cref{thm:lamport}, \cref{eq:lamport}, implies that to achieve a constant success probability, at least
\begin{equation}
    q\ge C\sqrt{\frac{2^n}{l^3}}
\end{equation} 
are necessary for some constant $C$, i.e. the upper and lower bounds on the number of queries the optimal attack requires indeed differ by a factor of $l$ up to constant factors. For the Winternitz scheme, the bounds differ by a factor of $w^2l$.

\paragraph{Acknowledgements.}
The authors thank Stacey Jeffery for helpful discussions. CM was funded by a NWO VENI grant (Project No. VI.Veni.192.159). CMM deeply thanks the African Institute for Mathematical Science, Quantum Leap Africa Rwanda and QuSoft Amsterdam for their funding and support.
MO was supported by NWO Vidi grant VI.Vidi.192.109.

\bibliographystyle{alphaurl}

\begin{thebibliography}{BHNP{\etalchar{+}}19}
	
	\bibitem[AASA{\etalchar{+}}20]{alagic2020status}
	Gorjan Alagic, Jacob Alperin-Sheriff, Daniel Apon, David Cooper, Quynh Dang,
	John Kelsey, Yi-Kai Liu, Carl Miller, Dustin Moody, Rene Peralta, et~al.
	\newblock Status report on the second round of the {NIST} post-quantum
	cryptography standardization process.
	\newblock {\em US Department of Commerce, {NIST}}, 2020.
	\newblock \href {https://doi.org/10.6028/NIST.IR.8309}
	{\path{doi:10.6028/NIST.IR.8309}}.
	
	\bibitem[AMRS20]{alagic2020quantum}
	Gorjan Alagic, Christian Majenz, Alexander Russell, and Fang Song.
	\newblock Quantum-access-secure message authentication via
	blind-unforgeability.
	\newblock In {\em Annual International Conference on the Theory and
		Applications of Cryptographic Techniques}, pages 788--817. Springer, 2020.
	\newblock URL: \url{https://ia.cr/2018/1150}, \href
	{http://arxiv.org/abs/1803.03761} {\path{arXiv:1803.03761}}, \href
	{https://doi.org/10.1007/978-3-030-45727-3_27}
	{\path{doi:10.1007/978-3-030-45727-3_27}}.
	
	\bibitem[BDF{\etalchar{+}}11]{boneh2011random}
	Dan Boneh, {\"O}zg{\"u}r Dagdelen, Marc Fischlin, Anja Lehmann, Christian
	Schaffner, and Mark Zhandry.
	\newblock Random oracles in a quantum world.
	\newblock In {\em International Conference on the Theory and Application of
		Cryptology and Information Security}, pages 41--69. Springer, 2011.
	\newblock URL: \url{https://ia.cr/2010/428}, \href
	{http://arxiv.org/abs/1008.0931} {\path{arXiv:1008.0931}}, \href
	{https://doi.org/10.1007/978-3-642-25385-0_3}
	{\path{doi:10.1007/978-3-642-25385-0_3}}.
	
	\bibitem[BDH11]{BDH11}
	Johannes Buchmann, Erik Dahmen, and Andreas H{\"u}lsing.
	\newblock {XMSS} - a practical forward secure signature scheme based on minimal
	security assumptions.
	\newblock In Bo-Yin Yang, editor, {\em Post-Quantum Cryptography}, pages
	117--129, Berlin, Heidelberg, 2011. Springer.
	\newblock URL: \url{https://ia.cr/2011/484}, \href
	{https://doi.org/10.1007/978-3-642-25405-5_8}
	{\path{doi:10.1007/978-3-642-25405-5_8}}.
	
	\bibitem[BHK{\etalchar{+}}19]{BHKNRS19}
	Daniel~J. Bernstein, Andreas H\"{u}lsing, Stefan K\"{o}lbl, Ruben Niederhagen,
	Joost Rijneveld, and Peter Schwabe.
	\newblock The {SPHINCS+} signature framework.
	\newblock In {\em Proceedings of the 2019 ACM SIGSAC Conference on Computer and
		Communications Security}, CCS'19, pages 2129--2146, New York, NY, USA, 2019.
	Association for Computing Machinery.
	\newblock URL: \url{https://ia.cr/2019/1086}, \href
	{https://doi.org/10.1145/3319535.3363229}
	{\path{doi:10.1145/3319535.3363229}}.
	
	\bibitem[BHNP{\etalchar{+}}19]{BHNPSS19}
	Xavier Bonnetain, Akinori Hosoyamada, Mar{\'i}a Naya-Plasencia, Yu~Sasaki, and
	Andr{\'e} Schrottenloher.
	\newblock Quantum attacks without superposition queries: The offline {S}imon's
	algorithm.
	\newblock In Steven~D. Galbraith and Shiho Moriai, editors, {\em Advances in
		Cryptology -- ASIACRYPT 2019}, pages 552--583, Cham, 2019. Springer.
	\newblock \href {http://arxiv.org/abs/2002.12439} {\path{arXiv:2002.12439}},
	\href {https://doi.org/10.1007/978-3-030-34578-5_20}
	{\path{doi:10.1007/978-3-030-34578-5_20}}.
	
	\bibitem[BHT98]{brassard1997quantum}
	Gilles Brassard, Peter H{\o}yer, and Alain Tapp.
	\newblock Quantum cryptanalysis of hash and claw-free functions.
	\newblock In Cl{\'a}udio~L. Lucchesi and Arnaldo~V. Moura, editors, {\em
		LATIN'98: Theoretical Informatics}, pages 163--169, Berlin, Heidelberg, 1998.
	Springer.
	\newblock \href {http://arxiv.org/abs/quant-ph/9705002}
	{\path{arXiv:quant-ph/9705002}}, \href {https://doi.org/10.1007/BFb0054319}
	{\path{doi:10.1007/BFb0054319}}.
	
	\bibitem[BLZ20]{blocki2020security}
	Jeremiah Blocki, Seunghoon Lee, and Samson Zhou.
	\newblock On the security of proofs of sequential work in a post-quantum world,
	2020.
	\newblock accepted for publication at ITC 2021.
	\newblock \href {http://arxiv.org/abs/2006.10972} {\path{arXiv:2006.10972}}.
	
	\bibitem[BR93]{Bellare1993a}
	Mihir Bellare and Phillip Rogaway.
	\newblock Random oracles are practical: {A} paradigm for designing efficient
	protocols.
	\newblock In {\em Proceedings of the 1st ACM conference on Computer and
		communications security}, pages 62--73. ACM, 1993.
	\newblock URL:
	\url{https://caislab.kaist.ac.kr/lecture/2010/spring/cs548/basic/B11.pdf},
	\href {https://doi.org/10.1145/168588.168596}
	{\path{doi:10.1145/168588.168596}}.
	
	\bibitem[BZ13]{boneh2013secure}
	Dan Boneh and Mark Zhandry.
	\newblock Secure signatures and chosen ciphertext security in a quantum
	computing world.
	\newblock In Ran Canetti and Juan~A. Garay, editors, {\em Advances in
		Cryptology -- CRYPTO 2013}, pages 361--379, Berlin, Heidelberg, 2013.
	Springer.
	\newblock URL: \url{https://ia.cr/2013/088}, \href
	{https://doi.org/10.1007/978-3-642-40084-1_21}
	{\path{doi:10.1007/978-3-642-40084-1_21}}.
	
	\bibitem[CFHL20]{CFHL20}
	Kai-Min Chung, Serge Fehr, Yu-Hsuan Huang, and Tai-Ning Liao.
	\newblock On the compressed-oracle technique, and post-quantum security of
	proofs of sequential work, 2020.
	\newblock URL: \url{https://ia.cr/2020/1305}, \href
	{http://arxiv.org/abs/2010.11658} {\path{arXiv:2010.11658}}.
	
	\bibitem[EGM96]{even1996line}
	Shimon Even, Oded Goldreich, and Silvio Micali.
	\newblock On-line/off-line digital signatures.
	\newblock {\em Journal of Cryptology}, 9(1):35--67, 1996.
	\newblock \href {https://doi.org/10.1007/0-387-34805-0_24}
	{\path{doi:10.1007/0-387-34805-0_24}}.
	
	\bibitem[GHHM20]{GHHM20}
	Alex~B. Grilo, Kathrin Hövelmanns, Andreas Hülsing, and Christian Majenz.
	\newblock Tight adaptive reprogramming in the {QROM}, 2020.
	\newblock URL: \url{https://ia.cr/2020/1361}, \href
	{http://arxiv.org/abs/2010.15103} {\path{arXiv:2010.15103}}.
	
	\bibitem[GHS16]{GHS16}
	Tommaso Gagliardoni, Andreas H{\"u}lsing, and Christian Schaffner.
	\newblock Semantic security and indistinguishability in the quantum world.
	\newblock In Matthew Robshaw and Jonathan Katz, editors, {\em Advances in
		Cryptology -- CRYPTO 2016}, pages 60--89, Berlin, Heidelberg, 2016. Springer.
	\newblock URL: \url{https://ia.cr/2015/355}, \href
	{http://arxiv.org/abs/1504.05255} {\path{arXiv:1504.05255}}, \href
	{https://doi.org/10.1007/978-3-662-53015-3_3}
	{\path{doi:10.1007/978-3-662-53015-3_3}}.
	
	\bibitem[GKS20]{GKS20}
	Tommaso Gagliardoni, Juliane Krämer, and Patrick Struck.
	\newblock Quantum indistinguishability for public key encryption, 2020.
	\newblock URL: \url{https://ia.cr/2020/266}, \href
	{http://arxiv.org/abs/2003.00578} {\path{arXiv:2003.00578}}.
	
	\bibitem[GMR88]{goldwasser1988digital}
	Shafi Goldwasser, Silvio Micali, and Ronald~L Rivest.
	\newblock A digital signature scheme secure against adaptive chosen-message
	attacks.
	\newblock {\em SIAM Journal on computing}, 17(2):281--308, 1988.
	\newblock \href {https://doi.org/10.1137/0217017} {\path{doi:10.1137/0217017}}.
	
	\bibitem[GYZ17]{GYZ17}
	Sumegha Garg, Henry Yuen, and Mark Zhandry.
	\newblock New security notions and feasibility results for authentication of
	quantum data.
	\newblock In Jonathan Katz and Hovav Shacham, editors, {\em Advances in
		Cryptology -- CRYPTO 2017}, pages 342--371, Cham, 2017. Springer.
	\newblock URL: \url{https://eprint.iacr.org/2017/538.pdf}, \href
	{http://arxiv.org/abs/1607.07759} {\path{arXiv:1607.07759}}, \href
	{https://doi.org/10.1007/978-3-319-63715-0_12}
	{\path{doi:10.1007/978-3-319-63715-0_12}}.
	
	\bibitem[Hay02]{hayashi2002optimal}
	Masahito Hayashi.
	\newblock Optimal sequence of quantum measurements in the sense of {S}tein's
	lemma in quantum hypothesis testing.
	\newblock {\em Journal of Physics A: Mathematical and General}, 35(50):10759,
	2002.
	\newblock \href {http://arxiv.org/abs/quant-ph/0208020}
	{\path{arXiv:quant-ph/0208020}}, \href
	{https://doi.org/10.1088/0305-4470/35/50/307}
	{\path{doi:10.1088/0305-4470/35/50/307}}.
	
	\bibitem[HBG{\etalchar{+}}18]{RFC8391}
	Andreas H{\"u}lsing, Denise Butin, Stefan-Lukas Gazdag, Joost Rijneveld, and
	Aziz Mohaisen.
	\newblock {XMSS}: {E}xtended hash-based signatures. {RFC} 8391, 2018.
	\newblock URL: \url{https://tools.ietf.org/html/rfc8391}, \href
	{https://doi.org/10.17487/RFC8391} {\path{doi:10.17487/RFC8391}}.
	
	\bibitem[KLLNP16]{KLLNP16}
	Marc Kaplan, Ga{\"e}tan Leurent, Anthony Leverrier, and Mar{\'i}a
	Naya-Plasencia.
	\newblock Breaking symmetric cryptosystems using quantum period finding.
	\newblock In Matthew Robshaw and Jonathan Katz, editors, {\em Advances in
		Cryptology -- CRYPTO 2016}, pages 207--237, Berlin, Heidelberg, 2016.
	Springer.
	\newblock \href {http://arxiv.org/abs/1602.05973} {\path{arXiv:1602.05973}},
	\href {https://doi.org/10.1007/978-3-662-53008-5_8}
	{\path{doi:10.1007/978-3-662-53008-5_8}}.
	
	\bibitem[Lam79]{Lam79}
	Leslie Lamport.
	\newblock Constructing digital signatures from a one way function.
	\newblock Technical Report SRI-CSL-98, SRI International Computer Science
	Laboratory, 1979.
	\newblock URL: \url{http://lamport.azurewebsites.net/pubs/dig-sig.pdf}.
	
	\bibitem[LZ19]{LZ19}
	Qipeng Liu and Mark Zhandry.
	\newblock On finding quantum multi-collisions.
	\newblock In Yuval Ishai and Vincent Rijmen, editors, {\em Advances in
		Cryptology -- EUROCRYPT 2019}, pages 189--218, Cham, 2019. Springer.
	\newblock \href {http://arxiv.org/abs/1811.05385} {\path{arXiv:1811.05385}},
	\href {https://doi.org/10.1007/978-3-030-17659-4_7}
	{\path{doi:10.1007/978-3-030-17659-4_7}}.
	
	\bibitem[Mer89]{merkle1989certified}
	Ralph~C. Merkle.
	\newblock A certified digital signature.
	\newblock In {\em Conference on the Theory and Application of Cryptology},
	pages 218--238. Springer, 1989.
	\newblock \href {https://doi.org/10.1007/0-387-34805-0_21}
	{\path{doi:10.1007/0-387-34805-0_21}}.
	
	\bibitem[NC02]{nielsen2002quantum}
	Michael~A. Nielsen and Isaac Chuang.
	\newblock Quantum computation and quantum information, 2002.
	\newblock URL:
	\url{http://csis.pace.edu/~ctappert/cs837-19spring/QC-textbook.pdf}, \href
	{https://doi.org/10.1023/A:1012603118140}
	{\path{doi:10.1023/A:1012603118140}}.
	
	\bibitem[Sho94]{shor1994algorithms}
	Peter~W. Shor.
	\newblock Algorithms for quantum computation: discrete logarithms and
	factoring.
	\newblock In {\em Proceedings 35th annual symposium on foundations of computer
		science}, pages 124--134. IEEE, 1994.
	\newblock URL:
	\url{http://citeseerx.ist.psu.edu/viewdoc/download?doi=10.1.1.123.5183&rep=rep1&type=pdf},
	\href {https://doi.org/10.1109/SFCS.1994.365700}
	{\path{doi:10.1109/SFCS.1994.365700}}.
	
	\bibitem[SS17]{SS17}
	Thomas Santoli and Christian Schaffner.
	\newblock Using {S}imon's algorithm to attack symmetric-key cryptographic
	primitives.
	\newblock {\em Quantum Info. Comput.}, 17(1–2):65--78, February 2017.
	\newblock \href {http://arxiv.org/abs/1603.07856} {\path{arXiv:1603.07856}},
	\href {https://doi.org/10.26421/QIC17.1-2-4}
	{\path{doi:10.26421/QIC17.1-2-4}}.
	
	\bibitem[Wat18]{watrous2018theory}
	John Watrous.
	\newblock {\em The theory of quantum information}.
	\newblock Cambridge University Press, 2018.
	\newblock URL: \url{https://cs.uwaterloo.ca/~watrous/TQI/TQI.pdf}.
	
	\bibitem[Zha15]{zhandry2015secure}
	Mark Zhandry.
	\newblock Secure identity-based encryption in the quantum random oracle model.
	\newblock {\em International Journal of Quantum Information}, 13(04):1550014,
	2015.
	\newblock URL: \url{https://ia.cr/2012/076}, \href
	{https://doi.org/10.1142/S0219749915500148}
	{\path{doi:10.1142/S0219749915500148}}.
	
	\bibitem[Zha19]{Zhandry19}
	Mark Zhandry.
	\newblock How to record quantum queries, and applications to quantum
	indifferentiability.
	\newblock In Alexandra Boldyreva and Daniele Micciancio, editors, {\em Advances
		in Cryptology -- CRYPTO 2019}, pages 239--268, Cham, 2019. Springer.
	\newblock URL:
	\url{https://www.cs.princeton.edu/~mzhandry/docs/papers/QIndiff.pdf}, \href
	{https://doi.org/10.1007/978-3-030-26951-7_9}
	{\path{doi:10.1007/978-3-030-26951-7_9}}.
	
\end{thebibliography}
\newcommand{\etalchar}[1]{$^{#1}$}

\appendix
\section{Lemmas used for proving the security of the Lamport OTS}\label{apx:Lamport lemmas}
\subsection{Proof of \cref{Commutator argument1,lem:Commutator}}\label{apx:L8andL15}

Here, we prove that the unitary $U_h$ does not significantly modify the secret key register. Towards that end, we show that the random oracle unitary $U_h$ approximately commute with the projector onto the relevant parts of the secret key register being in state $\ket{\Phi}$.

\LemComWin*

\begin{proof}
Recall from \cref{hash oracles} that $U_h$ compares the input in register $X$ to the secret key in $\Gamma$, and stores the output in register $Y$.
Recall from \cref{eq:Uh} that $U_h$ is defined as follows:
\begin{equation*}
    (U_h)_{XY\Gamma} =
    \of*{
        \prod_{i=1}^l
        \prod_{j=0}^{w-2}
        (U_i^j)_{XY\Gamma_i^j}
    }
    U^{\neq}_{XY\Gamma}.
\end{equation*}
Using \cref{lem:commutators}%
, we get
\begin{equation}
    \norm*{\sof*{
        U_h, \Phi_{\Gamma_{i'}^{\leq j'}}
    }}_\infty
    \leq
    \sum_{i=1}^l
    \sum_{j=0}^{w-2}
    \norm*{\sof*{
        U_i^j, \Phi_{\Gamma_{i'}^{\leq j'}}
    }}_\infty
    + \norm*{\sof*{
        U^{\neq}_{XY\Gamma},
        \Phi_{\Gamma_{i'}^{\leq j'}}
    }}_\infty.\label{eq:ineq1}
\end{equation}
We will bound the two terms separately.

To deal with the first term, recall from \cref{Uij} that
\begin{equation}
    (U^j_i)_{XY\Gamma^j_i\Gamma^{j+1}_i}
= P^=_{X\Gamma^j_i} \x \of[\Big]{
  (\CNOT\xp{n})_{\Gamma^{j+1}_i:Y}
  - \1
 }
+ \1,
\end{equation}
so the expression in first term of \cref{eq:ineq1} is given by
\begin{align}
    \norm*{\sof*{
        U_i^j, \Phi_{\Gamma_{i'}^{\leq j'}}
    }}_\infty
    &= \norm*{\sof*{
        P^=_{X\Gamma^j_i} \x \of*{
      (\CNOT\xp{n})_{\Gamma^{j+1}_i:Y}
      - \1
     }, \Phi_{\Gamma_{i'}^{\leq j'}}
    }}_\infty.
    \label{eq:ineq2}
\end{align}
If $i \neq i'$ or $j' < j$, the commutator vanishes because on every register at least one of the two operators acts as the identity matrix.
Let us assume now that $i = i'$ and $j'\ge j$. Then we can upper bound the norm in \cref{eq:ineq2} by noting that $\Phi_{\Gamma_{i}^{\leq j'}} = \Phi_{\Gamma_i^j} \x \Phi'$ where $\Phi'$ is a projector on the remaining registers, and then separating $X\Gamma_i^j$ out from the remaining registers.
Starting with the triangle inequality and then using this observation,
\begin{align*}
    \norm*{\sof*{
        U_i^j, \Phi_{\Gamma_{i}^{\leq j'}}
    }}_\infty
    &\le 2
    \norm*{
        \of*{
            P^=_{X\Gamma^j_i} \x
            \of*{(\CNOT\xp{n})_{\Gamma^{j+1}_i:Y} - \1}
        } \cdot
        \Phi_{\Gamma_{i}^{\leq j'}}
    }_\infty \\
    &\le 2
    \norm*{P^=_{X\Gamma^j_i}\Phi_{\Gamma_i^j}}_\infty
    \norm*{(\CNOT\xp{n})_{\Gamma^{j+1}_i:Y} - \1}_\infty
    \norm*{\Phi'}_\infty \\
    &\le 4
    \norm*{P^=_{X\Gamma^j_i}\Phi_{\Gamma_i^j}}_\infty \\
    &\le 4 \cdot 2^{-n/2},
\end{align*}
where we used $\norm{(\CNOT\xp{n})_{\Gamma^{j+1}_i:Y} - \1}_\infty \leq 2$ and $\norm{\Phi'}_\infty \leq 1$ to obtain the third inequality, and then \cref{lem:PQ} for the last inequality.
In summary,
\begin{equation}
    \norm*{\sof*{
        U_i^j, \Phi_{\Gamma_{i'}^{\leq j'}}
    }}_\infty
    \le \delta_{i,i'} \delta_{j' \geq j} \cdot
    4 \cdot 2^{-n/2},
\end{equation}
so the first term in \cref{eq:ineq1} can be upper bounded as
\begin{equation}
    \sum_{i=1}^l
    \sum_{j=0}^{w-2}
    \norm*{\sof*{
        U_i^j, \Phi_{\Gamma_{i'}^{\leq j'}}
    }}_\infty
    \le 4(j'+1) \cdot 2^{-n/2}.
    \label{eq:term1}
\end{equation}

To upper bound the second term in \cref{eq:ineq1}, recall from \cref{eq:Uneq} that
$U^{\neq}_{XY\Gamma}
= P^{\neq}_{X\Gamma} \cdot \of*{U'_{XY} - \1} + \1$.
Using \cref{lem:commutators}, we can inductively expand the projectors $P^{\neq}_{X\Gamma}$ and $\Phi_{\Gamma_{i'}^{\leq j'}}$ defined in \cref{eq:Pneq,eq:R} to get
\begin{align}
    \norm*{\sof*{
        U^{\neq}_{XY\Gamma},
        \Phi_{\Gamma_{i'}^{\leq j'}}
    }}_\infty
    &=
    \norm*{\sof*{
        P^{\neq}_{X\Gamma} \cdot \of*{U'_{XY} - \1},
        \Phi_{\Gamma_{i'}^{\leq j'}}
    }}_\infty \nonumber \\
    &\leq
    \sum_{i=1}^l
    \sum_{j=0}^{w-2}
    \sum_{k=0}^{j'}
    \norm*{\sof*{
        P^{\neq}_{X\Gamma_i^j} \cdot \of*{U'_{XY} - \1},
        \Phi_{\Gamma_{i'}^k}
    }}_\infty. \label{eq:ijk}
\end{align}
We can simplify this further by using \cref{lem:commutators} once again:
\begin{align*}
    \norm*{\sof*{
        P^{\neq}_{X\Gamma_i^j} \cdot \of*{U'_{XY} - \1},
        \Phi_{\Gamma_{i'}^k}
    }}_\infty
    &\leq
    \norm*{\sof*{
        P^{\neq}_{X\Gamma_i^j},
        \Phi_{\Gamma_{i'}^k}
    }}_\infty
    +
    \norm*{\sof*{
        U'_{XY} - \1,
        \Phi_{\Gamma_{i'}^k}
    }}_\infty.
\end{align*}
The second term vanishes since the query operation $U'_{XY}$ from \cref{eq:U'} acts trivially on $\Gamma_{i'}^k$.
Furthermore, we can replace $P^{\neq}$ by $P^=$ in the first term since the two operators differ by $\1$, which commutes with everything.
With these observations, \cref{eq:ijk} simplifies to
\begin{equation*}
    \norm*{\sof*{
        U^{\neq}_{XY\Gamma},
        \Phi_{\Gamma_{i'}^{\leq j'}}
    }}_\infty
    \leq
    \sum_{i=1}^l
    \sum_{j=0}^{w-2}
    \sum_{k=0}^{j'}
    \norm*{\sof*{
        P^=_{X\Gamma_i^j},
        \Phi_{\Gamma_{i'}^k}
    }}_\infty.
\end{equation*}
Using \cref{lem:PQ}, we obtain the following bound on each term:
\begin{equation*}
    \norm*{\sof*{
        P^=_{X\Gamma_i^j},
        \Phi_{\Gamma_{i'}^k}
    }}_\infty
    \leq
    2
    \delta_{i,i'}
    \delta_{j,k}
    \cdot 2^{-n/2}.
\end{equation*}
Putting these together, the second term in \cref{eq:ineq1} can be upper bounded as
\begin{equation*}
    \norm*{\sof*{
        U^{\neq}_{XY\Gamma},
        \Phi_{\Gamma_{i'}^{\leq j'}}
    }}_\infty
    \leq
    2 (j'+1) \cdot 2^{-n/2}.
\end{equation*}

Combining this with \cref{eq:term1}, both terms in \cref{eq:ineq1} can be upper bounded as
\begin{equation}
    \norm*{\sof*{
        U_h, \Phi_{\Gamma_{i'}^{\leq j'}}
    }}_\infty
    \leq
    6(j'+1)2^{-n/2}
    \leq
    6(w-1) \cdot 2^{-n/2}
\end{equation}
where we used $j' \leq w-2$.
\end{proof}

As a corollary, we obtain the corresponding lemma for the Lamport OTS by setting $w=2$.

\LemCom*

\subsection{Proof of \cref{Orthogonality argument}}\label{apx:QPS}

\LemOrth*

\begin{proof}
Recall from \cref{eq:Q,eq:invariant projector,eq:Phi(a)} that
\begin{align}
    Q_{l+1}^{m^*}
    &= \bigotimes_{i=1}^l
    \Phi^\perp_{S_i^{m^*_i}}, &
    P_S
    &= \sum_{\alpha \in \widehat{B^c}} \Phi(\alpha)_S
    = \sum_{\alpha \in \widehat{B^c}}
    \bigotimes_{i=1}^l
    \bigotimes_{j=0}^1
    \Phi(\alpha_i^j)_{S^j_i},
\end{align}
where
$\Phi(0) = \proj{\Phi} = \Phi$ and
$\Phi(1) = \1 - \proj{\Phi} = \Phi^\perp$,
and the set $\widehat{B^c}$ was defined in \cref{def of B^c hat} as
\begin{align}
    \widehat{B^c} = \bigcup_{m \in B^c}
    \set[\Big]{
        \alpha \in \set{0,1}^{2l}
        \mathrel{\Big|}
        \text{$\alpha_i^{\bar{m}_i} = 0$ for all $i=1,\dotsc,l$}
    }.
    \label{eq:Bchat2}
\end{align}

Substituting the expression for $Q^{m^*}_{l+1}$,
\begin{align}
    Q^{m^*}_{l+1} P_S
  &=
    \sum_{\alpha\in \widehat{B^c}}\of[\Big]{
        \of[\big]{(\Phi^\perp)\xp{l}}_{S_1^{m^*_1} \cdots S_l^{m^*_l}}\otimes
        \1_{S_1^{\bar{m}_1^*}
        \cdots S_l^{\bar{m}_l^*}}
    }
    \Phi(\alpha).
    \label{eq:ortho}
\end{align}
Consider a single term $\alpha \in \widehat{B^c}$ in the sum. By definition of $\widehat{B^c}$, there exists a message $m\in B^c$ such that $\alpha_i^{\bar{m}_i}=0$ for all $i$. But as $B\ni m^*\neq m\in B^c$, there exists an index $i$ such that $m_i^*\neq m_i$, or equivalently $m_i^*=\bar m_i$. Therefore we can simplify the corresponding term to
\begin{align*}
&\of[\Big]{
        \of[\big]{(\Phi^\perp)\xp{l}}_{S_1^{m^*_1} \cdots S_l^{m^*_l}}\otimes
        \1_{S_1^{\bar{m}_1^*}
        \cdots S_l^{\bar{m}_l^*}}
    }
    \Phi(\alpha)\\
  &=\of[\Big]{
        \Phi^\perp_{S_1^{m_1^*}}\otimes\cdots\otimes
        \Phi^\perp_{S_i^{m_i^*}}\otimes\cdots\otimes
        \Phi^\perp_{S_l^{m_l^*}}\otimes
        \1_{S_1^{\bar{m}_1^*} \cdots S_l^{\bar{m}_l^*}}
    }\\
    &\quad\,\,
    \of[\Big]{
        \Phi_{S_1^{\bar{m}_1}}
        \otimes\cdots\otimes
        \Phi_{S_i^{\bar{m}_i}}
        \otimes\cdots\otimes
        \Phi_{S_l^{\bar{m}_l}}
        \otimes
        \Phi(\alpha_1^{m_1})_{S_1^{m_1}}\otimes \cdots\otimes  \Phi(\alpha_l^{m_l})_{S_l^{m_l}}
        }
    \\
    &=
    \cdots\otimes
        \underbrace{
        (\Phi^\perp\Phi)_{S _i^{m^*_i}}
        }_{=0}\otimes\cdots
    \\
   &=0.
\end{align*}
Following the same analysis for all the terms in \cref{eq:ortho} leads to the required result.

\end{proof}

\subsection{Proof of \cref{final state invariant if no hash queries}}\label{apx:PSB}

Here, we show that if there are no hash queries before the $\Sign$ query, projecting $P_S$ onto the adversary's state after the $\Sign$ query leaves it completely in the range of $P_S$.

\LemNoHash*

\begin{proof}
We know that all secret key sub-registers are in state $\ket{\Phi}$ if there are no hash queries before $\Sign$ query, i.e.,
\[\ket{\psi_0}_{ M\Sigma EBS}
-\Phi\xp{2l}_S\ket{\psi_0}_{M\Sigma EBS}
=0\]
where $\ket{\psi_0}_{M\Sigma EBS}$ is the adversary's state immediately before the $\Sign$ query.
Let us denote the state after a single $\Sign$ query by
\begin{equation}
  \ket{\psi_1}_{M\Sigma EBS}
  = B\Sign_{\sk}
    \ket{\psi_0}_{M\Sigma EBS}.
\end{equation}
We want to show that applying $P_S$ onto the state $\ket{\psi_1}_{M\Sigma EBS}$ leaves it invariant. That is,
\begin{equation}
    \bigl\|\ket{\psi_1}_{M\Sigma EBS}
    -P_S\ket{\psi_1}_{M\Sigma EBS}\bigr\|_2
    = 0.
\end{equation}
Recall that
\begin{align}
    P_S\ket{\psi_1}_{ M\Sigma EBS}
    &=P_SB
    \Sign_{\sk}
    \ket{\psi_0}_{M\Sigma EBS}\\
    &=P_SB
    \Sign_{\sk}
    \ket{\psi_0^1}_{M\Sigma EBS} + 
    P_SB
    \Sign_{\sk}
    \ket{\psi_0^0}_{M\Sigma EBS},
    \label{P_Sapplied to post signature state}
\end{align}
where the superscripts $1$ and $0$ refer to blinded ($B$) and un-blinded ($B^c$) messages, respectively.
Recall from \cref{eq:psi0-L} that:
\begin{align}
    \ket{\psi_0^1}_{M\Sigma BES}
    &
  = \left(\sum_{m\in B}
    \sum_{\sigma\in(\{0,1\}^n)^l}
    \kappa_{m\sigma 1}
    \ket{m}_M
    \ket{\sigma}_{\Sigma}
    \ket{1}_B 
    \ket{\alpha_{m\sigma 1}}_E\right)
    \otimes
    \left(\ket{\Phi}^{\otimes 2l}\right)_S%
    \\
    \ket{\psi_0^0}_{M\Sigma BES}
    &
     = \left(\sum_{m\in B^c}
    \sum_{\sigma\in(\{0,1\}^n)^l}
    \kappa_{m\sigma 0}
    \ket{m}_M
    \ket{\sigma}_{\Sigma}
    \ket{0}_B \ket{\alpha_{m\sigma 0}}_E\right)
    \otimes
    \left(\ket{\Phi}^{\otimes 2l}\right)_S
\end{align}
We start by deriving the first term  of the right-hand side of \cref{P_Sapplied to post signature state}. If there are no hash queries before the $\Sign$ query, the secret key register is fully in uniform superposition state $\ket{\Phi}$ before the single $\Sign$ query. Since we are dealing with the first term that corresponds to blinded messages, there is no signature. Thus, the joint state of all secret key registers is in the range of the projector $\Phi$ and
\begin{align}
   P_SB
   \Sign_{\sk}
   \ket{\psi_0^1}_{M\Sigma EBS}
   =\ket{\psi_0^1}_{M\Sigma EBS}.
   \label{Ps leave psi01 invariant}
\end{align}

Next, we evaluate the second term of \cref{P_Sapplied to post signature state} that corresponds to un-blinded messages.
We start by setting
\begin{align*}
    \ket{\tau}_{M\Sigma BES} 
    &= P_SB
    \Sign_{\sk}
    \ket{\psi_0^0}_{M\Sigma EBS} \\
    &= \sum_{m\in B^c}
    P_S
    B
    \Sign_{\sk}
    \ket{\gamma(m)}_{\Sigma EB}
    \ket{\Phi}^{\otimes 2l}_S
\end{align*}
where
\begin{align*}
    \ket{\gamma(m)}_{\Sigma EB}
   = \sum_{\sigma\in(\{0,1\}^n)^l}
     \kappa_{m\sigma 0}
    \ket{m}_M
    \ket{\sigma}_{\Sigma}
    \ket{0}_B \ket{\alpha_{m\sigma}}_E.
\end{align*}

From the definition of $\widehat{B^c}$ in \cref{eq:Bchat2} we know that for any message $m \in B^c$ there exists an $i \in \set{1,\dotsc,l}$ for which the secret key register $S^{\bar{m}_i}_i$ corresponding to $\bar{m}_i$ is in the range of the projector $\Phi$. Since the projector $P_S = \sum_{\alpha \in \widehat{B^c}} \Phi(\alpha)_S$ acts on the entire secret key register, we can split it into two sums where the first is over all the projectors for which $\alpha^{\bar{m}_i}_i = 0$ (i.e., $S^{\bar{m}_i}_i$ is in the range of $\Phi(0) = \Phi$) and the second is over the remaining projectors ($\alpha_i^{\bar{m}_i} = 1$ for at least one $i$). Accordingly, for any message $m \in B^c$, we can split $P_S$ as
\begin{equation}\label{eq:decomposition}
    P_S
  = \sum_{\alpha \in \widehat{B^c}} \Phi(\alpha)
  = \sum_{\substack{\alpha \in \widehat{B^c} \\
    \forall i: \alpha_i^{\bar{m}_i} = 0}}
    \Phi(\alpha) +
    \sum_{\substack{\alpha \in \widehat{B^c} \\
    \exists i: \alpha_i^{\bar{m}_i} = 1}}
    \Phi(\alpha)
  = P_S^1(m) + P_S^2(m).
\end{equation}
Note that this split can differ from one message $m \in B^c$ to another.
Therefore, $\ket{\tau}_{M\Sigma BES}$ becomes
\begin{align}
    \ket{\tau}_{M\Sigma BES}
  &=\sum_{m\in B^c}
    \of[\big]{P_S^1(m) + P_S^2(m)}
    B\Sign_{\sk}
    \ket{\gamma(m)}_{\Sigma EB}
    \ket{\Phi}^{\otimes 2l}_S.
    \label{equation superposition of messages}
\end{align}
For a fixed message $m$, $P_S^1(m)$ can be written as 
\[P_S^1(m)
  = \sum_{\substack{\alpha \in \widehat{B^c} \\
    \forall i: \alpha_i^{\bar{m}_i} = 0}}
    \Phi(\alpha)
=\1_{S^{m_1}_1\cdots S^{m_l}_l}\otimes 
(\Phi\xp{l})_{S^{\bar{m}_1}_1\cdots S^{\bar{m}_l}_l}.\]
Thus,
\begin{align}
   & 
    P_S^1(m)
    B
    \Sign_{\sk}
    \ket{\gamma(m)}_{\Sigma EB}
    \ket{\Phi}^{\otimes 2l}_S\nonumber\\
  &=
    B
    \Sign_{\sk}
    \ket{\gamma (m)}_{\Sigma EB}
    \of[\Big]{
        \1_{S^{m_1}_1\cdots S^{m_l}_l}\otimes 
        (\Phi\xp{l})_{S^{\bar{m}_1}_1\cdots S^{\bar{m}_l}_l}
    }
    \ket{\Phi}^{\otimes 2l}_S\nonumber\\
  &=
    B
    \Sign_{\sk}
    \ket{\gamma (m)}_{\Sigma EB}
    \ket{\Phi}^{\otimes 2l}_S\label{eq:fixed-message-Ps1-invariant}
\end{align}
where the second equality follows because $P_S^1(m)$ and $
B\Sign_{\sk}$ act on different sub-registers. More precisely, $P_S^1(m)$ acts on the secret key registers corresponding to the complementary message $\bar{m}$ while $B\Sign_{\sk}$ acts on the secret key registers corresponding to the message $m$, i.e.
\[[P_S^1(m),B\Sign_{\sk}]\ket m_M = 0.\]
This is only true for the fixed message $m$.

Note that $P_S^1(m)$ and $P_S^2(m)$ are orthogonal, i.e.
\[P_S^2(m)P_S^1(m) = 0.\]
So \cref{eq:fixed-message-Ps1-invariant} implies
\[P_S^2(m)B
    \Sign_{\sk}
    \ket{\gamma(m)}_{\Sigma EB}
    \ket{\Phi}^{\otimes 2l}_S=P_S^2(m)P_S^1(m)B
    \Sign_{\sk}
    \ket{\gamma(m)}_{\Sigma EB}
    \ket{\Phi}^{\otimes 2l}_S=0.\]
Adding to this the result obtained in \cref{Ps leave psi01 invariant}, we get
\[P_S
    B
    \Sign_{\sk}
    \ket{\gamma(m)}_{\Sigma EB}
    \ket{\Phi}^{\otimes 2l}_S =~
    B
    \Sign_{\sk}
    \ket{\gamma(m)}_{\Sigma EB}
    \ket\Phi^{\otimes 2l}_S.\]
Hence,
\begin{align}
    \bigl\|\ket{\psi_1}_{M\Sigma EBS}
-P_S\ket{\psi_1}_{M\Sigma EBS}\bigr\|_2
= 0 \label{P_S leave invariant psi1}
\end{align}
by \cref{equation superposition of messages}.
\end{proof}

\subsection{Proof of \cref{commutator of ps and uh}}\label{apx:UhPS}

\LemComP*

\begin{proof}
Recall from \cref{eq:Uh} that $U_h$ is defined as follows:
\begin{equation}
    (U_h)_{XYS} =
    \of*{
        \prod_{i=1}^l
        \prod_{j=0}^1
        (U_i^j)_{XYS^j_i}
    }
    U^{\neq}_{XYS}.
    \label{eq:Uh Lamport}
\end{equation}
Similar to \cref{Uij}, the unitary $U^j_i$ is defined as
\begin{equation}
    (U^j_i)_{XY\Gamma^j_i\Gamma^{j+1}_i}
    = P^=_{X\Gamma^j_i} \x \of[\Big]{
      (\CNOT\xp{n})_{\Gamma^{j+1}_i:Y}
      - \1
     }
    + \1,
    \label{eq:U_i^j}
\end{equation}
where for the Lamport OTS $j=0$ and $\Gamma_i^{j+1}$ represents the register that stores the $i$-th $n$-bit string of the public key which is described in \cref{sec:LOTS}. More precisely, in the Lamport OTS, the pair of indices $(i,j)$ of the secret and public keys takes the role of the index $i$ in the Winternitz scheme, and whether a string is part of the secret key ($j=0$) or the public key ($j=1$) in the Lamport scheme determines the value of the index $j$ from the Winternitz scheme.
If we introduce a register $P^j_i$ that stores the $(i,j)$-th block of the public key, \cref{eq:U_i^j} for the Lamport OTS becomes
\begin{equation}
    (U_{i,j})_{XY S^j_i P^j_i}
    = P^=_{XS^j_i} \x \of[\Big]{
      (\CNOT\xp{n})_{P^j_i:Y}
      - \1
     }
    + \1,
    \label{eq:Uij for Lamport}
\end{equation}
Also, similar to \cref{eq:Uneq},
\begin{equation}
    U^{\neq}_{XYS}
    = P^{\neq}_{XS} U'_{XY}
    + \of[\big]{
        \1_{XS}
        - P^{\neq}_{XS}
    } \x \1_{Y}.
    \label{eq:UXYS}
\end{equation}

By substituting the formula for $U_h$ from \cref{eq:Uh Lamport}, we get
\begin{align}
    \bigl\|[U_h, P_S]\bigr\|_\infty
    \label{eq:firstofmany}
    &= \biggl\|\bigg[\Big(
    \prod_{\substack{i=1,\dotsc,l\\j=0,1}}
    (U_{i,j})_{XYS^j_iP^j_i}
    \Big)
    U^{\neq}_{XYS}, 
    P_S\bigg]\biggr\|_\infty\\
    &= \biggl\|\bigg[
    \prod_{\substack{i=1,\dotsc,l\\j=0,1}}
    (U_{i,j})_{XYS^j_iP^j_i},P_S\bigg]
    U^{\neq}_{XYS} +
    \prod_{\substack{i=1,\dotsc,l\\j=0,1}}
    (U_{i,j})_{XYS^j_iP^j_i}
    \bigg[U^{\neq}_{XYS},P_S\bigg]
    \biggr\|_\infty\\ 
    &\leq \biggl\|\bigg[
    \prod_{\substack{i=1,\dotsc,l\\j=0,1}}
    (U_{i,j})_{XYS^j_iP^j_i},P_S\bigg]
    U^{\neq}_{XYS} \biggr\|_\infty+
    \biggl\|
    \prod_{\substack{i=1,\dotsc,l\\j=0,1}}
    (U_{i,j})_{XYS^j_iP^j_i}
    \bigg[U^{\neq}_{XYS},P_S\bigg]
    \biggr\|_\infty\\
    &\leq \biggl\|\bigg[
    \prod_{\substack{i=1,\dotsc,l\\j=0,1}}
    (U_{i,j})_{XYS^j_iP^j_i},P_S\bigg]
    \biggr\|_\infty
    \biggl\|
    U^{\neq}_{XYS} \biggr\|_\infty+
    \biggl\|
    \prod_{\substack{i=1,\dotsc,l\\j=0,1}}
    (U_{i,j})_{XYS^j_iP^j_i}\biggr\|_\infty
    \biggl\|
    \biggl[U^{\neq}_{XYS},P_S\bigg]
    \biggr\|_\infty,
\end{align}
where the first equality follows from the definition of $U_h$, the second follows from \cref{eq:ABC}, and the two inequalities follow respectively from the triangle inequality and the sub-multiplicative property of the operator norm.
Since $U^{\neq}_{XYS}$ is an unitary,
$\|U^{\neq}_{XYS}\|_\infty = 1$.
Similarly,
\[\biggl\|
\prod_{\substack{i=1,\dotsc,l\\j=0,1}}
(U_{i,j})_{XYS^j_iP^j_i}
\biggr\|_\infty
\leq
\prod_{\substack{i=1,\dotsc,l\\j=0,1}}
\Bigl\|(U_{i,j})_{XYS^j_iP^j_i}\Bigr\|_\infty
\leq 1\]
and hence
\begin{equation*}
    \bigl\|[U_h, P_S]\bigr\|_\infty
    \le\biggl\|\biggl[
    \prod_{\substack{ i=1,\dotsc,l\\j=0, 1}}(U_{i,j})_{XYS^j_iP^j_i},P_S\biggr]
    \biggr\|_\infty+
    \biggl\|
    \biggl[U^{\neq}_{XYS},P_S\bigg]
    \biggr\|_\infty.
\end{equation*}
Substituting $U^{\neq}_{XYS}$ from \cref{eq:UXYS} and using the triangle inequality lead to
\begin{align}
    \bigl\|[U_h, P_S]\bigr\|_\infty
    &\leq
    \biggl\|\bigg[
        \prod_{\substack{i=1,\dotsc,l\\j=0,1}}
        (U_{i,j})_{XYS^j_iP^j_i},
        P_S
    \bigg] 
    \biggr\|_\infty
    +
    \biggl\|
    \bigg[
        P^{\neq}_{XS} U'_{XY}
        + \of[\big]{\1_{XS} - P^{\neq}_{XS}} \x \1_{Y},
        P_S
    \bigg]
    \biggr\|_\infty\\
    &\leq
    \biggl\|\bigg[
        \prod_{\substack{i=1,\dotsc,l\\j=0,1}}
        (U_{i,j})_{XYS^j_iP^j_i},P_S
    \bigg] 
    \biggr\|_\infty
    +
    \biggl\|
    \Big[P^{\neq}_{XS}U'_{XY},P_S\Big]
    \biggr\|_\infty
    +
    \biggl\|
    \Big[
        \of[\big]{\1_{XS} - P^{\neq}_{XS}} \x \1_{Y},
        P_S
    \Big]
    \biggr\|_\infty
    \label{subadditivity}
\end{align}
where $U_{XY}'$ is the standard random oracle. We will now bound separately each term of the latter equation.

Let us pick some register $S_i^j$ and split up the invariant projector $P_S$ into a sum of two terms according to whether they contain $\Phi$ or $\Phi^\perp = \1 - \Phi$ on this register:
\begin{equation}
    P_S
    = \Phi_{S_i^j} \x 
      \tilde\Phi^0_{(S_i^j)^c}
    + \Phi^\perp_{S_i^j} \x
      \tilde\Phi^1_{(S_i^j)^c}
    \label{eq:PS split}
\end{equation}
where, for $b \in \set{0,1}$,
\begin{equation*}
    \tilde\Phi^b_{(S_i^j)^c}
    = \sum_{\substack{\alpha\in\widehat{ B^c}\\\alpha_i^j=b}}
    \bigotimes_{\substack{i'=1,\dotsc,l\\j'=0,1\\(i',j')\neq (i,j)}}
    \Phi(\alpha_{i'}^{j'})_{S_{i'}^{j'}}
\end{equation*}
is a sum of mutually orthogonal projectors and hence a projector itself.
Consequently,
\begin{equation}
    \norm[\big]{
        \tilde\Phi^b_{(S_i^j)^c}
    }_\infty
    = 1.
    \label{eq:Phib norm}
\end{equation}
Using \cref{lem:commutators} and substituting $U_{i,j}$ from \cref{eq:Uij for Lamport}
gives, 
\begin{align*}
    \biggl\|\bigg[
    \prod_{\substack{i=1,\dotsc,l\\j=0,1}}
    (U_{i,j})_{XYS^j_iP_i^j},
    P_S
    \bigg]
    \biggr\|_\infty
   &\leq 
    \sum_{\substack{i=1,\dotsc,l\\j=0,1}}
    \biggl\|\Big[
    (U_{i,j})_{XYS^j_iP_i^j},
    P_S
    \Big]
    \biggr\|_\infty\\
    &= 
    \sum_{\substack{i=1,\dotsc,l\\j=0,1}}
    \biggl\|\Big[
        P^=_{XS^j_i} \x \of[\Big]{
      (\CNOT\xp{n})_{P_i^j:Y}
      - \1
     }
    + \1
        ,P_S
    \Big]
    \biggr\|_\infty\\
    & \le
    \sum_{\substack{i=1,\dotsc,l\\j=0,1}}
    \norm*{\Big[P^=_{XS^j_i},P_S\Big]}_\infty
    \underbrace{\norm*{(\CNOT\xp{n})_{P_i^j:Y}
      - \1}_\infty}_{=2}\\
    &= 2
    \sum_{\substack{i=1,\dotsc,l\\j=0,1}}
    \norm*{\Big[P^=_{XS^j_i},P_S\Big]}_\infty
\end{align*}
where the %
last inequality follows from the sub-multiplicative property of norm and the fact that $[\1,P_S] = 0$. Recall that here $j=0$ and $\pk_i$ is the $i$-th public key prepared in the computational basis.
Substituting $P_S$ from \cref{eq:PS split} gives
\begin{align*}
    \sum_{\substack{i=1,\dotsc,l\\j=0,1}}
    \norm*{\Big[P^=_{XS^j_i},P_S\Big]}_\infty
    &= \sum_{\substack{i=1,\dotsc,l\\j=0,1}}
    \norm*{\Big[
    P^=_{XS^j_i},
    \Phi_{S_i^j} \x 
      \tilde\Phi^0_{(S_i^j)^c}
    + \Phi^\perp_{S_i^j} \x
      \tilde\Phi^1_{(S_i^j)^c}
    \Big]}_\infty\\
    &\le
    \sum_{\substack{i=1,\dotsc,l\\j=0,1}}
    \norm*{\Big[
        P^=_{XS^j_i},
        \Phi_{S_i^j}
    \Big]}_\infty
    \norm*{
        \tilde\Phi^0_{(S_i^j)^c}
    }_\infty +
    \norm*{\Big[
        P^=_{XS^j_i},
        \1 - \Phi_{S_i^j}
    \Big]}_\infty
    \norm*{
        \tilde\Phi^1_{(S_i^j)^c}
    }_\infty\\
    &= 2 \sum_{\substack{i=1,\dotsc,l\\j=0,1}}
    \norm*{\Big[
        P^=_{XS^j_i},
        \Phi_{S_i^j}
    \Big]}_\infty,
\end{align*}
where we used the triangle inequality, sub-multiplicativity of the norm, and \cref{eq:Phib norm}. Recall from \cref{lem:PQ} that
$\bigl\|\big[
    P^=_{XS^j_i},
    \Phi_{S^j_i}
\big]\bigr\|_\infty
= 2 \cdot 2^{-n/2}$.
Hence,
\begin{equation}
    \sum_{\substack{i=1,\dotsc,l\\j=0,1}}
    \norm*{\Big[
        P^=_{XS^j_i},
        P_S
    \Big]}_\infty
    \leq
    2 \sum_{\substack{i=1,\dotsc,l\\j=0,1}}
    \norm*{\Big[
        P^=_{XS^j_i},
        \Phi_{S_i^j}
    \Big]}_\infty
    = 2 \cdot 2l \cdot 2\cdot 2^{-n/2}
    = 8l \cdot 2^{-n/2},
    \label{eq:commutator with PS}
\end{equation}
and the first term in \cref{subadditivity} can be upper bounded as
\begin{align*}    
    \biggl\|\bigg[
        \prod_{\substack{i=1,\dotsc,l\\j=0,1}}
        (U_{i,j})_{XYS^j_iP^j_i},
        P_S
    \bigg]\biggr\|_\infty
    \le 2
    \sum_{\substack{i=1,\dotsc,l\\j=0,1}}
    \norm*{\Big[P^=_{XS^j_i},P_S\Big]}_\infty
    \leq 16l \cdot 2^{-n/2}.
\end{align*}

Next, we bound the second term of \cref{subadditivity}.
Substituting
\begin{equation}
    P^{\neq}_{XS} =
    \prod_{i=1}^l
    \prod_{j=0}^1
    P^{\neq}_{XS_i^j}
\end{equation}
from \cref{eq:Pneq}, we get
\begin{align*}
    \norm*{
        \big[P^{\neq}_{XS}U'_{XY},P_S\big] 
    }_\infty
    &=
    \norm[\bigg]{
    \bigg[
        \prod_{\substack{i=1,\dotsc,l\\j=0,1}}
        P^{\neq}_{XS^j_i}U'_{XY},P_S
    \bigg]
    }_\infty\\
    &\leq
    \sum_{\substack{i=1,\dotsc,l \\ j=0,1}}
    \norm[\Big]{
        \big[P^{\neq}_{XS^j_i},P_S\big] 
    }_\infty
    +
    \norm[\Big]{
        \underbrace{\big[U'_{XY},P_S\big]}_{=0} 
    }_\infty\\
    &=
    \sum_{\substack{i=1,\dotsc,l\\j=0,1}}
    \norm[\Big]{
    \Big[
        \1_{XS^j_i}-P^=_{XS^j_i},
        P_S
    \Big]
    }_\infty\\
    &=
    \sum_{\substack{i=1,\dotsc,l \\ j=0,1}}
    \norm[\Big]{
    \big[
        P^=_{XS^j_i},
        P_S
    \big] 
    }_\infty\\
	&\leq 8l\cdot 2^{-n/2}
\end{align*}
where we used the triangle inequality, the fact that $\big[U'_{XY},P_S\big] = 0$ because $U'_{XY}$ and  $P_S$ act on different registers, and finally \cref{eq:commutator with PS} to obtain the last inequality.

Lastly, we bound the third term of \cref{subadditivity} as follows:
\begin{align*}
    &\biggl\| 
    \bigg[
    \of[\bigg]{
        \1_{XS} -
        \prod_{\substack{i=1,\dotsc,l\\j=0,1}}
        P^{\neq}_{XS^j_i}
    } \otimes \1_{Y},
    P_S
    \bigg]
    \biggr\|_\infty\\
    &=
    \biggl\| 
        \underbrace{\big[\1_{XS},P_S\big]}_{=0}
        \otimes \, \1_{Y} - 
        \bigg[
            \prod_{\substack{i=1,\dotsc,l\\j=0,1}}
            P^{\neq}_{XS^j_i},
            P_S
        \bigg]
        \otimes \1_{Y}
    \biggr\|_\infty\\
    &\leq
    \biggl\| 
        \bigg[
            \prod_{\substack{i=1,\dotsc,l\\j=0,1}}
            P^{\neq}_{XS^j_i},
            P_S
        \bigg]
    \biggr\|_\infty\\
    &\leq
    \sum_{\substack{i=1,\dotsc,l \\ j=0,1}}
    \Bigl\|
        \Big[\1_{XS^j_i}-P^=_{XS^j_i},P_S\Big]
    \Bigr\|_\infty\\
    &\leq 8l\cdot 2^{-n/2},
\end{align*}
where the reasoning is similar to what we used for bounding the second term.

Replacing the three terms of \cref{subadditivity} by their respective bounds gives
\begin{equation*}
  \bigl\|[U_h,P_S]\bigr\|_\infty
    \leq
    16l\cdot 2^{-n/2} +
    8l\cdot 2^{-n/2} +
    8l\cdot 2^{-n/2}
    =
    32l\cdot 2^{-n/2},
\end{equation*}
as claimed.
\end{proof}

\section{Lemmas used for proving the security of the Winternitz OTS}\label{apx:Winternitz lemmas}

In this appendix, we prove \cref{A'stateInv:winternitz,lem:orth-Ql+1-Pga-Wots,ComPgamUh:winternitz} which are the main technical ingredients of \cref{thm:winternitz} on security of the Winternitz OTS (another ingredient, \cref{lem:Commutator}, was already proved in \cref{apx:L8andL15}).
The proofs of these lemmas are similar to those for the Lamport OTS in \cref{apx:Lamport lemmas}.

\subsection{Proof of \cref{A'stateInv:winternitz}}

The proof of this lemma is similar to the proof of \cref{final state invariant if no hash queries} in \cref{apx:PSB}.

\LemNoHashWin*

\begin{proof}
If there are no hash queries before the $\Sign$ query, all hash chain registers (except for the last) are in state $\ket{\Phi}$, i.e.
\[\ket{\psi_0}_{M\Sigma EB\Gamma}
-\Phi_\Gamma^{\otimes l(w-1)}\ket{\psi_0}_{M\Sigma EB\Gamma}
=0\]
where $\ket{\psi_0}_{M\Sigma EB\Gamma}$ is adversary's state immediately before the $\Sign$ query.
Our goal is to show that applying the invariant projector $P_\Gamma$ onto the post-signature state $\ket{\psi_1}_{M\Sigma EB\Gamma} = B
\Sign_{\sk}
\ket{\psi_0}_{M\Sigma EB\Gamma}$ leaves it unchanged, i.e.
\[\bigl\|\ket{\psi_1}_{M\Sigma EB\Gamma}
-P_\Gamma\ket{\psi_1}_{M\Sigma EB\Gamma}\bigr\|_2
= 0.\]
We have:
\begin{align}
P_\Gamma\ket{\psi_1}_{M\Sigma EB\Gamma}
&=P_\Gamma B\Sign_{\sk}
\ket{\psi_0}_{M\Sigma EB\Gamma}\nonumber\\
&=P_\Gamma B\Sign_{\sk}
\ket{\psi_0^1}_{M\Sigma EB\Gamma} + 
P_\Gamma B\Sign_{\sk}
\ket{\psi_0^0}_{M\Sigma EB\Gamma}.
\label{PGamma applied to post signature state}
\end{align}
Recall from \cref{eq:psi0 and psi1} that $\ket{\psi_0^0}$ and $\ket{\psi_0^1}$ refer respectively to the blinded message ($m\in B$) term and the un-blinded message ($m\in B^c$) term of the queried message.
Note that using \cref{eq:psi0-W}, their respective expressions is given by:
\begin{align}
  \ket{\psi _0^1}_{M\Sigma EB\Gamma} 
   &=\sum_{m\in B}
    \sum_{\sigma\in(\{0,1\}^n)^l}
    \kappa_{m\sigma 1} 
    \ket{m}_M
    \ket{\sigma}_{\Sigma}
    \ket{\alpha_{m\sigma 1}}_E
    \ket{1}_B
    \ket \nu _\Gamma\label{eq:psi01}
    \\
    \ket{\psi _0^0}_{M\Sigma EB\Gamma} 
   &=\sum_{m\in B^c}
    \sum_{\sigma\in(\{0,1\}^n)^l}
    \kappa_{m\sigma 0} 
    \ket{m}_M
    \ket{\sigma}_{\Sigma}
    \ket{\alpha_{m\sigma 0}}_E
    \ket{0}_B
    \ket \nu _\Gamma \label{eq:psi00}
\end{align}
Let us analyze each term of the right hand-side of \cref{PGamma applied to post signature state} separately.
Given that the first term corresponds to blinded messages, there is no signature, so all the hash chain registers (except for the last) are in range of the projector $\Phi$. Thus, the state $\ket{\psi_0^0}_{M\Sigma EB\Gamma}$ remains unchanged, that is
\begin{align}
    P_\Gamma
    B\Sign_{\sk}
    \ket{\psi_0^1}_{M\Sigma EB\Gamma}
    =\ket{\psi_0^1}_{M\Sigma EB\Gamma}.
    \label{Pgamma leave lambda10 invariant}
\end{align}
For the second term, we know by the checksum that if a message $m \in B^c$ with corresponding block $b(m)$ has been queried for any other message $m'\neq m$, the corresponding block $b(m')$ contains at least one index $i$ such that $b_i' < b_i$ with $1 \leq i \leq l$ where $b'_i$ is the $i$-th block of $b(m')$.

Recall from \cref{eq:b(m)} that $b(m)$ is defined by $b=b(m) = (b_1,\dotsc,b_l) = m \parallel C(m)$, where $C(m)$ is the checksum corresponding to the message $m$.

Since $P_\Gamma$ acts on the whole hash chain register but we are interested only in those sub-registers that are consistent with $\widehat{B^c}$, we can split $P_\Gamma$ into a sum of two orthogonal projectors, just like we did in \cref{eq:PS split}.
Letting
\begin{equation}
    A(m)
    = \set{\alpha \in \widehat{B^c} : \text{$\alpha_i^{j} = 0$ for all $(i,j)$ with $j<b_i(m)$}}
    \subseteq \widehat{B^c},
    \label{A(m)}
\end{equation}
we can split $P_\Gamma$ as follows:
\begin{align*}
    P_\Gamma 
    = \sum_{\alpha \in \widehat{B^c}} \Phi(\alpha)_\Gamma 
 &= \sum_{\alpha \in A(m)}
    \Phi(\alpha)_\Gamma +
    \sum_{\alpha \in \widehat{B^c} \setminus A(m)}
    \Phi(\alpha)_\Gamma \\
 &= P_\Gamma^1(m) +
    P_\Gamma^2(m).
\end{align*}
The first projector $P_\Gamma^1(m)$ simply collects all terms that correspond to those hash chain registers for which $\alpha_i^j = 0$. More precisely, $P_\Gamma^1(m)$ corresponds to the hash chain registers $\Gamma_i^j$ with $j < b_i$ in range of $\Phi$ while $P_\Gamma^2(m)$  corresponds to the remaining hash chain registers.
Note that the way $P_\Gamma$ is split differs from one message to another.

Using this decomposition of $P_\Gamma$, the second term of the right-hand side of \cref{PGamma applied to post signature state} can be expressed as 
\begin{align}
    P_\Gamma
    B\Sign_{\sk}
    \ket{\psi_0^0}_{M\Sigma EB\Gamma}
 &= \sum_{m\in B^c}
    \left( P_\Gamma^1(m) +  P_\Gamma^2(m)\right)
    B\Sign_{\sk}
    \ket{\gamma(m)}_{ \Sigma EB}
    \ket \nu _\Gamma
    \label{Pgamma applied to lambda11}
\end{align}
where
\begin{align*}
    \ket{\gamma(m)}_{ \Sigma EB}
 = \sum_{\sigma\in(\{0,1\}^n)^l}
    \kappa_{m\sigma 0}
    \ket m_M
   \ket{\sigma}_{\Sigma}
   \ket{\alpha_{m\sigma 0}}_E
   \ket{0}_B
\end{align*}
comes from the definition of $\ket{\psi_0^0}_{M\Sigma EB\Gamma}$ in \cref{eq:psi00} and $\ket\nu_\Gamma$ is the initial hash chain block (before any query) as defined in \cref{initial:hash-chain}.
For a fixed message $m$, $P_\Gamma^1(m)$ can also be written as 
\[P_\Gamma^1(m)
= \sum_{\alpha \in A(m)} \Phi(\alpha)_\Gamma
=\Phi_{\Gamma^{<b_1(m)}_1\cdots\Gamma^{<b_l(m)}_l}\otimes 
\1_{\Gamma_{1}^{\ge b_1(m)}\cdots\Gamma_l^{\ge b_l(m)}},\]
where $\Gamma_j^{<j}=\Gamma_i^0\Gamma_i^1\cdots\Gamma_i^{j-1}$ and $\Gamma_i^{\ge j}$ is defined analogously.

Thus,
\begin{align}
   & 
    P_\Gamma^1(m)
    B\Sign_{\sk}
    \ket{\gamma(m)}_{ \Sigma EB}
    \ket\nu_\Gamma\nonumber\\
  &=
    B
    \Sign_{\sk}
    \ket{\gamma(m)}_{ \Sigma EB}
    \of*{
    \Phi_{\Gamma^{<b_1(m)}_1\cdots\Gamma^{<b_l(m)}_l}\otimes 
\1_{\Gamma_{1}^{\ge b_1(m)}\cdots\Gamma_l^{\ge b_l(m)}}
    }
    \ket \nu _\Gamma\nonumber\\*
  &=
    B
    \Sign_{\sk}
    \ket{\gamma(m)}_{ \Sigma EB}
    \ket \nu _\Gamma
   \label{Pgamma leave state of fixed-message-invariant}
\end{align}
where the second equation follows because $P_\Gamma^1(m)$ and $B\Sign_{\sk}$ act on different registers and therefore commute. More precisely, $P_\Gamma^1(m)$ acts on the hash chain sub-registers of messages $m'$ whose blocks fulfil the condition $b_i(m') < b_i(m)$ while $B\Sign_{\sk}$ acts on hash chain registers corresponding to the message $m$ that was signed, i.e.
\[[P_\Gamma^1(m),B
\Sign_{\sk}]
\ket m_M = 0.\]
Note that this is only true for a fixed message $m$. 
The last equation follows from the fact that $\Gamma$ (except for the last sub-chain) is in state $\ket{\Phi}$.

Since $P_\Gamma^1$ and $P_\Gamma^2$ are sums of mutually orthogonal projectors,
\[P_\Gamma^2(m)P_\Gamma^1(m) = 0.\]
Therefore, by \cref{Pgamma leave state of fixed-message-invariant},
\[P_\Gamma^2(m)
    B\Sign_{\sk}
    \ket{\gamma(m)}_{ \Sigma EB}
    \ket\nu_\Gamma=0.\]
Hence, from the latter equation and \cref{Pgamma leave lambda10 invariant} it follows that
\begin{align*}
   P_\Gamma
    B\Sign_{\sk}
    \ket{\gamma(m)}_{ \Sigma EB}
    \ket\nu_\Gamma
   =
    B\Sign_{\sk}
    \ket{\gamma(m)}_{ \Sigma EB}
    \ket\nu_\Gamma.
\end{align*}
Hence, by \cref{Pgamma applied to lambda11},
\begin{align}
    \bigl\|\ket{\psi_1}_{M\Sigma EB\Gamma} -
    P_\Gamma\ket{\psi_1}_{M
    \Sigma EB\Gamma}\bigr\|_2
    = 0.
\label{Pgamma leave lambda1 invariant}
\end{align}
We conclude that the projection of $P_\Gamma$ onto the post-signature state leaves it unchanged considering that hash queries before $\Sign$ query leave the hash chain registers (except for the last) in state $\ket{\Phi}$.
\end{proof}

\subsection{Proof of \cref{lem:orth-Ql+1-Pga-Wots}}

The proof of this lemma is similar to the proof of \cref{Orthogonality argument} in \cref{apx:QPS}.

\LemOrthWots*

\begin{proof}
Recall from \cref{eq:Qq} that
\begin{align}
    Q_{l+1}^{b^*}
    &= \bigotimes_{i=1}^l
    \Phi^\perp_{\Gamma_i^{b^*_i}}
\end{align}
where
$\Phi^\perp = \1 - \Phi$ and
$\Phi = \proj{\Phi}$.
Also, recall from \cref{def of Be complement} that
\begin{align}
    \widehat{B^c} = \bigcup_{m \in B^c}
    \set[\Big]{
        \alpha \in \set{0,1}^{l(w-1)}
        \mathrel{\Big|}
        \text{$\alpha_i^j = 0$ for all $i=1,\dotsc,l$ and $j < b_i(m)$}
    }
\end{align}
where $b_i(m)$ is the $i$-th block of $m$ concatenated with its checksum, see \cref{sec:WOTS}.
Moreover, recall from \cref{eq:PGamma,eq:Phi(a)-Wots} that
\begin{equation}
    P_\Gamma 
    = \sum_{\alpha \in \widehat{B^c}} \Phi(\alpha)_\Gamma
    = \sum_{\alpha \in \widehat{B^c}}
    \bigotimes_{i=1}^l
    \bigotimes_{j=0}^{w-2}
    \Phi(\alpha_i^j)_{\Gamma^j_i}.
\end{equation}
where
$\Phi(0) = \Phi$ and
$\Phi(1) = \Phi^\perp$.

Combining the expressions of $Q_{l+1}^{b^*}$ and $P_\Gamma$, we get
\begin{align}
    Q_{l+1}^{b^*} P_\Gamma
  = \sum_{\alpha\in\widehat{B^c}}
    \of[\Big]{
        \of[\big]{(\Phi^\perp)\xp{l}}_
        {\Gamma_1^{b_1^*}\cdots
        \Gamma_l^{b_l^*}}\otimes
        \1_{\Gamma_{(b^*_i,i)^c}}^{\otimes l(w-2)}
    }
    \Phi(\alpha).
    \label{eq:QP}
\end{align}

Recall from \cref{sec:WOTS} that by construction of the checksum, if the block $b$ of a message $m$ is computed, the block $b'$ of any other message $m'$ contains at least an index $i$ such that $b_i' < b_i$, $i=1,\dotsc,l$.
Consider a single term $\alpha \in \widehat{B^c}$ in the sum of \cref{eq:QP} and an un-blinded message $m\in B^c$ with associated block $b=b(m)$. 
By definition of $\widehat{B^c}$, there exists a message $m'\in B^c$ such that $\alpha_i^{b_i'}=0$ for all $i$. This implies by construction of the checksum that $b'_i<b_i$ for all $i$. But as $B\ni m^*\neq m\in B^c$, there exists an index $i$ such that $b_i^*< b_i$, or equivalently $b_i^*=b_i'$.
Therefore we can rewrite the corresponding term as follows:
\begin{align*}
&   \of[\Big]{
    \of[\Big]{
        \Phi^\perp_{\Gamma_1^{b_1^*}}
        \otimes\cdots\otimes
        \Phi^\perp_{\Gamma_l^{b_l^*}}
    }
    \otimes
    \1_{\Gamma_{({b^*_i},i)^c}}^{\otimes l(w-2)}
    }
    \Phi(\alpha)\\
&=  \of[\Big]{
        \Phi^\perp_{\Gamma_1^{b_1^*}}
        \otimes\cdots
        \Phi^\perp_{\Gamma_i^{b_i^*}}
        \otimes\cdots\otimes
        \Phi^\perp_{\Gamma_l^{b_l^*}}
        \otimes
        \1_{\Gamma_{({b^*_i},i)^c}}^{\otimes l(w-2)}
    }\\
    &\quad\,\,
    \of[\Big]{
        \Phi_{\Gamma_1^{b_1'}}
        \otimes\cdots\otimes
        \Phi_{\Gamma_i^{b_i'}}
        \otimes\cdots\otimes
        \Phi_{\Gamma_l^{b_l'}}
        \otimes
        \Phi(\alpha_1^{b_1})_{\Gamma_1^{b_1}}\otimes \cdots\otimes  \Phi(\alpha_l^{b_l})_{\Gamma_1^{b_l}}
        }
    \\
    &=
    \cdots\otimes
        \underbrace{
        (\Phi^\perp\Phi)_{\Gamma_i^{b_i^*}}
        }_{=0}\otimes\cdots
    \\
  &= 0.
\end{align*}
The result follows by applying this argument to each term in \cref{eq:QP}.
\end{proof}

\subsection{Proof of \cref{ComPgamUh:winternitz}}

The proof of this lemma is similar to the proof of \cref{commutator of ps and uh} in \cref{apx:UhPS}.

\LemComPWin*

\begin{proof}
Here, we want to prove that the commutator of the invariant projector $P_\Gamma$ and the random oracle unitary $U_h$ is small. We remind the reader that we use the convention that operators are tensored with an identity on any missing registers, which should be clear from context. We begin by deriving a decomposition for $P_\Gamma$. By definition of $\widehat{B^c}$, for some $i$ and $j$, if $\alpha\in\widehat{B^c}$ such that $\alpha_i^{j'}=0$ for $j'<j$ and $\alpha_i^j=1$, then $\tilde\alpha\in\widehat{B^c}$ for all $\tilde\alpha$ such that $\alpha_{i'}^{j'}=\tilde\alpha_{i'}^{j'}$ if $i'\neq i$ or $j'\le j$. It follows that for a fixed $i\in\{1,\dotsc,l\}$, we can write
\begin{align}
    P_\Gamma
    &=\sum_{j=0}^{w-1}
    \of[\big]{\Phi\xp{j}}_{\Gamma_i^{\le j-1}}
    \otimes \Phi^{\perp}_{\Gamma_i^j}
    \otimes P^{(i,j)}_{\Gamma_{i^c}} +
    \of[\big]{\Phi\xp{(w-1)}}_{\Gamma_i}
    \otimes P^{(i,w)}_{\Gamma_{i^c}}\\
    &=\sum_{j=0}^{w-1} \of[\Big]{
    \of[\big]{\Phi\xp{j}}_{\Gamma_i^{\le j-1}} -
    \of[\big]{\Phi\xp{(j+1)}}_{\Gamma_i^{\le j}}}
    \otimes P^{(i,j)}_{\Gamma_{i^c}}+
    \of[\big]{\Phi\xp{(w-1)}}_{\Gamma_i}
    \otimes P^{(i,w)}_{\Gamma_{i^c}}
\end{align}
with
$\Phi = \proj{\Phi}$ and
$\Phi^\perp = \1 - \Phi$ with $\ket \Phi$ defined in \cref{eq:Phi}, and
\begin{equation}
    P^{(i,j)}=\sum_{\substack{\alpha\in A(m)\\\alpha_i^j=1}}\bigotimes_{\substack{i'=1,\dotsc,l\\i'\neq i}}\bigotimes_{j'=0,\dotsc,w-2}\Phi(\alpha_{i'}^{j'})
\end{equation}
where we use the convention that $\alpha_i^w=1$ and $A(m)$ was defined in \cref{A(m)} in the proof of \cref{A'stateInv:winternitz}. We can further rearrange to bring $P_\Gamma$ into the form
\begin{equation}
    P_\Gamma=\sum_{j=0}^{w-1}
    \of[\big]{\Phi\xp{j}}_{\Gamma_i^{\le j-1}}\otimes (A_i^j)_{\Gamma_{i^c}},
\end{equation}
where $A_i^j$ is a difference of two projectors\footnote{For some $j$, one of these projectors is the zero projector.}.
We begin by bounding
\begin{align}
   \bigl\|[U_h, P_\Gamma]\bigr\|_\infty
   \leq
    \sum_{\substack{i=1,\dotsc,l \\ j=0,\dotsc,w-2}}
    \biggl\|\big[
    (U_i^j)_{\Gamma ^j_i}
    ,P_\Gamma\big] 
    \biggr\|_\infty +
    2\sum_{\substack{i=1,\dotsc,l \\ j=0,\dotsc,w-2}}
    \biggl\|\big[
    P^{\neq}_{X\Gamma ^j_i},P_\Gamma\big] 
    \biggr\|_\infty\label{commutator uh and Pgamma}
\end{align}
using the same steps as in eqs.~\eqref{eq:firstofmany} to \eqref{subadditivity}.
We can bound the first term as
\begin{align*}
    \sum_{\substack{i=1,\dotsc,l \\ j=0,\dotsc,w-2}}
    \biggl\|\big[
    (U_i^j)_{\Gamma_i XY}
    ,P_\Gamma\big] 
    \biggr\|_\infty
    \le 2\sum_{j=0}^{w-1}\sum_{\substack{i=1,\dotsc,l \\ j'=0,\dotsc,w-2}}
    \biggl\|\big[
    (U_i^j)_{\Gamma_i XY},\of[\big]{\Phi\xp{j'}}_{\Gamma_i^{\le j'-1}}\big] 
    \biggr\|_\infty,
\end{align*}
where the inequality follows from the decomposition of $P_\Gamma$ above, the triangle inequality, and the fact that $\|A_i^j\|_\infty\le 2$.
By \cref{eq:term1} in the proof of \cref{lem:Commutator},
\begin{equation}
    \biggl\|\big[
    (U_i^j)_{\Gamma_i XY},\of[\big]{\Phi\xp{j'}}_{\Gamma_i^{\le j'-1}}\big]
    \biggr\|_\infty\le 4 (w-1)2^{-n/2}.
\end{equation}

For the second term in \cref{commutator uh and Pgamma}, we simplify it to
\begin{align}
    2\sum_{\substack{i=1,\dotsc,l \\ j=0,\dotsc,w-2}}
    \biggl\|\big[
    P^{\neq}_{X\Gamma ^j_i},P_\Gamma\big] 
    \biggr\|_\infty 
 &= 2\sum_{\substack{i=1,\dotsc,l \\ j=0,\dotsc,w-2}}
    \biggl\|\big[\1_{X\Gamma}-
    P^=_{X\Gamma ^j_i},P_\Gamma\big] 
    \biggr\|_\infty \nonumber\\   
 &= 
    2\sum_{\substack{i=1,\dotsc,l \\ j=0,\dotsc,w-2}}
    \biggl\|
    \big[P^=_{X\Gamma ^j_i},P_\Gamma\big] 
    \biggr\|_\infty.
\end{align}
We can now alternatively decompose $P_\Gamma$ similar to \cref{eq:decomposition}, i.e.
\begin{equation}
    P_\Gamma =
    \Phi_{\Gamma_i^j} \otimes
    \tilde\Phi^0_{\Gamma_{(i,j)^c}} +
    \Phi^\perp_{\Gamma_i^j} \otimes
    \tilde\Phi^1_{\Gamma_{(i,j)^c}}
\end{equation}
for some projectors $\tilde \Phi^b$, $b=0,1$. 
Using this decomposition, we bound
\begin{align*}
    2\sum_{\substack{i=1,\dotsc,l \\ j=0,\dotsc,w-2}}
    \biggl\|\big[
    P^{=}_{X\Gamma ^j_i},P_\Gamma\big] 
    \biggr\|_\infty 
  \leq  \frac{8l(w-1)}{2^{n/2}}   
\end{align*}
using the same calculations we performed in the proof of \cref{commutator of ps and uh}.  
By replacing the two terms of \cref{commutator uh and Pgamma} by their respective bounds, we get
\begin{align}
    \bigl\|[U_h, P_\Gamma]\bigr\|_\infty
  &=  \sum_{\substack{i=1,\dotsc,l \\ j=0,\dotsc,w-2}}
     \biggl\|\big[
    (U_i^j)_{\Gamma ^j_i}
    ,P_\Gamma\big] 
    \biggr\|_\infty +
     2\sum_{\substack{i=1,\dotsc,l \\ j=0,\dotsc,w-2}}
    \biggl\|\big[
    P^{\neq}_{X\Gamma ^j_i},P_\Gamma\big] 
    \biggr\|_\infty\nonumber\\
 &\leq
    \frac{8lw(w-1)}{2^{n/2}} + 
    \frac{8l(w-1)}{2^{n/2}}\nonumber\\
  &= \frac{8l(w+1)(w-1)}{2^{n/2}},
  \label{final bound of Pgamma and Uh}
\end{align}
as desired.
\end{proof}

\end{document}